%% file: sample2e.tex
\documentclass[jair,twoside,11pt,theapa]{article}
\usepackage{jair, theapa, rawfonts}
\usepackage{amsmath}
\usepackage{amsthm}
\usepackage{amssymb}
\usepackage[inline]{enumitem}
\usepackage{xcolor}
\usepackage{graphicx}
\usepackage{booktabs}
\usepackage{multirow} 
\usepackage{url}
\usepackage{algorithm}
\usepackage{algpseudocode}

\input{mathinput}

\jairheading{1}{1993}{1-15}{6/91}{9/91}
\ShortHeadings{Designing Efficient Mechanisms for the $m$-CFLP}
{Auricchio, Zhang, \& Zhang}

\begin{document}

\title{From Optimal Transport to Efficient Mechanisms for the $m$-Capacitated Facilities Location Problem with Bayesian Agents\thanks{A preliminary version of this paper appeared in the proceedings of the International Conference on Autonomous Agents and Multiagent Systems \cite{auricchio2023extended}. }}

\author{\name Gennaro Auricchio \email gennaro.auricchio@unipd.it \\ \addr Università degli Studi di Padova, Via Trieste 63, Padua, 35131, Italy
\AND
       \name Jie Zhang \email jz2558@bath.ac.uk \\
       \addr University of Bath, Claverton Down, Bath, BA2 7AY, United Kingdom
       \AND
       \name Mengxiao Zhang \email mengxiao.zhang@durham.ac.uk \\
       \addr Durham University, Stockton Road, Durham, DH1 3LE, United Kingdom
       }


\maketitle

\begin{abstract}
In this paper, we study of the $m$-Capacitated Facility Location Problem ($m$-CFLP) on the line from a Bayesian Mechanism Design perspective and propose a novel class of mechanisms: the \textit{Extended Ranking Mechanisms} (ERMs).
We first show that an ERM is truthful if and only if it satisfies a system of inequalities that depends on the capacities of the facilities we need to place.
We then establish a connection between the $m$-CFLP and a norm minimization problem in the Wasserstein space, which enables us to show that if the number of agent goes to infinity the limit of the ratio between the expected Social Cost of an ERM and the expected optimal Social Cost is finite and characterize its value.
Noticeably, our method generalizes to encompass other truthful mechanisms and other metrics, such as the $l_p$ and Maximum costs.
We conclude our theoretical analysis by characterizing the optimal ERM tailored to a $m$-CFLP and a distribution $\mu$, that is the ERM whose limit Bayesian approximation ratio is the lowest compared to all other feasible ERMs.
We consider mainly two frameworks: \begin{enumerate*}[label=(\roman*)] 
\item in the first framework, the total facility capacity matches the number of agents,
\item in the second framework, $m=2$. When we consider the Maximum Cost, we retrieve the optimal ERM for every $\mu$, while for the Social Cost, we characterize the solution when the measure $\mu$ is symmetric.
\end{enumerate*}
Lastly, we numerically compare the performance of the ERMs against other truthful mechanisms and evaluate how quickly the Bayesian approximation ratio converges to its limit.
%
\end{abstract}

\section{Introduction}
\label{Introduction}

Mechanism Design seeks to establish protocols for aggregating the private information of a group of strategic agents to optimize a global social objective. 
Nonetheless, optimizing a communal goal solely based on reported preferences often leads to undesirable manipulation driven by the agents' self-interested behaviour. 
Hence, a key property that a mechanism should possess is \textit{truthfulness}, which ensures that no agent can gain an advantage by reporting false information. 
Unfortunately, this stringent condition is often incompatible with the optimization the communal objective, hence implementing a truthful mechanism necessarily leads to suboptimal outcomes.
To quantify the efficiency loss entailed by a truthful mechanism, Nisan and Ronen introduced the concept of \textit{worst-case approximation ratio}, which represents the maximum ratio between the objective achieved by the truthful mechanism and the optimal objective attainable across all possible agents' reports \cite{nisan1999algorithmic}.
Therefore, the lower the worst-case approximation ratio of a mechanism is, the better.
Ever since this pioneering work, Mechanism Design has been used to approach many key problems in Social Choice Theory, such as one-sided matching \cite{10.1257/000282802762024728,manlove2013algorithmics,auricchio2023manipulability}, two-sided matching \cite{cho2022two,balinski1999tale,kojima2019new}, auction design \cite{archer2004approximate,chu2008truthful,dobzinski2006truthful}, and scheduling \cite{lavi2007truthful,christodoulou2023proof}.

A prominent problem in Mechanism Design is the $m$-Facility Location Problem ($m$-FLP) \cite{brandt2016handbook,chan2021mechanism}.
In its fundamental guise, the $m$-FLP involves locating $m$ facilities amidst $n$ self-interested agents.
Each agent requires access to a facility, making it preferable for them to have a facility located as close as possible to their position. 
Furthermore, each facility can accommodate any number of agents, thus agents can select which facility to use devoid of any concern about possible overloads.
The study of $m$-FLP from an algorithmic mechanism design viewpoint was initiated by Procaccia and Tennenholtz in \cite{procaccia2013approximate}.
When $m=1,2$ there are several truthful mechanisms, such as the median mechanism \cite{black1948rationale} and its generalizations \cite{barbera2001strategy,moulin1980strategy}, that achieve small constant worst-case approximation ratios.
When $m> 2$, however, these efficiency guarantees are much more negative.
Indeed, it is well-known that no truthful mechanism has a bounded approximation ratio while being deterministic and anonymous  \cite{fotakis2014power}.
It is worthy of notice however, that this impossibility result does not apply to randomized mechanisms \cite{fotakis2013strategyproof} and to instances where the number of agents is equal to the number of facilities plus one \cite{escoffier2011strategy}.
A natural extension of the $m$-FLP is the $m$-Capacitated Facility Location Problem ($m$-CFLP), in which every facility has a capacity limit.
The capacity limit constraints the amount of agents that the facility can serve.
In this case, the solution does not only elicit the positions of the facilities, but it also specifies to which facility every agent is assigned to, ensuring that no facility is overloaded.
To the best of our knowledge, there are very few works that analyzed the Mechanism Design aspects of the $m$-CFLP.
Moreover, all the existing results are conducted in the classic worst-case analysis, where the designer has no information on the agents and therefore aims to define mechanisms that work well on every possible input, regardless of the likelihood of the input. 
This type of analysis is too pessimistic and gives little insight into how to select a mechanism tailored to a specific task.
For example, there are currently no mechanisms capable of locating more than two capacitated facilities while achieving a finite approximation ratio.
Indeed, it seems that most of the negative results known for the $m$-FLP hold for the $m$-CFLP as well.
Furthermore, even if we restrict our attention to the case where $m=2$, the approximation ratio of all the known truthful mechanisms depends linearly on the number of agents \cite{aziz2020facility}, making these efficiency guarantees less meaningful as the scale of the problem increases.
In this paper, we show that these issues do not arise when we study the $m$-CFLP from a Bayesian Mechanism Design perspective.
In Bayesian Mechanism Design, every agent's location is a random variable whose law is known to the designer \cite{hartline2010bayesian} thus the scope of the mechanism is not to minimize the cost in the worst possible case but rather to define truthful routines that work well in expectation.
A previous version of this paper appeared in the proceedings of the International Conference on Autonomous Agents and Multiagent Systems \cite{auricchio2023extended}. 
In this improved version, we extend our results to the Maximum Cost and the $l_p$ costs and generalize the techniques to retrieve the optimal ERM to a wider class of problems.
In particular, we provide a sufficient condition that characterizes the optimal ERM with respect to the Social Cost for any measure and retrieve an algorithm to compute the optimal ERM for any probability measure with respect to the Maximum Cost.
Moreover, we run a comprehensive set of experiments that evaluates the performances of the mechanisms with respect to the Maximum and $l_p$ costs and their convergence rates.
%

\subsection{Our Contribution and Structure of the Paper}

Our contribution is as follows.
%
In Section~\ref{sec:ERM}, we define and study the Extended Ranking Mechanisms (ERMs), a generalization of the class of Ranking Mechanisms introduced in \cite{aziz2020facility}.
Our extension allows to tackle a broader framework in which the combined capacity of the facilities surpasses the total amount of agents. 
Moreover, a Ranking Mechanism is truthful if and only if it places all the facilities at one or two different locations, while a ERM does not necessarily require such constraint to ensure truthfulness.
In Section \ref{sec:analysis}, we study the ERMs in a Bayesian framework, thus every agent's position is represented by a set of i.i.d. random variables.
Our investigation reveals that $m$-CFLP is equivalent to a norm minimization problem over a subset of the Wasserstein space.
By leveraging the properties of the Wasserstein distance, we then establish the convergence of the Bayesian approximation ratio, i.e., the ratio between the expected cost incurred by the mechanism and the expected optimal cost, as the number of agents tends toward infinity. 
We show that the limit Bayesian approximation ratio depends on the probability distribution $\mu$ describing the agents, the vector $\vec q$ determining the capacities of the facilities, and the characteristics of the mechanism.
We then characterize the set of condition under which our results extended to the Maximum and $l_p$ costs.
Finally, we show that our techniques allow to study the Bayesian approximation ratio of previously introduced mechanisms.
In Section \ref{sec:para}, we identify the optimal ERM tailored to the $m$-CFLP whose agents are distributed according to $\mu$ given the capacities determined by $\vec q$.
First, we show that an optimal ERM exists and characterize it as a solution to a minimization problem.
We then narrow our focus on two scenarios: \textit{no-spare capacity} framework in which the total capacity of the facilities matches the total number of agents, and the $2$-CFLP.
In particular, we fully characterize the optimal ERM for a uniformly distributed population.
Finally, in Section \ref{sec:experiment}, we validate our findings through extensive numerical experiments. 
In particular, we compare the performances of the ERMs with the performances of other truthful mechanisms, such as the InnerGap Mechanism \cite{ijcai2022p75} and the Extended Endpoint Mechanism \cite{aziz2020facility}.
From our experiments, we observe that a well-tuned ERM outperforms all the other mechanisms whenever $n\ge 20$.
Furthermore, the limit Bayesian approximation ratio is a reliable estimation of the Bayesian approximation ratio of the ERM when the number of agents is greater than $20$.

\subsection{Related Work}

The $m$-Facility Location Problem ($m$-FLP) and its variations are significant issues in many practical domains, such as disaster relief \cite{doi:10.1080/13675560701561789}, supply chain management \cite{MELO2009401}, healthcare \cite{ahmadi2017survey}, clustering \cite{hastie2009elements}, and public facilities accessibility \cite{barda1990multicriteria}.
Procaccia and Tennenholtz initially studied the Mechanism Design study of the $m$-FLP, laying the groundwork for this field in their pioneering work \cite{procaccia2013approximate}.
Following that, mechanisms with constant approximation ratios to place one or two facilities on trees \cite{DBLP:conf/sigecom/FeldmanW13}, circles \cite{DBLP:conf/wine/LuWZ09}, general graphs \cite{DBLP:conf/sigecom/DokowFMN12}, and metric spaces \cite{anshelevich2017randomized,DBLP:conf/sigecom/TangYZ20} were introduced.
Notice however that all these positive results are confined to the case in which we have at most two facilities to place and/or the number of agents is limited.
Moreover, different works generalize the initial framework proposed in \cite{procaccia2013approximate}, by considering different agents' preferences \cite{peters2020preferences,MEI201846}, different costs \cite{feigenbaum2017approximately}, and additional constraints \cite{DBLP:conf/aaai/Filos-RatsikasL15,feldman2016voting,xu2021two}.
In this paper, we analyze the $m$-Capacitated Facility Location Problem ($m$-CFLP), a variant of the $m$-FLP in which each facility can accommodate a finite number of agents.
The Mechanism Design aspects of the $m$-CFLP have only recently begun to attract attention.
Indeed, the game theoretical framework for the $m$-CFLP that we consider was introduced in \cite{aziz2020facility}.
This work defined and studied various truthful mechanisms, like the InnerPoint Mechanism, the Extended Endpoint Mechanism, and the Ranking Mechanisms.
A more theoretical study of the problem was then presented in \cite{ijcai2022p75}, demonstrating that no mechanism can position more than two capacitated facilities while adhering to truthfulness, anonymity, and Pareto optimality.
Additionally, another paper dealing with Mechanism Design aspects of the $m$-CFLP is \cite{aziz2020capacity}.
However, it explores a different framework where only one facility needs to be placed and is unable to serve all agents.
To the best of our knowledge, all the results on the $m$-CFLP concerned the classic Mechanism Design framework that evaluates the performances of the mechanism based on worst-case analysis. 
In this paper, we consider an alternative approach: the Bayesian Mechanism Design perspective.
Unlike traditional Mechanism Design, where the designer lacks information about agent types, in the Bayesian Mechanism Design framework each agent's type follows a known probability distribution \cite{hartline2013bayesian,chawla2014bayesian}.
In this framework, we are able to determine a probability distribution over the set of all the possible inputs of the mechanism, enabling us to consider the expected cost of a mechanism.
Bayesian Mechanism Design has been applied to investigate routing games \cite{gairing2005selfish}, facility location problems \cite{zampetakis2023bayesian}, combinatorial mechanisms using $\epsilon$-greedy mechanisms \cite{lucier2010price}, and auction mechanism design \cite{chawla2007algorithmic,yan2011mechanism}.
%

%
Over the past few decades, Optimal Transport (OT) has found application in Theoretical Computer Science related fields, such as Computer Vision \cite{Rubner2000,Pele2009}, Computational Statistics \cite{Levina2001}, Clustering \cite{auricchio2019computing}, and Machine Learning \cite{scagliotti2023normalizing,frogner2015learning,scagliotti23,Cuturi2014}.
However, the only field related to mechanism design that has been explored using OT theory is auction design \cite{daskalakis2013mechanism}.
In their work, the authors demonstrated that the optimal auction mechanism for independently distributed items can be characterized by the Dual Formulation of an OT problem.
Moreover, they utilized this relationship to derive the optimal mechanism for various item classes, thereby establishing a fruitful application of OT theory in the context of mechanism design.
%


\section{Preliminaries}
\label{sec:preliminaries}

In this section, we fix the notations on the $m$-CFLP, Bayesian Mechanism Design, and Optimal Transport (OT). 
\paragraph{The $m$-Capacitated Facility Location Problem.}

Given a set of self-interested agents $[n]:=\{1,\dots,n\}$, we denote with $\vec x \in \erre^n$ the vector containing the agents' positions.
Likewise, given $m\in\mathbb{N}$, we denote with $\vec c\in \mathbb{N}^m$ the vector containing the capacities of the facilities, namely $\vec c=(c_1,\dots,c_m)$.
In this setting, a facility location is defined by three objects:
\begin{enumerate*}[label=(\roman*)]
    \item a $m$-dimensional vector $\vec y=(y_1,\dots,y_m)$ whose entries are $m$ positions on the line, 
    \item a permutation $\pi\in\mathcal{S}_m$ that decides the capacity of the facility built at $y_j$, so that $c_{\pi(j)}$ is the capacity of the facility built at $y_j$, and
    \item a matching $\Gamma\subset [n]\times [m]$ that determines how the agents are assigned to facilities, i.e. $(i,j)\in\Gamma$ if and only if the agent at $x_i$ is assigned to $y_j$. Due to the capacity constraints, the degree of vertex $j\in[m]$ according to $\Gamma$ is at most $c_{\pi(j)}$. Similarly, every agent is assigned to only one facility, thus the degree of every $i\in[n]$ according to $\Gamma$ is $1$.
\end{enumerate*} 
Given the positions of the facilities $\vec y$ and a matching $\Gamma$, we define the cost of an agent positioned in $x_i$ as $d_{i,\Gamma}(x_i,\vec y)=|x_i-y_j|$, where $(i,j)$ is the unique edge in $\Gamma$ adjacent to $i$.
Given a matching $\Gamma$ and a permutation $\pi \in \mathcal{S}_m$, a cost function is a map $C_\Gamma:\erre^n\times \erre^m\to [0,+\infty)$ that associates to $(\vec x , \vec y)$ the overall cost of placing the facilities at the positions in $\vec y$ following the permutation $\pi$ and assigning the agents positioned at $\vec x$ according to $\Gamma$.\footnote{In what follows, we omit $\Gamma$ from the indexes of $d$ and $C$ when it is clear from the context which matching we are considering.}
Given a vector $\vec x\in\erre^n$ containing the agents' positions, the $m$-\textit{Capacitated Facility Location Problem} with respect to the cost $C$, consists in finding the locations for $m$ facilities, a permutation $\pi$, and the matching $\Gamma$ that minimize the function $\vec y \to C(\vec x,\vec y)$.
The most studied cost function is the \textit{Social Cost} ($\SCo$), which is defined as the sum of all the agents' costs.
Since multiplying the cost function by a constant does not affect the approximation ratio results, throughout the paper we set
\[
\SCo(\vec x, \vec y)=\frac{1}{n}\sum_{i\in[n]}d_{i,\Gamma}(x_i,\vec y).
\]
Albeit throughout most of the paper the results will be stated for the Social Cost, we will also show that the techniques we propose can be extended to the \textit{Maximum Cost} (MC), which is defined as $$\MCo(\vec x, \vec y):=\max_{i\in [n]}d_{i,\Gamma}(x_i,\vec y),$$
and the \textit{$l_p$ cost}, which are defined as 
\[
C_{l_p}(\vec x, \vec y):=\Big(\frac{1}{n}\sum_{i=1}^{n}d_{i,\Gamma}(x_i,\vec y)^p\Big)^{\frac{1}{p}}, \quad \quad p\ge 1.
\]
Notice that both the Social Cost and the Maximal Cost are a special case of the $l_p$ costs \cite{feigenbaum2017approximately}.
Indeed, the Social Cost corresponds to the $l_p$ cost when $p=1$, while the Maximum Cost correspond to the $l_p$ cost for $p=\infty$.

\paragraph{Mechanism Design and the Worst-Case Analysis.}

An \textit{$m$-facility location mechanism} is a function $f$ that takes the agents' reports $\vec x$ in input and returns a set of $m$ positions $\vec y$ on the line, a permutation $\pi\in\mathcal{S}_m$, and an agent-to-facility matching $\Gamma$.
An agent may misreport its position if it results in a facility location such that the agent's incurred cost is smaller than reporting truthfully.
A mechanism $f$ is said to be \textit{truthful} (or \textit{strategy-proof}) if, the cost of every agent is minimized when it reports its true position. That is, 
\[
d_i(x_i,f(\vec x))\le d_i(x_i,f(\vec x_{-i},x_i'))
\]
for any misreport $x_i'\in \erre$, where $\vec x_{-i}$ is the vector $\vec x$ without its $i$-th component.
Although deploying a truthful mechanism prevents agents from getting a benefit by misreporting their positions, this leads to a loss of efficiency.
To evaluate this efficiency loss, Nisan and Ronen introduced the notion of approximation ratio \cite{nisan1999algorithmic}.
Given a truthful mechanism $f$, its worst-case approximation ratio with respect to the Social Cost is defined as $ar(f):=\sup_{\vec x \in \mathbb{R}^{n}}\frac{SC_{f}(\vec x)}{SC_{opt}(\vec x)}$,
where $SC_f(\vec x)$ is the Social Cost of placing the facilities and assigning the agents to them following the output of $f$, while $SC_{opt}(\vec x)$ is the optimal Social Cost achievable when the agents' report is $\vec x$.
The approximation ratio with respect to the $l_p$ cost or the Maximum Costs are defined similarly.
The \textit{Ranking Mechanisms} are a class of mechanisms for the $m$-CFLP that work under the assumption that the total capacity of the facilities matches the number of agents \cite{aziz2020facility}.
Each Ranking Mechanism is defined by two parameters: a permutation $\pi\in\mathcal{S}_m$ and a vector $\vec t=(t_1,\dots,t_m)\in[n]^m$. 
Given $\pi$ and $\vec t$ the routine of the Ranking Mechanism is as follows:
\begin{enumerate*}[label=(\roman*)]
    \item given $\vec x$ the vector containing the agents' reports ordered non-decreasingly, then the mechanism places the facility with capacity $c_{\pi(j)}$ at $x_{t_j}$.
    \item The agents are assigned to the facility from left to right while respecting the capacity constraints.
\end{enumerate*}
It was shown in \cite{aziz2020facility} that a Ranking Mechanism is truthful if and only if every $t_j$ admits at most two different values.
Moreover, the approximation ratio of these mechanisms is bounded only when the number of agents is even, the number of facilities to places is $2$, and the two facilities have the same capacity.
In this case, the mechanism is also called InnerPoint Mechanism (IM) and its approximation ratio is $\frac{n}{2}-1$.

\paragraph{Bayesian Mechanism Design.}
In Bayesian Mechanism Design, we assume that the agents' types follow a probability distribution and study the performance of mechanisms from a probabilistic viewpoint \cite{hartline2010bayesian,hartline2013bayesian}.
Every agent's type is then described by a random variable $X_i$.
In what follows, we assume that every $X_i$ is identically distributed according to a law $\mu$ and independent from the other random variables.
A mechanism $f$ is said to be truthful if, for every agent $i$, it holds
\begin{equation}
    \EE_{\vec X_{-i}}[d_i(x_i,f(x_i,\vec X_{-i}))]\le \EE_{\vec X_{-i}}[d_i(x_i,f(x_i',\vec X_{-i}))] \quad\quad\quad \forall x_i\in\erre,
\end{equation}
where $x_i$ agent $i$'s true type, $\vec X_{-i}$ is the $(n-1)$-dimensional random vector that describes the other agents' type, and $\EE_{\vec X_{-i}}$ is the expectation with respect to the joint distribution of $\vec X_{-i}$.
Given $\beta\in\erre$, a mechanism $f$ is a $\beta$-approximation if $\EE[SC_f(\vec X_n)]\le \beta\,\EE[SC_{opt}(\vec X_n)]$ holds, so that the lower $\beta$ is, the better the mechanism is.
To unify the notation, we define the Bayesian approximation ratio for the Social Cost as the ratio between the expected Social Cost of a mechanism and the expected Social Cost of the optimal solution.
More formally, given a mechanism $f$, its Bayesian approximation ratio is defined as follows 
\begin{equation}
\label{eq:B_ar}    
    B_{ar,1}^{(n)}(f):=\frac{\EE[SC_f(\vec X_n)]}{\EE[SC_{opt}(\vec X_n)]},
\end{equation}
where the expected value is taken over the joint distribution of the vector $\vec X_n := (X_1,\dots,X_n)$.
Notice that, if $B_{ar,1}^{(n)}(f)<+\infty$, then $f$ is a $B_{ar,1}^{(n)}(f)$-approximation.
Similarly, we define the Bayesian approximation ratio of $f$ with respect to the Maximum Cost and the $l_p$ costs: we denote the Bayesian approximation ratio of $f$ with respect to the $l_p$ cost with $B_{ar,p}^{(n)}(f)$ and the Bayesian approximation ratio with respect to the Maximum Cost with $B_{ar,\infty}^{(n)}(f)$.

\paragraph{Basic Notions on Optimal Transport.}
In the following, we denote with $\PP(\erre)$ the set of probability measures over $\erre$. 
Given $\gamma\in\PP(\erre)$, we denote with $spt(\gamma)\subset \erre$ the support of $\gamma$, that is, the smallest closed set $C\subset \erre$ such that $\gamma(C)=1$.
We denote with $\PP_k(\erre)$ the set of probability measures over $\erre$ whose support consists of $k$ points. 
That is, $\nu\in\PP_k(\erre)$ if and only if $\nu=\sum_{j=1}^k \nu_j\delta_{x_j}$, where $x_j\in \erre$ for every $j\in [k]$, $\nu_j\ge 0$ are real values such that $\sum_{j=1}^k\nu_j=1$, and $\delta_{x_j}$ is the Dirac's delta centered in $x_j$.
Given two measures $\alpha,\beta\in\PP(\erre)$ and $p\in [1,+\infty]$, the $p$-th order Wasserstein distance between $\alpha$ and $\beta$ is defined as 
\begin{equation}
    \label{eq:wass_p}
    W_p^p(\alpha,\beta)=\min_{\pi\in\Pi(\alpha,\beta)}\int_{\erre\times \erre}|x-y|^pd\pi,
\end{equation}
where $\Pi(\alpha,\beta)$ is the set of probability measures over $\erre\times \erre$ whose first marginal is equal to $\alpha$ and the second marginal is equal to $\beta$ \cite{kantorovich2006translocation}.
When $p=+\infty$, we set $W_\infty(\alpha,\beta)=\min_{\pi\in\Pi(\alpha,\beta)}\max_{(x,y)\in spt(\pi)} |x-y|.$
It is well-known that, for every $p\in [1,\infty]$, $W_p$ is a metric over $\PP(\erre)$.
For a complete introduction to the Optimal Transport theory, we refer the reader to \cite{villani2009optimal} and \cite{santambrogio2015optimal}.

\paragraph{Basic Assumptions.}
Finally, we layout the basic assumptions of our framework.
In what follows, we tacitly assume that the underlying distribution $\mu$ satisfies all the following properties:
\begin{enumerate*}[label=(\roman*)]
\item $\mu\in\PP(\erre)$ is absolutely continuous. We denote with $\rho_\mu$ its density.
\item The support of $\mu$ is an interval and that $\rho_\mu$ is strictly positive on the interior of the support.
\item $\rho_\mu$ is differentiable on the support of $\mu$.
\item The probability measure $\mu$ has finite first moment, i.e. $\int_\erre |x|d\mu < +\infty$.
\end{enumerate*}
Notice that, according to this set of assumptions, the cumulative distribution function (c.d.f.) of $\mu$, namely $F_\mu$, is locally bijective.

\section{The Extended Ranking Mechanisms}
\label{sec:ERM}

Since our study aims to examine the behaviour of the mechanism as the number of agents goes to infinity, expressing the problem in terms of absolute capacities is unsuitable.
For this reason, we rephrase the specifics of the $m$-CFLP in terms of percentages.
In particular, we shift our focus to the percentage capacity vector (p.c.v.) $\vec q\in(0,1)^m$.
Each value $q_j$ corresponds to the percentage of agents that the $j$-th facility can accommodate.
Given $\vec q$ and a number of agents $n$, we recover the absolute capacity of the $j$-th facility by setting $c_j=\floor{q_j(n-1)}+1$.
Conversely, when we are given the absolute capacities $c_j$ for a given number of agents $n$, the corresponding p.c.v. is $q_j=\frac{c_j}{n}$.
Without loss of generality, we assume that the entries of $\vec q$ are ordered non-increasingly.
%


\begin{mechanism}[Extended Ranking Mechanisms (ERMs)]
    Let $\vec q=(q_1,\dots,q_m)$ be a p.c.v..
    Given a permutation $\pi\in\mathcal{S}_m$ and a non-decreasingly ordered vector $\vec v=(v_1,\dots,v_m)\in [0,1]^m$, that is $v_j\le v_{j+1}$, the routine of the ERM associated with $\pi$ and $\vec v$, namely $\RM$, is as follows:
    \begin{enumerate*}[label=(\roman*)]
        \item First, we collect the reports of the agents and order them non-decreasingly. We denote with $\vec x$ the ordered vector containing the agents' reports, thus $x_i\le x_{i+1}$.
        \item Second, we elicit $m$ positions on the line by setting $y_j=x_{\floor{v_j(n-1)}+1}$ for every $j\in [m]$ and place the facility with capacity $q_{\pi(j)}$ at $y_j$.
        \item Finally, we assign every agent to the facility closer to the position they reported (break ties arbitrarily without overloading the facilities).
        \end{enumerate*}
\end{mechanism}

Notice that, in the routine of the ERM, the vector $\vec v$ plays the role that $\vec t$ plays in the routine of a Ranking Mechanism. 
Indeed, the main difference between the Ranking Mechanisms and our generalization lies in how the mechanism matches the agents to the facilities.
While the Ranking Mechanisms assign the agents monotonically from left to right, the ERM assigns every agent to the facility that is closer to their report.
Depending on the pair $(\pi,\vec v)$ and $\vec q$, however, the matching returned by the ERM might overload some of the facilities.
We say that a couple $(\pi,\vec v)$ induces a \textit{feasible} ERM if, for every $n\in\mathbb{N}$ and every $\vec x\in\erre^n$, the output of $\RM$ is a facility location for the $m$-CFLP induced by $\vec q$.
Given a p.c.v. $\vec q$, the set of parameters $(\pi,\vec v)$ that induce a feasible ERM is characterizable through a system of inequalities.
For the sake of simplicity, we first consider the case in which $\vec v$ does not have two equal entries, that is $v_j\neq v_l$ for every $j\neq l$.

\begin{theorem}
\label{thm:RM_feasible}
    Given $\vec q$ and $\vec v$ such that $v_j\neq v_i$ for every $j\neq i\in[m]$, then $\RM$ is feasible if and only if the following system of inequalities are satisfied 
    \begin{equation}
    \label{eq:feas_system}
    \begin{cases}
        q_{\pi(1)}\ge v_2\\
        q_{\pi(2)}\ge v_3-v_1\\
        \quad\quad\vdots\\
        q_{\pi(m-1)}\ge v_{m}-v_{m-2}\\
        q_{\pi(m)}\ge 1-v_{m-1}
    \end{cases}.
\end{equation}
\end{theorem}

\begin{proof}
%
Let $\vec x$ be the vector that contains the agents' reports ordered in a non-decreasing order.
First, we show that every time the system \eqref{eq:feas_system} holds, then $\RM$ does not overload any facility.
By definition of $\RM$, we have that every agent is assigned to its closest facility.
Furthermore, every facility shares its position with an agent, depending on the vector $\vec v$. 
Let us now focus the $j$-th facility, which, by definition, it is placed at $y_j=x_{\floor{v_j (n-1)}+1}$.
Without loss of generality, let us assume that $1<j<m$, thus there are two facilities, namely $y_{j-1}$ and $y_{j+1}$, such that $y_{j-1} \le y_j\le y_{j+1}$.
The maximum number of agents whose report is closer to $y_j$ is $\floor{v_{j+1}(n-1)}-\floor{v_{j-1}(n-1)}-1$.
Since the capacity of the $j$-th facility is $\floor{q_{\pi(j)} (n-1)}+1$, we have that the facility at $y_j$ cannot be overloaded if and only if $\floor{q_{\pi(j)}(n-1)}+1\ge \floor{v_{j+1}(n-1)}-\floor{v_{j-1}(n-1)}-1$.
We now show that the latter condition is implied by $q_{\pi(j)}\ge v_{j+1}-v_{j-1}$ or, equivalently, $q_{\pi(j)}(n-1)\ge (v_{j+1}-v_{j-1})(n-1)$.
Indeed, it is easy to see that $\floor{q_{\pi(j)}(n-1)}+1\ge q_{\pi(j)}(n-1)\ge \floor{q_{\pi(j)}(n-1)}$, thus we have $\floor{q_{\pi(j)}(n-1)}+1\ge q_{\pi(j)}(n-1)\ge (v_{j+1}-v_{j-1})(n-1)\ge \floor{v_{j+1}(n-1)}-\floor{v_{j-1}(n-1)}-1$.
%
Through a similar argument, we deal with $y_1$ and $y_m$.
We now show the inverse implication, i.e. if $(\pi,\vec v)$ satisfies system \eqref{eq:feas_system}, then $\RM$ is feasible.
Toward a contradiction, assume that one of the inequalities does not holds.
Without loss of generality, assume that $v_2>q_{\pi(1)}$, there then exists a value $N\in\mathbb{N}$ such that $v_2-\frac{2}{N-1}>q_{\pi(1)}$.
Let us then consider $n=N$, then it holds $\floor{v_2(N-1)}-1>\floor{q_{\pi(1)}(N-1)}$, hence $\floor{v_2(N-1)}>\floor{q_{\pi(1)}(N-1)}+1$.
Let us then consider the instance $\vec x\in\mathbb{R}^N$ defined as $x_i=0$ if $i=1,\dots,\floor{v_2(N-1)}$ and $x_i=1$ otherwise.
On this instance $\RM$ place a facility at $y_1=0$ and all the others at $1$, the facility at $0$ has capacity $\floor{q_{\pi(1)}(N-1)}+1$.
By definition, the number of agents whose closest facility is $y_1$ is $\floor{v_2(N-1)}>\floor{q_{\pi(1)}(N-1)}+1$, which contradicts the feasibility of the mechanism.
\end{proof}

The feasibility of the ERM ensures also its truthfulness.

\begin{theorem}
\label{thm:truthfulness}
    Given a p.c.v. $\vec q$, any feasible $\RM$  is truthful. 
    Thus it is also truthful in the Bayesian framework.
\end{theorem}

\begin{proof}
    Let $\vec q$ be a p.c.v. and let $\RM$ be a feasible ERM.
    Toward a contradiction, let $\vec x$ be an instance in which an agent is able to manipulate.
    Without loss of generality, let us assume that the manipulative agent's real position is $x_i$.
    Since every agent is assigned to the facility that is closer to the position they report, the manipulation performed by the agent at $x_i$ must alter the position of the facilities, as otherwise the cost of the manipulative agent cannot decrease.
    We notice that the positions of the facilities are determined through the same routine of the Percentile mechanisms, \cite{sui2013analysis}.
    Owing to the truthfulness of the Percentile mechanisms, we infer that no agent can misreport in such a way that a facility gets closer to its position, which concludes the proof.
\end{proof}

%

\begin{corollary}
\label{crr:sparecapacity}
    Given a p.c.v. $\vec q\in(0,1)^m$, let us fix $\vec v\in[0,1]^m$.
    Then, if $(\max_{j\in[m]}v_j-\min_{j\in[m]}v_j)>\sum_{j\in[m]}q_j-1$, the mechanism $\RM$ is not feasible for every $\pi\in\mathcal{S}_m$.
\end{corollary}

\begin{proof}
    It follows from Theorem \ref{thm:RM_feasible}, indeed if $\RM$ is feasible, then \eqref{eq:feas_system} holds.
    If we sum all the inequalities of system \eqref{eq:feas_system}, we get
    \begin{align*}
        \sum_{j\in[m]}q_{\pi(j)}=\sum_{j\in[m]}q_j&\ge v_2+(v_3-v_1)+\dots+(v_m-v_{m-2})+(1-v_{m-1})=1+v_m-v_1,
    \end{align*}
    since $\vec v$ is ordered increasingly.
    Thus, if $v_m-v_1=\big(\max_{j\in[m]}v_j-\min_{j\in[m]}v_j\big)>\sum_{j\in[m]}q_j-1$
    holds, system \eqref{eq:feas_system} does not, which concludes the proof.
\end{proof}

Finally, we notice that if $\vec v$ has at least two equal entries, not all the feasible ERMs necessarily satisfy system \eqref{eq:feas_system}.
This is the case for the \texttt{all-median} mechanism, which places all the facilities at the median agent.
%
%
Indeed, depending on $\vec q$, the vector $\vec v=(0.5,\dots,0.5)$ may not satisfy system \eqref{eq:feas_system} for any $\pi\in\mathcal{S}_m$.
When $v_j=v_l=v$ for some indices $j,l\in [m]$, we need to consider all the facilities placed at $x_{\floor{v(n-1)}+1}$ as if they were a unique facility whose capacity is the total capacity of the facilities placed at $x_{\floor{v(n-1)}+1}$.
By doing so, we extend Theorem \ref{thm:RM_feasible} and \ref{thm:truthfulness} to the cases in which $\vec v$ has at least two equal entries.
In particular, given $\vec v\in[0,1]^m$, let $\vec v'\in [0,1]^{m'}$ be the vector containing all the different entries of $\vec v$, then, the mechanism $\texttt{ERM}^{(\pi,\vec v)}$ is feasible if and only if the following system is satisfied
\[
\begin{cases}
        \sum_{l\in[m];\;s.t.\;v_{\pi(l)}=v_1'}q_{\pi(l)}\ge v_2'\\
        \sum_{l\in[m];\;s.t.\;v_{\pi(l)}=v_2'}q_{\pi(l)}\ge v_3'-v_1'\\
        \quad\quad\vdots\\
        \sum_{l\in[m];\;s.t.\;v_{\pi(l)}=v_{m-1}'}q_{\pi(l)}\ge v_{m'}'-v_{m'-2}'\\
        \sum_{l\in[m];\;s.t.\;v_{\pi(l)}=v_{m'}'}q_{\pi(l)}\ge 1-v_{m'-1}'
\end{cases}.
\]

Notice that Theorem \ref{thm:RM_feasible}, Theorem \ref{thm:truthfulness}, and Corollary \ref{crr:sparecapacity} hold regardless of the objective we are seeking to minimise, that is the Social Cost, the $l_p$ costs, and the Maximum Cost.
%
%

\section{The Bayesian Analysis of the ERMs}
\label{sec:analysis}
In this section, we present our main result, which characterizes the limit of the Bayesian approximation ratio of any feasible ERM as a function of $\vec v$, $\pi$, $\mu$, and $\vec q$ for all the relevant costs.
We start considering the Social Cost and extend all our results to the $l_p$ costs and the Maximum Cost in a dedicated section.
From now on, we tacitly assume that $q_i \, n\in\mathbb{N}$.

\subsection{The Social Cost case}

\label{SCsection}
First, we relate the $m$-CFLP to a norm minimization problem in the Wasserstein Space.

\begin{lemma}
\label{lemma:equivalenceSC_W1}
    Given a p.c.v. $\vec q$, let $\vec x\in\erre^n$ be the vector containing the agents' reports.
    Let us fix $\mu_{\vec x}:=\frac{1}{n}\sum_{i\in [n]}\delta_{x_i}$. 
    Then, it holds
    \begin{equation}
    \label{eq:min_discreteopt}
      SC_{opt}(\vec x)=\min_{\sigma\in\mathcal{S}_m}\min_{\zeta\in\PP_{\sigma,\vec q}(\erre)}W_1(\mu_n,\zeta),  
    \end{equation}
    where $\PP_{\sigma,\vec q}(\erre)$ is the set of probability measures such that $\zeta=\sum_{j\in[m]}\zeta_j\delta_{y_j}$, where $y_1\le y_2\le\dots\le y_m$ and $\zeta_j\le q_{\sigma(j)}$ for every $j\in[m]$.
    Similarly, given a permutation $\pi\in\mathcal{S}_m$ and a vector $\vec v$ that induce a feasible ERM, it holds 
    \begin{equation}
    \label{eq:discrete_proj2}
    SC_{\pi,\vec v}(\vec x)=\min_{\lambda_j\le q_{\pi(j)},\sum_{j\in[m]}\lambda_j=1}W_1(\mu_n,\lambda),
    \end{equation}
    where $SC_{\pi,\vec v}(\vec x)$ is the Social Cost attained by $\RM$ on instance $\vec x$, $\lambda=\sum_{j\in[m]}\lambda_j\delta_{x_{r_j}}$, and $r_j=\floor{v_j(n-1)}+1$.
\end{lemma}

\begin{proof}
    First, we show that \eqref{eq:min_discreteopt} holds.
    Given $\vec x$ and $\vec q$, let us consider an optimal solution to the $m$-CFLP.
    In particular, we denote with $\vec y$ be the vector containing $m$ positions on the line, $\sigma$ the permutation that determines how to distribute the facilities among the positions in $\vec y$, and $\Gamma$ a matching that assigns every agent to a facility without overloading the facilities.
    Let us now consider the following transportation plan $\pi_{x_i,y_j}= \sum_{(i,j)\in\Gamma}\frac{1}{n}\delta_{(x_i,y_j)}$.
    By definition of $\Gamma$, we have that $\sum_{j\in[m]}\pi_{x_i,y_j}=\frac{1}{n}$ and $\sum_{i\in[n]}\pi_{x_i,y_j}=\frac{k_j}{n}\le\frac{c_{\sigma(j)}}{n}=q_{\sigma(j)}$,
    where $k_j$ is the degree of $j\in[m]$ according to $\Gamma$.
    Thus, if we set $\nu_{n,m}=\sum_{j\in[m]}\frac{k_j}{n}\delta_{y_j}$, we have that $\nu_{n,m}\in\PP_{\sigma,m}(\erre)$.
    It is easy to see that $SC(\vec x)=\sum_{i\in[n],j\in [m]}|x_i-y_j|\pi_{x_i,y_j}\ge W_1(\mu_n,\nu_{n,m})$, hence $SC_{opt}(\vec x)\ge \min_{\sigma\in\mathcal{S}_m}\min_{\zeta\in\PP_{\sigma,m}(\erre)} W_1(\mu_n,\zeta)$.
    To conclude, we show that the inverse inequality holds.
    Let $(\sigma,\nu_{m,n})$ be a minimizer of the right hand side of \eqref{eq:min_discreteopt}, we show that there exists a set of $m$ locations $\vec y$, a permutation $\sigma'$, and a matching $\Gamma$ whose Social Cost is $W_1(\mu_n,\nu_{n,m})$.
    First, let us set $\vec y=(y_1,\dots,y_m)$, where $\{y_j\}_{j\in [m]}$ is the support of $\nu_{m,n}$.
    Let $q_j=\frac{c_j}{n}$, then there exists an optimal transportation plan $\pi$ between $\mu_n$ and $\nu_{m,n}$ such that for every $i\in[n]$, there exists a unique $j\in[m]$ such that $\pi_{x_i,y_j}\neq 0$ \cite{auricchio2022structure}.
    Notice that, by definition of transportation plan, if $\pi_{x_i,y_j}\neq 0$ for only a couple of values $i\in[n]$, $j\in[m]$, then $\pi_{x_i,y_j}=\frac{1}{n}$.
    Thus, the set $\Gamma=\{(i,j)\in [n]\times [m];\; \pi_{x_i,y_j}=\frac{1}{n}\}$ is well defined.
    Since for every $i\in[n]$ we have that there exists a unique $j\in[m]$ for which it holds $\pi_{x_i,y_j}=\frac{1}{n}$, we have that the degree of every $i\in[n]$ is one.
    Since $\sum_{i\in[n]}\pi_{x_i,y_j}=(\nu_{m,n})_j\le \frac{c_{\sigma(j)}}{n}$, we have that the degree of every $j\in[m]$ is less than $c_{\sigma(j)}$.
    Thus the triplet $(\vec y,\sigma,\Gamma)$ is a feasible location for the $m$-CFLP.
    Since the Social Cost of $(\vec y,\sigma,\mu)$ is equal to $\sum_{i\in[n],j\in[m]}|x_i-y_j|\pi_{x_i,y_j}=W_1(\mu_n,\nu_{m,n})$, we have that $\min_{\sigma\in\mathcal{S}_m}\min_{\zeta\in\PP_{\sigma,m}(\erre)} W_1(\mu_n,\zeta)=W_1(\mu_n,\nu_{m,n})\ge SC_{opt}(\vec x)$,
    which concludes the proof.
    Notice that, since $\min_{\sigma\in\mathcal{S}_m}\min_{\zeta\in\PP_{\sigma,m}(\erre)} W_1(\mu_n,\zeta)= SC_{opt}(\vec x)$ holds, we can refine the previous argument and show that, given an optimal facility location for the $m$-CFLP it is possible to build a solution to the minimization problem in the right hand side of \eqref{eq:min_discreteopt}.
    Vice-versa, given a solution to the minimization problem in the right hand side of \eqref{eq:min_discreteopt}, it is possible to build an optimal facility location for the $m$-CFLP.
    Using the same argument and the same constructions, it is possible to show that \eqref{eq:discrete_proj2} holds.
\end{proof}

%
%
%
The connection between the $m$-CFLP and Optimal Transport theory highlighted in Lemma \ref{lemma:equivalenceSC_W1} enables us to exploit the properties of the Wasserstein distances and to characterize the limit Bayesian approximation ratio of every feasible ERM.

\begin{theorem}
\label{thm:limitBAR}
    Given the p.c.v. $\vec q$, let $\vec v\in (0,1)^m$ and $\pi\in\mathcal{S}_m$ be such that $\RM$ is feasible.
    Then, it holds
    \begin{equation}
    \label{eq:bayesianlimitRM}
        \lim_{n\to +\infty}\frac{\EE[SC_{\pi,\vec v}(\vec x)]}{\EE[SC_{opt}(\vec x)]}=\frac{W_1(\mu,\nu_{\vec v})}{W_1(\mu,\nu_{m})},
    \end{equation}
    where $\nu_{m}$ is a solution to the following minimization problem
    \begin{equation}
    \label{eq:min_proj}
      \min_{\sigma\in \mathcal{S}_m}\min_{\zeta\in\PP_{\sigma,\vec q}(\erre)}W_1(\mu,\zeta),
    \end{equation}
    and $\nu_{\vec v}$ is a solution to 
    \begin{equation}
    \label{eq:min_proj2}
      \min_{\lambda_j\le q_{\pi(j)},\sum_{j\in[m]}\lambda_j=1}W_1\Big(\mu,\sum_{j\in[m]}\lambda_j\delta_{F_\mu^{[-1]}(v_j)}\Big).
    \end{equation}
\end{theorem}






\begin{proof}
    The proof consists of three steps:
    \begin{enumerate*}[label=(\roman*)]
        \item First, we show that the expected optimal Social Cost for the $m$-CFLP converges to the objective value of problem 
        \eqref{eq:min_proj}.
        \item Second, we show that the expected Social Cost of $\RM$ converges to the objective value of the minimization problem in \eqref{eq:min_proj2}.
        \item We combine the two convergence results to retrieve \eqref{eq:bayesianlimitRM}.
    \end{enumerate*}

    \textbf{The limit of the expected optimal Social Cost.}
    First, we show that a solution to problem \eqref{eq:min_proj} always exists.
    For every $m$ the number of permutation is finite, thus it suffices to show that, for every $\sigma\in\mathcal{S}_m$ there exists a solution to the problem $\min_{\zeta\in\PP_{\sigma,m}(\erre)}W_1(\mu,\zeta)$.
    Since $\mu$ has finite first moment, i.e. $\int_\erre |x|d\mu<+\infty$, the minimization problem is well-posed and its minimum is finite.
    Moreover, the set $\PP_{\sigma,m}(\erre)$ is a closed set in $\PP(\erre)$, thus a solution to the minimization problem does exist (although it may not be unique).
    Let $\nu_{m}$ be such that $W_1(\mu,\nu_{m})=\min_{\sigma\in \mathcal{S}_m}\min_{\zeta\in\PP_{\sigma,\vec q}(\erre)}W_1(\mu,\zeta)$.
    %
    From Lemma \ref{lemma:equivalenceSC_W1}, we know that, for every $n\in\mathbb{N}$ and for every  $\vec x\in\erre^n$, there exists a $\nu_{\vec x,m}$ such that $SC_{opt}(\vec x)=W_1(\mu_n,\nu_{\vec x,m})$.
    By definition of $\nu_{m}$, we have that $W_1(\mu,\nu_{m})\le W_1(\mu,\nu_{\vec x,m})$.
    Since $W_1$ is a distance, it holds $W_1(\mu,\nu_{m})\le W_1(\mu,\nu_{\vec x,m})\le W_1(\mu,\mu_n)+W_1(\mu_n,\nu_{\vec x,m})$.
    By rearranging the terms and by taking the expected value with respect to the distribution of $\vec X$, we obtain $\EE [W_1(\mu,\nu_{m})]-\EE [W_1(\mu_n,\nu_{\vec x,m})]\le \EE [W_1(\mu,\mu_n)]$. 
    By a similar argument, we have that $W_1(\mu_n,\nu_{\vec x,m})\le W_1(\mu_n,\nu_m)\le W_1(\mu,\mu_n)+W_1(\mu,\nu_m)$, hence $\EE [W_1(\mu_n,\nu_{\vec x,m})]-\EE[W_1(\mu,\nu_{m})]\le \EE [W_1(\mu,\mu_n)]$.
    We then infer that $|\EE[W_1(\mu,\nu_{m})]-\EE[W_1(\mu_n,\nu_{\vec x,m})]|\le \EE[W_1(\mu,\mu_n)]$.
    Since the right handside of this inequality converges to $0$ as $n$ goes to $+\infty$ (see \cite{bobkov2019one}), we infer that $ \lim_{n\to\infty}\EE [W_1(\mu_n,\nu_{\vec x,m})]=W_1(\mu,\nu_{m})$,
    which concludes the first part of the proof.

    \textbf{The limit of the expected Social Cost of the Mechanism.}
    First, we retrieve a minimizer to problem \eqref{eq:min_proj2}.
    Given $\vec v$, let us set $y_j=F_\mu^{[-1]}(v_j)$.
    For every $\lambda=\sum_{j\in[m]}\lambda_j\delta_{y_j}$ such that $\sum_{j\in[m]}\lambda_j=1$, we have that
    \[
    W_1(\mu,\lambda)=\sum_{j\in[m]}\int_{F_\mu^{[-1]}\big(\sum_{l=1}^{j-1}\lambda_l\big)}^{F_\mu^{[-1]}\big(\sum_{l=1}^{j}\lambda_l\big)}|x-y_j|d\mu\ge \int_{\erre}\min_{j\in[m]}|x-y_j|d\mu.
    \]
    Let us now consider the measure $\nu_{\vec v}=\sum_{j\in[m]}\Big(F_\mu(z_j)-F_\mu(z_{j-1})\Big)\delta_{y_j}$,
    where $z_j=\frac{y_{j+1}+y_j}{2}$ if $j=2,\dots, m-1$, $z_0=-\infty$ and $z_m=+\infty$.
    We then have
    \begin{align*}
        W_1(\mu,\nu_{\vec v})&=\sum_{j\in[m]}\int_{F_\mu^{[-1]}\big(\sum_{l=1}^{j-1}F_\mu(z_l)-F_\mu(z_{l-1})\big)}^{F_\mu^{[-1]}\big(\sum_{l=1}^{j}F_\mu(z_l)-F_\mu(z_{l-1})\big)}|x-y_j|d\mu
        =\int_{\erre}\min_{j\in[m]}|x-y_j|d\mu.
    \end{align*}
    Thus $\nu_{\vec v}$ is a solution to problem \eqref{eq:min_proj2}.
    We now study the limit of the expected Social Cost of the mechanism.
    The argument used for this part is similar to the one used for the limit expected optimal Social Cost but more delicate.
    Indeed, in this case, the set on which we minimize the Wasserstein distance does depend on the agents' report $\vec x$.
    Moreover, the sets on which are formulated problem \eqref{eq:min_proj2} and \eqref{eq:discrete_proj2} are, in general, different.
    To overcome these issues, we need to define two auxiliary probability measures, namely $\phi$ and $\psi$.
    Given $\vec x\in\erre^n$, let $\nu_{\vec v}$ and $\nu_{\vec x, \vec v}$ be the solutions to problem \eqref{eq:min_proj2} and \eqref{eq:discrete_proj2}, respectively.
    We define the measures $\phi=\sum_{j\in[m]}(\nu_{\vec x, \vec v})_j\delta_{y_j}$, where $y_j$ is the support of the measure $\nu_{\vec v}$.
    For every $n\in\mathbb{N}$ and every $\vec x\in \erre^n$, we have that $W_1(\mu,\nu_{\vec v})\le W_1(\mu,\phi)\le W_1(\mu,\mu_n)+W_1(\mu_n,\nu_{\vec x, \vec v})+W_1(\nu_{\vec x, \vec v},\phi).$
    We therefore infer
    \begin{equation}
        \label{eq:inequality_1}
         W_1(\mu,\nu_{\vec v})-W_1(\mu_n,\nu_{\vec x, \vec v})\le W_1(\mu,\mu_n)+W_1(\nu_{\vec x, \vec v},\phi).
    \end{equation}
    Similarly, given $\vec x\in\erre^n$, we define $\psi=\sum_{j\in[m]}(\nu_{\vec v})_j\delta_{y_{\vec x,j}}$, where $\{y_{\vec x,j}\}_{j\in[m]}$ is the support of $\nu_{\vec x, \vec v}$.
    We then have
    \begin{equation}
        \label{eq:inequality_2}
         W_1(\mu_n,\nu_{\vec x, \vec v})-W_1(\mu,\nu_{\vec v})\le W_1(\mu,\mu_n)+W_1(\nu_{\vec v},\psi).
    \end{equation}
   
    Since the Wasserstein distance is always non negative, we can combine the estimations in \eqref{eq:inequality_1} and \eqref{eq:inequality_2}, to obtain
    \[
    |W_1(\mu_n,\nu_{\vec x, \vec v})-W_1(\mu,\nu_{\vec v})|\le W_1(\mu,\mu_n)+W_1(\nu_{\vec v},\psi)+W_1(\nu_{\vec x, \vec v},\phi).
    \]
    If we take the expectation of both sides of the inequality the inequality still holds.
    Since $\lim_{n\to\infty}\EE [W_1(\mu_n,\mu)]=0$ (see \cite{bobkov2019one}), we just need to prove that $\lim_{n\to\infty}\EE [W_1(\nu_{\vec v},\psi)]=\lim_{n\to\infty}\EE [W_1(\nu_{\vec x, \vec v},\phi)]=0$.
    Let us consider $\EE [W_1(\nu_{\vec v},\psi)]$, the convergence of $\EE [W_1(\nu_{\vec x, \vec v},\phi)]$ follows by a similar argument.
    We notice that $\psi$ and $\nu_{\vec v}$ have different supports, but $\psi_j=(\nu_{\vec v})_j$, thus it holds $\EE [W_1(\nu_{\vec v},\psi)]\le \sum_{j\in[m]}\psi_j\EE [|y_j-y_{\vec x,j}|]$, where $y_{\vec x,j}$ is the $j$-th point in the support of $\nu_{\vec x, \vec v}$.
    By definition of ERM, it holds $y_{\vec x,j}=x_{\floor{v_j(n-1)}+1}$.
    Since the $\big(\floor{v_j(n-1)}+1\big)$-th order statistics converges to the $v_j$-th quantile of $\mu$ \cite{de1979bahadur}, we have that $\EE [|y_j-y_{\vec x,j}|]\to 0$ as $n\to \infty$.
    Indeed, since $(\nu_{\vec x, \vec v})_j\le 1$ for every $j\in[m]$, we have $\EE[W_1(\nu_{\vec x, \vec v},\phi)]\le \sum_{j\in[m]}\EE[|y_{\vec x,j}-y_j|]$, where $y_{\vec x,j}=x_{\floor{v_j(n-1)}+1}$ and $y_j=F_\mu^{[-1]}(v_j)$.
    Using Bahadur's representation formula \cite{10.1214/aoms/1177699450}, we have that
    \begin{equation}
        \label{eq:bahadur}
        X_{\floor{(n-1)v_j}+1}-F_\mu^{[-1]}(v_j)=\frac{S_n(F_\mu^{[-1]}(v_j))-v_j}{\rho_\mu(F_\mu^{[-1]}(v_j))}+R_n,
    \end{equation}
    where $R_n$ is the rest of the Bahadur's formula, for which holds $R_n\le O(n^{-\frac{3}{4}})$ with probability $1$, $\rho_\mu$ is the density of $\mu$, and $S_n(t)=\frac{1}{n}\sum_{i=1}^n \mathbb{I}(X_i)_{\{X_i\le t\}}$ where $\mathbb{I}(X_i)_{\{X_i\le t\}}=1$ if $X_i\le t$ and $\mathbb{I}(X_i)_{\{X_i\le t\}}=0$ otherwise.
    Since $\rho_\mu$ is always strictly positive on the interior of the support of $\mu$ and $v_j\in (0,1)$ for every $j=1,\dots,m$, we have
    \begin{align*}
        \EE[&|y_{\vec x,j}-y_j|]\le K \EE\big[|S_n(F_\mu^{[-1]}(v_j))-v_j|\big]+\EE[|R_n|]\\
        &\le K\EE\bigg[\frac{1}{n}|\sum_{i=1}^n(\mathbb{I}_{\{X_i\le F_\mu^{[-1]}(v_j)\}}(X_i)-v_j)|\bigg]+ Cn^{-\frac{3}{4}}\\
        &\le K\EE\Big[\Big|\sum_{i=1}^n\frac{(\mathbb{I}_{\{X_i\le F_\mu^{[-1]}(v_j)\}}(X_i))-v_j)}{n}\Big|\Big]+ Cn^{-\frac{3}{4}} \\ 
        & \leq \frac{K}{\sqrt{n}}Var(\mathbb{I}_{\{X_1\le F_\mu^{[-1]}(v_j)\}}(X_1))+ Cn^{-\frac{3}{4}},
    \end{align*}
    since the $X_i$ are i.i.d. and $\mathbb{I}_{\{X_i\le F_\mu^{[-1]}(v_j)\}}(X_i)$ is a Bernoulli variable that with probability $v_j$ is equal to $1$ and equal to $0$ otherwise, thus its variance is finite.
    We then conclude that $\EE [W_1(\nu_{\vec x, \vec v},\phi)]$ converges to $0$.
    %

    \textbf{Characterizing the Bayesian approximation ratio.}
    To conclude, notice that the distance between an absolutely continuous measure and a discrete measure is always greater than zero, thus $\lim_{n\to\infty}\EE [SC_{opt}(\vec X)]>0$.
    For this reason, we have that the limit of the ratio is equal to the ratio of the limits, which proves \eqref{eq:bayesianlimitRM}.
    %
\end{proof}

Notice that our result applies only to feasible $\RM$, such that $\vec v\in(0,1)^m$, since, for general measures $\mu$, the values $F_\mu^{[-1]}(0)$ and $F_\mu^{[-1]}(1)$ might not be finite.
%
%
%
%
Finally, we notice that the Bayesian approximation ratio is invariant to positive affine transformation of $\mu$.
%
%

\begin{corollary}
\label{crr:scale_invariance}
    Let $\vec q$ be a p.c.v. and let $X$ be the random variable associated with $\mu$.
    Given $\alpha>0$ and $\beta\in\erre$, let $\mu_{\alpha,\beta}$ be the probability distribution associated with $\alpha X +\beta$.
    Then, the asymptotical Bayesian approximation ratio of any feasible ERM is the same regardless of whether the agent type is distributed according to $\mu$ or $\mu_{\alpha,\beta}$.
\end{corollary}

\begin{proof}
    Let $\nu^{(\alpha,\beta)}_m$ be the solution to $\min_{\sigma\in\mathcal{S}_m}\min_{\zeta\in\PP_{\sigma,m}(\erre)}W_1(\mu_{\alpha,\beta},\zeta)$ and let $\nu_m$ be the solution to \eqref{eq:min_proj}.
    Let us fix $L(x)=\alpha x +\beta$.
    Since $\alpha>0$, $L$ is a bijective and monotone-increasing function.
    We denote the inverse function of $L$ with $H$.
    It is well-known that if $X$ is a random variable whose law is $\mu$, then the law of $L(X)=\alpha X+\beta$ is $L_\#\mu$.
    We recall that $L_\#\mu$ is the pushforward of the measure $\mu$, defined as $L_\#\mu(A)=\mu(H(A))$, see \cite{villani2009optimal}.
    By definition, we have that $L_\#\delta_{x}=\delta_{\alpha x + \beta}$, thus $L_\#(\sum_{j\in [m]}\lambda_j\delta_{x_j})=\sum_{j\in [m]}\lambda_j\delta_{\alpha x_j + \beta}$.
    We now show that $\nu^{(\alpha,\beta)}_m=L_\#\nu_m$.
    First, notice that, since $\alpha>0$ the function $L$ is monotone increasing, thus if $\nu_m\in\PP_{\sigma,m}(\erre)$, then $L_\#\nu_m\in\PP_{\sigma,m}(\erre)$.
    Toward a contradiction, let us now assume that $\gamma$ is such that $W_1(\gamma,L_\#\mu)<W_1(L_\# \nu_m,L_\#\mu)$.
    Let us now define $\eta=H_\#\gamma\in\PP_k(\erre)$.
    Since $H$ is also monotone increasing, we have $\gamma\in\PP_{\sigma,m}(\erre)$ and $\gamma=L_\#\eta$.
    Furthermore, by the properties of the Wasserstein distance, we have $W_1(L_\#\eta,L_\#\mu)=\alpha W_1(\eta,\mu)$
    and $W_1(L_\#\nu_m,L_\#\mu)=\alpha W_1(\nu_m,\mu)$.
    In particular, we get $W_1(\eta,\mu)<W_1(\nu_m,\mu)$,
    which contradicts the optimality of $\nu_m$.
    We therefore conclude that $L_\#\nu_m=\nu^{(\alpha,\beta)}_m$ and that $W_1(\nu^{(\alpha,\beta)}_m\mu_{\alpha,\beta})=\alpha W_1(\nu_m,\mu)$.
    To conclude, we observe that
    \begin{align*}
            \min_{\substack{\lambda_j\le q_{\sigma(j)},\\\sum_{j\in [m]}\lambda_j=1}}&W_1\Big(\mu_{\alpha,\beta},\sum_{j\in[m]}\lambda_j\delta_{\alpha y_j+\beta}\Big) = \alpha\min_{\substack{\lambda_j\le q_{\sigma(j)},\\\sum_{j\in [m]}\lambda_j=1}}W_1\Big(\mu,\sum_{j\in[m]}\lambda_j\delta_{y_j}\Big),
    \end{align*}
    for any set of $m$ values $\{y_j\}_{j\in[m]}$.
    We then observe that $F^{[-1]}_{\mu_{\alpha,\beta}}(v_j)=\alpha F^{[-1]}_{\mu}(v_j) +\beta$, thus
    \begin{align*}
        &\min_{\substack{\lambda_j\le q_{\sigma(j)},\\\sum_{j\in [m]}\lambda_j=1}}W_1\Big(\mu_{\alpha,\beta},\sum_{j\in[m]}\lambda_j\delta_{\alpha F^{[-1]}_{\mu}(v_j)+\beta}\Big)=\min_{\substack{\lambda_j\le q_{\sigma(j)},\\\sum_{j\in [m]}\lambda_j=1}}W_1\Big(\mu_{\alpha,\beta},\sum_{j\in[m]}\lambda_j\delta_{F^{[-1]}_{\mu_{\alpha,\beta}}(v_j)}\Big).
    \end{align*}

    By taking the ratio of the two expectations, the factor $\alpha$ evens out, hence the limit of the Bayesian approximation ratio of any ERM does not depend on $\alpha$ nor $\beta$.
\end{proof}

\subsection{Extending to the $l_p$ costs and the Maximum Cost}
\label{sect:extlpMC}

In this section, we extend the findings of Section \ref{SCsection} to encompass the $l_p$ costs and the Maximum Cost.
First, we extend Lemma \ref{lemma:equivalenceSC_W1}.
%

\begin{lemma}
\label{lemma:equivalenceSC_lpmc}
    Given a p.c.v. $\vec q$, let $\vec x\in\erre^n$ be the vector containing the agents' reports.
    Let us fix $\mu_{\vec x}:=\frac{1}{n}\sum_{i\in [n]}\delta_{x_i}$. 
    Then, it holds
    \begin{equation}
    \label{eq:min_discreteoptlp}
      (C_{l_p})_{opt}(\vec x)=\min_{\sigma\in\mathcal{S}_m,\zeta\in\PP_{\sigma,\vec q}(\erre)}W_p(\mu_n,\zeta)\quad\text{and}\quad \MCo_{opt}(\vec x)=\min_{\sigma\in\mathcal{S}_m,\zeta\in\PP_{\sigma,\vec q}(\erre)}W_\infty(\mu_n,\zeta),  
    \end{equation}
    where $\PP_{\sigma,\vec q}(\erre)$ is the set of probability measures such that $\zeta=\sum_{j\in[m]}\zeta_j\delta_{y_j}$, where $y_1\le y_2\le\dots\le y_m$ and $\zeta_j\le q_{\sigma(j)}$ for every $j\in[m]$.
    Similarly, given a permutation $\pi\in\mathcal{S}_m$ and a vector $\vec v$ that induce a feasible ERM, it holds 
    \begin{equation}
    \label{eq:discrete_proj2lp}
    (C_{l_p})_{\pi,\vec v}(\vec x)=\min_{\substack{\lambda_j\le q_{\pi(j)},\\\sum_{j\in[m]}\lambda_j=1}}W_p(\mu_n,\lambda)\quad\text{and}\quad \MCo_{\pi,\vec v}(\vec x)=\min_{\substack{\lambda_j\le q_{\pi(j)},\\\sum_{j\in[m]}\lambda_j=1}}W_\infty(\mu_n,\lambda), 
    \end{equation}
     where $(C_{l_p})_{\pi,\vec v}(\vec x)$ and $MC_{\pi,\vec v}(\vec x)$ are the $l_p$ and Maximum Cost attained by $\RM$ on instance $\vec x$, respectively, $\lambda=\sum_{j\in[m]}\lambda_j\delta_{x_{r_j}}$, and $r_j=\floor{v_j(n-1)}+1$.
\end{lemma}

\begin{proof}
    The proof of the Lemma follows the same argument used to prove Lemma \ref{lemma:equivalenceSC_W1}.
    Indeed, the construction used to build a solution to problem \eqref{eq:min_discreteoptlp} starting from a solution to the $m$-CFLP (and \textit{vice-versa}) can be implemented to handle any $l_p$ cost function. 
\end{proof}

Owing to Lemma \ref{lemma:equivalenceSC_lpmc}, we can adapt the argument used to prove Theorem \ref{thm:limitBAR} to encompass the $l_p$ and Maximum costs.
Notice, however, that the set of conditions under which the Bayesian approximation ratio converges depends on the value of $p$ we are considering.

\begin{theorem}
\label{thm:limitBAR_p_MC}
    Given the p.c.v. $\vec q$, let $\vec v\in (0,1)^m$ and $\pi\in\mathcal{S}_m$ be such that $\RM$ is feasible.
    Given $p\in(1,+\infty)$, then if $\int_{\erre}|x|^p d\mu<+\infty$ holds, we have
    \begin{equation}
    \label{eq:bayesianlimitRMlp}
        \lim_{n\to +\infty}\frac{\EE[(C_{l_p})_{\pi,\vec v}(\vec x)]}{\EE[(C_{l_p})_{opt}(\vec x)]}=\frac{W_p(\mu,\nu_{\vec v})}{W_p(\mu,\nu_{m})},
    \end{equation}
    where $\nu_{m}$ is a solution to 
    \begin{equation}
    \label{eq:lp_1}
        \min_{\sigma\in \mathcal{S}_m}\min_{\zeta\in\PP_{\sigma,\vec q}(\erre)}W_p(\mu,\zeta),
    \end{equation}
    while $\nu_{\vec v}$ is a solution to 
    \begin{equation}
    \label{eq:lp_2}
        \min_{\lambda_j\le q_{\pi(j)},\sum_{j\in[m]}\lambda_j=1}W_p\Big(\mu,\sum_{j\in[m]}\lambda_j\delta_{F_\mu^{[-1]}(v_j)}\Big).
    \end{equation}
    If $\mu$ has a compact support, then we have that 
    \begin{equation}
    \label{eq:bayesianlimitRMmc}
        \lim_{n\to +\infty}\frac{\EE[\MCo_{\pi,\vec v}(\vec x)]}{\EE[\MCo_{opt}(\vec x)]}=\frac{W_\infty(\mu,\nu_{\vec v})}{W_\infty(\mu,\nu_{m})},
    \end{equation}
    where $\nu_{m}$ is a solution to \eqref{eq:lp_1} for $p=\infty$ and $\nu_{\vec v}$ is a solution to \eqref{eq:lp_2} for $p=\infty$.
\end{theorem}

\begin{proof}
    First, we notice that condition $\int_{\erre}|x|^pd\mu<+\infty$, ensures that $\min_{\nu\in\PP_k(\erre)}W_p(\mu,\nu)<+\infty$.
    Similarly, if $\mu$ has compact support, then $\min_{\nu\in\PP_k(\erre)}W_\infty(\mu,\nu)$ is finite.
    In particular, the right-handside of the limits in \eqref{eq:bayesianlimitRMlp} and \eqref{eq:bayesianlimitRMmc} are well-defined.
    %

    %
    %
    Since every $W_p$ is a metric, we can adapt the argument used in the proof of Theorem \ref{thm:limitBAR} to infer that the denominator of the Bayesian approximation ratio converges to $\min_{\nu\in\PP_k(\erre)}W_p(\mu,\nu)$ as $n$ goes to $\infty$.
    The argument used to prove the convergence of the numerator is, however, different.
    Let $W_p$ be the Wasserstein distance associated with $p\in(1,+\infty)$.
    Using the metric properties of $W_p$, we retrieve 
    \[
        |W_p(\mu_n,\nu_{\vec x, \vec v})-W_p(\mu,\nu_{\vec v})|\le W_p(\mu,\mu_n)+W_p(\nu_{\vec v},\psi)+W_p(\nu_{\vec x, \vec v},\phi),
    \]
    where  $\nu_{\vec v}=\sum_{j\in[m]}(\nu_{\vec v})_j\delta_{y_{j}}$ is the solution to $\min_{\substack{\lambda_j\le q_{\pi(j)},\\\sum_{j\in[m]}\lambda_j=1}}W_p\Big(\mu,\sum_{j\in[m]}\lambda_j\delta_{F_\mu^{[-1]}(v_j)}\Big)$,
        $\nu_{\vec x, \vec v}=\sum_{j\in[m]}(\nu_{\vec v})_{\vec x,j}\delta_{y_{\vec x,j}}$ is the solution to problem \eqref{eq:discrete_proj2lp},
        $\psi=\sum_{j\in[m]}(\nu_{\vec v})_j\delta_{y_{\vec x,j}}$, where $\{y_{\vec x,j}\}_{j\in[m]}$ is the support $\nu_{\vec x, \vec v}$, and
        $\phi=\sum_{j\in[m]}(\nu_{\vec x, \vec v})_j\delta_{y_j}$.
    %

    %
    Owing to the assumption on $\mu$ and to the results in \cite{bobkov2019one}, we have that $\lim_{n\to\infty}W_p(\mu,\mu_n)=0$.
    We now handle $W_p(\nu_{\vec v},\psi)$.
    If $p<\infty$, we have that
    \begin{align*}
        \EE[W_p(\nu_{\vec v},\psi)]&\le\EE\Big[\Big(\sum_{j\in[m]}(\nu_{\vec v})_j|y_j-y_{\vec x,j}|^p\Big)^{\frac{1}{p}}\Big]\le \EE\Big[\Big(\sum_{j\in[m]}|y_j-y_{\vec x,j}|^p\Big)^{\frac{1}{p}}\Big]\\
        & \le \EE\Big[\sum_{j\in[m]}\big(|y_j-y_{\vec x,j}|^p\big)^{\frac{1}{p}}\Big]=\EE\Big[\sum_{j\in[m]}|y_j-y_{\vec x,j}|\Big] = \sum_{j\in[m]}\EE\Big[|y_j-y_{\vec x,j}|\Big].
    \end{align*}
    By the same argument used to prove Theorem \ref{thm:limitBAR}, we infer that $\lim_{n\to\infty}\EE[W_p(\nu_{\vec v},\psi)]=0$.
    Finally, we notice that if $p=\infty$, we have that $\EE[W_\infty(\nu_{\vec v},\psi)]=\EE[\max_{j\in[m]}|y_j-y_{\vec x,j}|]\le \EE\Big[\sum_{j\in[m]}|y_j-y_{\vec x,j}|\Big] \le \sum_{j\in[m]}\EE[|y_j-y_{\vec x,j}|]$,
    which allows us to conclude the proof for the case $p=\infty$.
    Similarly, we have that $W_\infty(\nu_{\vec x, \vec v},\phi)$ goes to $0$ as the number of agent increases, which concludes the proof.
\end{proof}

%

\subsection{The Bayesian Study of previously introduced menchanisms}
\label{ex:im}
    We now use Theorem \ref{thm:limitBAR} and \ref{thm:limitBAR_p_MC} to the study of the Innerpoint and the Extended Endpoint Mechanisms.
    For the sake of simplicity, we limit our study to the Bayesian approximation ratio of the Social Cost, since the other costs are handled through a similar argument.

    \paragraph{The Innerpoint Mechanism.}
    Consider now a $2$-CFLP in which we have an even number of agents, thus $n=2k$, and the two facilities have the same capacity, so that $\vec q=(0.5,0.5)$.
    The Innerpoint Mechanism (IM) places the two facilities at $y_1$ and $y_2$, where $y_1$ is the $k$-th agent report from the left and $y_2$ the $(k+1)$-th agent report from the left.
    Every agent whose report is on the left of $y_1$ is assigned to the facility at $y_1$, while the others to $y_2$. 
    Since the IM works only for instances with an even number of agents, the limit will taken for $k\to \infty$.
    It is easy to see that, for every given $k\in \mathbb{N}$, the IM can be written as a Ranking Mechanism.
    Indeed, if we set $\vec v=(0.5-\frac{1}{2k},0.5)$ and $\pi=Id$ (where $Id(i)=i$ for every $i\in [2]$), the output of the IM and the output of the ERM associated to $\vec v$ and $\pi$ coincide on every instance.
    However, the vector $\vec v$ determining the output of the ERM depends on $k$, thus we cannot directly apply Theorem \ref{thm:limitBAR} to infer the limit Bayesian approximation ratio of the IM.
    To overcome this issue, we consider the Ranking Mechanism induced by the vector $\vec v=(0.5,0.5)$.
    Since the two facilities have the same capacity the mechanism does not depend on the permutation we choose, thus we will omit it for the sake of simplicity.
    Indeed, for every $k\in\mathbb{N}$ and $\vec x\in\erre^{2k}$, it holds $SC_{IM}(\vec x)\le SC_{RM^{(Id,\vec v)}}(\vec x)$, hence
    \begin{align}
    \label{eq:inequalitybound}
        \lim_{k\to\infty}\frac{\EE [SC_{IM}(\vec X)]}{\EE [SC_{opt}(\vec X)]}&\le\lim_{n\to\infty}\frac{\EE [SC_{Id,\vec v}(\vec X)]}{\EE [SC_{opt}(\vec X)]}=\frac{W_1(\mu,\delta_{med(\mu)})}{W_1(\mu,\nu_{m})},
    \end{align}
    where $\nu_{m}$ is a solution to problem \eqref{eq:min_proj}.
    Notice the inequality in \eqref{eq:inequalitybound} allows us to conclude that the Bayesian approximation ratio of the IM is finite.
    We now prove the other inequality.
    Notice that if $k>10$ and $\epsilon\in(0,0.1)$, it holds true the following $SC_{IM}(\vec x)\ge SC_{\texttt{ERM}^{(Id,\vec v_\epsilon)}}(\vec x)$,
    where $\vec v_{\epsilon}=(0.5-\epsilon,0.5)$.
    Notice that the mechanism $\texttt{ERM}^{(Id,\vec v_\epsilon)}$ is not feasible for $\vec q=(0.5,0.5)$, but it is feasible if we consider $\vec q_\epsilon=(0.5+\epsilon,0.5+\epsilon)$.
    Indeed, for $\vec q_\epsilon$ the routine of the IM is still well-defined and the cost of the output of IM does not depend on which $\vec q$ we consider, thus the inequality holds.
    If we take the average and the limits for $k\to \infty$ on both sides, we infer from Theorem
    \ref{thm:limitBAR} that
    \begin{align*}
        \nonumber\lim_{k\to\infty}\frac{\EE [SC_{IM}(\vec X)]}{\EE [SC_{opt}(\vec X)]}&\ge\lim_{n\to\infty}\frac{\EE [SC_{Id,\vec v_\epsilon}(\vec X)]}{\EE [SC_{opt}(\vec X)]}=\frac{W_1(\mu,\eta_\epsilon)}{W_1(\mu,\nu_{m})},
    \end{align*}
    where $\eta_\epsilon=\lambda_\epsilon\delta_{F^{[-1]}_\mu(0.5-\epsilon)}+(1-\lambda_\epsilon)\delta_{med(\mu)}$ and $\lambda_\epsilon=F_\mu\Big(\frac{F^{[-1]}_\mu(0.5-\epsilon)+F^{[-1]}_\mu(0.5)}{2}\Big)$.
    %
    Since $\eta_\epsilon$ converges in probability to $\delta_{med(\mu)}$ when $\epsilon\to 0$, then $\lim_{k\to\infty}\frac{\EE [SC_{IM}(\vec X)]}{\EE [SC_{opt}(\vec X)]}\ge  \frac{W_1\big(\mu,\delta_{med(\mu)}\big)}{W_1(\mu,\nu_{m})}$,
    hence the limt of $B_{ar,1}(IM)$
    is $\frac{W_1(\mu,\delta_{med(\mu)})}{W_1(\mu,\nu_{m})}$, where $\nu_{m}$ is the solution to problem \eqref{eq:min_proj}.
    %
    %

    %
    Finally, the same argument allows to study the limit Bayesian approximation ratio of the InnerChoice Mechanism proposed in \cite{ijcai2022p75} and any truthful Ranking Mechanism.
    In particular, the InnerChoice Mechanism has the same limit Bayesian approximation ratio of the IM, with respect to all the cost functions.
%


\paragraph{The Extended Endpoint Mechanism.}
Given a p.c.v. $\vec q$ , the Extended Endpoint Mechanism (EEM) is a mechanism able to handle any $2$-CFLP \cite{aziz2020capacity}.
In our formalism, the routine of the EEM is as follows.
Let $\vec x$ be a vector containing the agents' report, without loss of generality, we assume that $\vec x$ is ordered non decreasingly, i.e. $x_i\le x_{i+1}$.
We define $A_1=\{x_i\;\textrm{s.t.}\; |x_i-x_1|\le\frac{1}{2}|x_1-x_n|\}$ and $A_2=\{x_i\;\textrm{s.t.}\; x_i\notin A_1\}$.
If $|A_1|\ge |A_2|$, the EEM determines the position of two facilities as follows:
\begin{enumerate*}[label=(\roman*)]
    \item If $|A_1|\le \floor{q_1(n-1)}+1$ and $|A_2|\le \floor{q_2(n-1)}+1$, we set $y_1=x_1$ and $y_2=x_n$, we place the facility with capacity $q_1$ at $y_1$ and the facility with capacity $q_2$ at $y_2$.

    \item If $|A_1|> \floor{q_1(n-1)}+1$ and $|A_2|\le \floor{q_2(n-1)}+1$, we set $y_1=2x_{\floor{q_1(n-1)}+2}-x_n$ and $y_2=x_n$, we place the facility with capacity $q_1$ at $y_1$ and the facility with capacity $q_2$ at $y_2$.

    \item If $|A_1|\le \floor{q_1(n-1)}+1$ and $|A_2|> \floor{q_2(n-1)}+1$, we set $y_1=x_1$ and $y_2=2x_{n-(\floor{q_2(n-1)}+1)}-x_1$, we place the facility with capacity $q_1$ at $y_1$ and the facility with capacity $q_2$ at $y_2$.
\end{enumerate*}
Every agent is assigned to the facility that is closer to the position (break ties arbitrarily without  overloading any facility).
Finally, if $|A_1|> |A_2|$, we switch the roles of the two facilities in the cases described above.

We now study the limit Bayesian approximation ratio of the EEM when the agents are distributed according to a uniform distribution and $\vec q=(q_1,q_2)$ with $q_1>0.5>q_2$.
For the sake of simplicity, we consider only a odd number of agents $n=2k+1$.
First, we notice that we can restrict our attention to the class of instances in which $|A_1|>|A_2|$.
Indeed, due to the symmetry of the uniform distribution, we have that the every instance for which $|A_1|>|A_2|$ holds, uniquely identifies an instance for which it holds $|A_1|<|A_2|$\footnote{Since $n$ is odd we have that we cannot have $|A_1|=|A_2|$.}.
Let us consider the expected Social Cost of the instances for which it holds $|A_1|>|A_2|$, which we denote with $\mathcal{A}$.
Using the properties of the expected values, we can compute the expected Social Cost of the EEM as the sum of three conditioned expected values
\begin{align*}
    \EE[SC_{\pi,\vec v}(\vec X)]&=\EE[SC_{\pi,\vec v}(\vec X)|\vec X\in C_1]P(\vec X\in C_1)+\EE[SC_{\pi,\vec v}(\vec X)|\vec X\in C_2]P(\vec X\in C_2)\\
    &\quad+\EE[SC_{\pi,\vec v}(\vec X)|\vec X\in C_3]P(\vec X\in C_3),
\end{align*}
where \begin{enumerate*}[label=(\roman*)]
    \item $C_1$ contains the instances for which $|A_1|\le \floor{q_1(n-1)}+1$ and $|A_2|\le\floor{q_2(n-1)}+1$;
    \item $C_2$ contains the instances for which $|A_1|> \floor{q_1(n-1)}+1$ and $|A_2|\le \floor{q_2(n-1)}+1$; and
    \item $C_3$ contains the instances for which $|A_1|\le \floor{q_1(n-1)}+1$ and $|A_2|>\floor{q_2(n-1)}+1$.
\end{enumerate*}
Notice that $\mathcal{A}=C_1\cup C_2 \cup C_3$.
Moreover, every $C_i$ is disjoint from the other, hence, since we are restricting our attention to only the instances in $\mathcal{A}$, we have $P(\vec X\in C_3)=1-P(\vec X \in C_1)-P(\vec X \in C_2)$.
First, we show that both $P(\vec X \in C_1)$ and $P(\vec X \in C_2)$ go to $0$ as the number of agents go to infinity, so that $\lim_{n\to \infty}P(\vec X\in C_3)=1$. 
We then show that $\lim_{n\to \infty}\EE[SC_{\pi,\vec v}(\vec X)|\vec X\in C_3]=W_1(\mu,(1-q_2)\delta_{0}+q_2\delta_{2(1-q_2)})$.
First, let us consider $P(\vec X \in C_1)$.
Let us denote with $\lambda=1-\frac{0.5+q_2}{2}$, then we have
\begin{align*}
    P(\vec X \in C_1)&=P(|A_1|\le \floor{q_1(n-1)}+1,|A_2|\le\floor{q_2(n-1)}+1)\\
    &\le P(|A_2|\le\floor{q_2(n-1)}+1)\le P\Big(X_{\floor{\lambda(n-1)}+1}\le \frac{X_1+X_n}{2}\Big),
\end{align*}
where the last inequality comes from the fact that if $|A_2|\le \floor{q_2(n-1)}+1$ then it must be that $X_{\floor{\lambda(n-1)}+1}\le \frac{X_1+X_n}{2}$, since $q_2<0.5$.
Let us now set $\epsilon>0$ and $r=\floor{\lambda(n-1)}+1$, denote with $A^c$ the complement of the set $A$, we have
\begin{align*}
    P(&\vec X \in C_1) \le P\Big(X_{r}\le \frac{X_1+X_n}{2}\Big)\\
    &=P\Big(X_{r}\le \frac{X_1+X_n}{2}\Big|X_1\le \epsilon, X_n\ge 1-\epsilon\Big)P(X_1\le \epsilon, X_n\ge 1-\epsilon)\\
    &\quad +P\Big(X_{r}\le \frac{X_1+X_n}{2}\Big|\Big\{X_1\ge \epsilon, X_n\le 1-\epsilon\Big\}^c\Big)P\left(\Big\{X_1\ge \epsilon, X_n\le 1-\epsilon\Big\}^c\right)\\
    &\le P\Big(X_{r}\le \frac{1+\epsilon}{2}\Big)\big(P(X_1\le \epsilon)+P( X_n\ge 1-\epsilon)\big) + P(X_1\ge \epsilon)+P( X_n\le 1-\epsilon),
\end{align*}
$P\Big(X_{r}\le \frac{X_1+X_n}{2}\Big|\Big\{X_1\ge \epsilon, X_n\le 1-\epsilon\Big\}^c\Big)\le 1$.
Notice that the cumulative distribution function of the $n$-th order statistic of a uniform distributed random variable is $t\to \Big(F_\mu(t)\Big)^n=t^n$ hence $P( X_n\le 1-\epsilon)=(1-\epsilon)^n$ and $\lim_{n\to \infty}P( X_n\le 1-\epsilon)=0$.
Likewise, we have that $\lim_{n\to \infty} P(X_1\ge \epsilon)=0$.
Lastly, we need to show that $P\Big(X_{r}\le \frac{1+\epsilon}{2}\Big)$ goes to $0$ as $n$ goes to infinity.
It is well-known that, for a suitable constant $K$, it holds $K\sqrt{n}(X_{\floor{\lambda(n-1)}+1}-\lambda)$ converges in distribution to the standard Gaussian $Z\sim\mathcal{N}(0,1)$.
Given $\epsilon<2\lambda-1$, we have that $P\Big(X_{r}\le \frac{1+\epsilon}{2}\Big)=P\Big(K\sqrt{n}(X_{r}-\lambda)\le K\sqrt{n}\Big(\frac{1+\epsilon}{2}-\lambda\Big)\Big)$.
Thus
\begin{align*}
    P\Big(X_{r}\le \frac{1+\epsilon}{2}\Big)&\le \Big|P\Big(X_{r}\le \frac{1+\epsilon}{2}\Big)-P\Big(Z\le K\sqrt{n}\Big(\frac{1+\epsilon}{2}-\lambda\Big)\Big)\Big|\\
    &\quad + P\Big(Z\le K\sqrt{n}\Big(\frac{1+\epsilon}{2}-\lambda\Big)\Big).
\end{align*}
Since $\Big(\frac{1+\epsilon}{2}-\lambda\Big)<0$, we have that $\lim_{n\to\infty}P\Big(Z\le K\sqrt{n}\Big(\frac{1+\epsilon}{2}-\lambda\Big)\Big)=0$.
Finally, by the convergence in distribution of $X_r$, see \cite{wilks2012mathematical}, we have $\lim_{n\to\infty}\Big|P\Big(X_{r}\le \frac{1+\epsilon}{2}\Big)-P\Big(Z\le K\sqrt{n}\Big(\frac{1+\epsilon}{2}-\lambda\Big)\Big)\Big|=0$,
which allows us to conclude that $\lim_{n\to\infty}P(\vec X\in C_1)=0$.

Similarly, we have that $ \lim_{n\to\infty}P(\vec X \in C_2)=0$, hence $\lim_{n\to\infty}P(\vec X\in C_3)=1$.
Lastly, we need to show that $\lim_{n\to \infty}\EE[SC_{\pi,\vec v}(\vec X)|\vec X\in C_3]=W_1(\mu,(1-q_2)\delta_{0}+q_2\delta_{2(1-q_2)})$.
%
Notice that if $\vec x\in C_3$, the EEM places a facility at $x_1$ and the other one at $2x_{n-(\floor{q_2(n-1)}+2)}-x_1$, the first $n-(\floor{q_2(n-1)}+1)$ agents are assigned to the facility placed at $x_1$, while the others to the facility at $2x_{n-\floor{q_2(n-1)}+2}-x_1$.
Since every agent is assigned to its closest facility, we can adapt the argument used to prove Lemma \ref{lemma:equivalenceSC_W1}, and infer that
\[
SC_{EEM}(\vec x)=W_1(\mu_n,q_1\delta_{x_1}+(1-q_1)\delta_{2x_{n-\floor{q_2(n-1)}+2}-x_1}).
\]
By the same argument used in the proof of Theorem \ref{thm:limitBAR}, it is possible to show that 
\begin{align*}
    \lim_{n\to \infty}W_1(\mu_n,q_1\delta_{x_1}&+(1-q_1)\delta_{2x_{\floor{q_1(n-1)}+2}-x_1})=W_1(\mu,(1-q_2)\delta_{0}+q_2\delta_{2(1-q_2)}),
\end{align*}
which concludes the proof.
%


\section{The Optimal Extended Ranking Mechanism}
\label{sec:para}

Owing to Theorem \ref{thm:limitBAR}, the limit Bayesian approximation ratio of any ERM depends on $\mu$, $\vec q$, $\pi$, and $\vec v$.
While $\mu$ and $\vec q$ are specifics of the problem, $\pi$ and $\vec v$ are parameters that the mechanism designer can tune.
In this section, we study how to determine an optimal ERM tailored to $\mu$ and $\vec q$, i.e. a mechanism $\RM$ such that 
\[
\lim_{n\to\infty}B_{ar,1}\big(\RM\big)\le\lim_{n\to\infty}B_{ar,1}\big(\texttt{ERM}^{(\pi',\vec v')}\big)
\]
for any other feasible $\texttt{ERM}^{(\pi',\vec v')}$.
Similarly, we define the optimal ERM with respect to the other objective functions.
%
%

\begin{theorem}
\label{thm:existance_problem}
Given $\mu$, a p.c.v. $\vec q$ and an objective function, there always exist a tuple $(\pi,\vec v)$ whose associated ERM, that is $\RM$.
Moreover, the couple $(\pi,\vec v)$ is a solution to the following minimization problem
\begin{align}
\label{eq:opt:ERM}
    &\min_{\pi\in\mathcal{S}_m}\min_{\vec v\in (0,1)^m} W_p\left(\mu,\sum_{j\in[m]}\eta_j\delta_{F_\mu^{[-1]}(v_j)}\right)\\
    \nonumber\text{s.t.}\;\; & v_{j+1}-v_{j-1}\le q_{\pi(j)}\;\;\text{for} \;\; j=2,\dots,m-1,\;\; v_1\le q_{\pi(1)}, \;\text{and}\;\; v_m\le 1-q_{\pi(m)},\\
    \label{eq:opt:coeff}& \eta_j=F_\mu\left(\frac{F^{[-1]}_\mu(v_{j+1})+F^{[-1]}_\mu(v_{j})}{2}\right)-F_\mu\left(\frac{F^{[-1]}_\mu(v_{j})+F^{[-1]}_\mu(v_{j-1})}{2}\right),
\end{align}
for $ j=2,\dots,m-1$, $\eta_1=F_\mu\Big(\frac{F^{[-1]}_\mu(v_{2})+F^{[-1]}_\mu(v_{1})}{2}\Big)$, $\eta_m=1-\eta_{m-1}$, and
where $p$ depends on the objective to minimize.
\end{theorem}

\begin{proof}
Without loss of generality, we prove the theorem for $p=1$, which correspond to the Social Cost.
First, we show that a solution to problem \eqref{eq:opt:ERM} always exists.
Since for every $m$ there are a finite number of $\pi\in\mathcal{S}_m$, it suffices to show that for every $\sigma\in\mathcal{S}_m$ there exists a solution to problem
\begin{align}
\label{eq:opt:ERM2}
    \min_{\vec v\in (0,1)^m} &W_1\left(\mu,\sum_{j\in[m]}\eta_j(\vec v)\delta_{F_\mu^{[-1]}(v_j)}\right)\\
    \nonumber\text{s.t.}\; &v_{j+1}-v_{j-1}\le q_{\pi(j)}\;\;\text{for} \;\; j=2,\dots,m-1,\;v_1\le v_{\pi(q)}, \;\text{and}\;\; v_m\le 1-v_{\pi(m)},
\end{align}
where $\eta_j(\vec v)$ are the coefficient associated with the permutation $\pi$ defined as in \eqref{eq:opt:coeff}.
Indeed, following the argument used to prove Theorem \ref{thm:limitBAR}, we have that, given $\vec v$, the probability measure $\eta(\vec v)=\sum_{j\in[m]}\eta_j(\vec v)\delta_{F_\mu^{[-1]}(v_j)}$ is a solution to 
\[
W_1(\mu,\eta(\vec v))=\min_{\lambda_j\le q_{\sigma(j)},\sum_{j\in[m]}\lambda_j=1}W_1\left(\mu,\sum_{j\in[m]}\lambda_j\delta_{F_\mu^{[-1]}(v_j)}\right).
\]
Therefore, problem \eqref{eq:opt:ERM} can be reduced to problem \eqref{eq:opt:ERM2}.
First, we consider the case in which $\mu$ has a compact support, hence the minimization problem is well defined over the closed set $[0,1]^m$, since the values $F_\mu^{[-1]}(0)$ and $F_\mu^{[-1]}(1)$ are well-defined.
Notice that $[0,1]^m$ is compact, thus if we show that the objective function of the problem is continuous with respect to $\vec v$ we conclude the proof for this case.
Let $\vec v^{(n)}$ be a succession of feasible vectors that converges to $\vec v$.
From our basic assumptions, we have that both $F_\mu$ and $F_\mu^{[-1]}$ are continuous, hence $\vec v\to \eta_j(\vec v)$ is continuous for every $j\in [m]$.
Likewise, every function $\vec v\to y_j(\vec v):=F_\mu^{[-1]}(v_j)$ is continuous, thus the sequence $\eta_n=\eta(\vec v^{(n)})$ weakly converges to $\eta(\vec v)$ in $\PP(\erre)$.
Since $W_1$ induces the weak topology over $\PP(\erre)$ \cite{villani2009optimal}, we infer $\lim_{n\to \infty}W_1(\mu,\eta(\vec v^{(n)}))=W_1(\mu,\eta(\vec v))$, thus a minimizer always exists.
To conclude the proof, we need to show that the minimizer, namely $\vec v$, is such that $v_j\neq 0,1$ for every $j\in [m]$.
Toward a contradiction, let us assume that $v_1=0$.
Without loss of generality, let us assume that $v_2>0$, thus $y_1< y_2$.
By definition, we have that $F_\mu^{[-1]}(v_1)=a$ and $F_\mu(F_\mu^{[-1]}(v_1))=v_1=0$.
Let us denote with $[a,b]$ the support of $\mu$, then there exists a $\epsilon>0$ such that $a+\epsilon<F_\mu^{[-1]}(v_2)=:y_2$, thus
\begin{align*}
    \mathcal{W}(\vec v)&=\int_{a}^{b}\min_{j\in[m]}|x-F_\mu^{[-1]}(v_j)|d\mu=\int^{\frac{a+y_2}{2}}_a(x-a)d\mu+\sum_{j>1}^m\int_{\frac{a+y_2}{2}}^b\min_{j>1}|x-F_\mu^{[-1]}(v_j)|d\mu\\
    &>\int_{a}^{\frac{a+y_2}{2}}|x-(a+\epsilon)|d\mu+\sum_{j>1}^m\int_{\frac{a+y_2}{2}}^b\min_{j>1}|x-F_\mu^{[-1]}(v_j)|d\mu\\
    &>\int_{a}^{b}\min_{j\in[m]}|x-F_\mu^{[-1]}(v^{(\epsilon)}_j)|d\mu=\mathcal{W}(\vec v^{(\epsilon)}),
\end{align*}
where $\vec v^{(\epsilon)}=(F_\mu(a+\epsilon),v_2,\dots,v_m)$, which contradicts the optimality of $\vec v$.
Let us now consider the case in which $\mu$ has unbounded support.
Let $\vec v^{(n)}$ be a minimizing sequence.
Up to a sub-sequence, we can assume that the $\vec v^{(n)}$ converges to $\vec v$ in $[0,1]^m$.
We now show that $\vec v\in (0,1)^m$.
First, let us assume that there exists a $j\in[m]$ for which $v_j\in (0,1)$.
Without loss of generality, let us assume that $v_1=0$ and that $0<v_2<1$ (if there are multiple entries of $\vec v$ that are equal to $0$ the study is the same).
Let us now consider $\vec w^{(n)}$ be such that $w^{(n)}_1=v_1^{(n)}$ and $w_j^{(n)}=v_j$ for $j>1$.
Since $\vec w^{(n)}$ and $\vec v^{(n)}$ have the same limit and $\mathcal{W}$ is continuous, we have that $\vec w^{(n)}$ is a minimizing sequence.
Since $w^{(n)}_1\to 0$, we have that $F_\mu^{[-1]}(w_1^{(n)})\to -\infty$, thus
\begin{align*}
    \mathcal{W}(\vec w^{(n)})=\int_{-\infty}^{\frac{y_1^{(n)}+y_2}{2}}|x-F_\mu^{[-1]}(w_1^{(n)})|d\mu+\int_{\frac{y_1^{(n)}+y_2}{2}}^{+\infty}\min_{j>1}|x-F_\mu^{[-1]}(v_j)|d\mu.
\end{align*}
Since $F_\mu^{[-1]}(w_1^{(n)})\to -\infty$, we have that $\mathcal{W}(\vec v)=\int_{-\infty}^{+\infty}\min_{j>1}|x-F_\mu^{[-1]}(v_j)|d\mu$.
Let us consider the vector $\vec v^{(\epsilon)}=(\epsilon,v_2,\dots,v_m)$, where $v_2>\epsilon>0$, then it holds
\begin{align*}
    \mathcal{W}(\vec v)&=\int_{-\infty}^{+\infty}\min_{j>1}|x-F_\mu^{[-1]}(v_j)|d\mu > \int_{-\infty}^{+\infty}\min_{j\in[m]}|x-F_\mu^{[-1]}(v^{(\epsilon)}_j)|d\mu=\mathcal{W}(\vec v^{(\epsilon)}),
\end{align*}
which contradicts the optimality of $\vec v$.
Similarly, we deal with the case $\vec v\in\{0,1\}^m$ and the case in which the support of $\mu$ is unbounded only from one side.
\end{proof}

Notice that the limit Bayesian approximation ratio of the optimal ERM does not necessarily converge to $1$, as the next example shows.

\begin{example}
\label{ex:impox_aoerm}
    Let $\mu$ be the uniform distribution $\mathcal{U}[0,1]$ and let $\vec q=(0.8,0.4)$ be a p.c.v..
    Since $\mu$ is symmetric, one of the solutions to problem \eqref{eq:min_proj} is $\nu_2=(0.4)\delta_{0.2}+(0.6)\delta_{0.7}$ (the other one is $\nu_2'=(0.6)\delta_{0.3}+(0.4)\delta_{0.8}$).
    However, by Corollary \ref{crr:sparecapacity}, there does not exist a permutation $\pi\in\mathcal{S}_m$ such that $\mathtt{ERM}^{((0.2,0.7), \pi)}$ is feasible.
    Thus, by Theorem \ref{thm:limitBAR}, no feasible ERM is such that $\lim_{n\to\infty}B_{ar,1}(\RM)=1$.
\end{example}

We now characterize the optimal ERM in two specific cases.
In the first one, the total capacity of the facilities is the same as the number of agents.
In the second one, we need to place two capacitated facilities and $\mu$ is symmetric.

\subsection{The No-Spare Capacity Case}
\label{sec:case2}

In the no-spare capacity case the total capacity of the facilities matches the number of agents, thus $\sum_{j\in[m]}q_j=1$.
Due to Corollary \ref{crr:sparecapacity}, feasible ERMs are induced by vectors $\vec v$ such that $v_j=p\in[0,1]$ for every $j\in[m]$.
This property allows to characterize the optimal ERM with respect to the Social, Maximum, and $l_2$ costs.
%
%

\begin{theorem}
\label{crr:nsc_best_RM}
    In the no-spare capacity case, the optimal ERM is unique.
    Moreover, we have that \begin{enumerate*}[label=(\roman*)]
        \item the $\texttt{all-median}$ mechanism, which places all the facilities at the median agent, is optimal with respect to the Social Cost,
        \item the $\RM$ is optimal with respect to the $l_2$ cost if $\vec v=(F_\mu(m),\dots,F_\mu(m))$ where $m$ is the mean of $\mu$, and
        \item if $\mu$ has compact support, then $\RM$ is optimal with respect to the Maximum Cost if $\vec v=(F_\mu(\frac{a+b}{2}),\dots,F_\mu(\frac{a+b}{2}))$ where $a=\inf\{x\in spt(\mu)\}$ and $b=\sup\{x\in spt(\mu)\}$.
        \end{enumerate*}
\end{theorem}

\begin{proof}
    Owing to Corollary \ref{crr:sparecapacity}, we have that any truthful ERM, namely $\RM$, must be such that $\vec v=(p,p,\dots,p)$ where $p\in(0,1)$.
    Thus, to complete the proof it is sufficient to prove that the function $M:y\to\int_{\erre}|x-y|d\mu$ is decreasing on $(-\infty,med(\mu))$ and increasing on $(med(\mu),+\infty)$.
    Indeed, let us assume that $y\in(-\infty,med(\mu))$ and let us consider $\lim_{h\to 0^+}\frac{\int_{\erre}|x-(y+h)|d\mu-\int_\erre |x-y|d\mu}{h}$.
    By the linearity of the integral, we have that $\int_\erre |x-y|d\mu=\int_{-\infty}^y (y-x)d\mu+\int_y^{+\infty}(x-y)d\mu$, similarly $\int_\erre |x-(y+h)|d\mu=\int_{-\infty}^{y+h} (y+h-x)d\mu+\int_{y+h}^{+\infty}(x-y-h)d\mu$, so that
    \begin{align*}
        \int_{\erre}|x-(y+h)|d\mu-\int_\erre |x-y|d\mu &=2\int_{y}^{y+h}(y-x)d\mu+\int_{-\infty}^yhd\mu-\int_{y}^{+\infty}hd\mu\\
        &\le h^2+\int_{-\infty}^yhd\mu -\int_{y}^{+\infty}hd\mu,
    \end{align*}
    thus $\partial_y M(y)=F_\mu(y)-(1-F_\mu(y))$.
    Since $y<med(\mu)$, we infer that $\partial_y M(y)<0$.
    Similarly, it holds $\partial_y M(y)>0$ if $y>med(\mu)$, which concludes the proof for the Social Cost.
    The results for the $l_2$ costs and the Maximum Cost follow by a similar argument.
\end{proof}

%
%
%
%
%
%

%
%

%
When the agents are distributed following a uniform distribution, that is $\mu=\mathcal{U}[0,1]$, we can express the limit Bayesian approximation ratio of the $\texttt{all-median}$ mechanism as a function of $\vec q$.
Indeed, since $\sum_{j\in[m]}q_j=1$, any solution to problem \eqref{eq:min_proj} separates the interval $[0,1]$ into $m$ intervals whose length is $q_j$.
Notice that the order in which $[0,1]$ is divided is irrelevant.
Furthermore, the facility is placed at the median of such interval, thus the objective value of \eqref{eq:min_proj} is $\frac{1}{4}\sum_{j\in[m]}q_j^2$.
By Theorem \ref{thm:limitBAR}, the limit Bayesian approximation ratio of the \texttt{all-median} mechanism is $\lim_{n\to\infty}B_{ar,1}(\texttt{all-median})=(\sum_{j\in[m]} q_j^2)^{-1}$, since the asymptotic cost of placing all the facilities at the median point is $\int_0^1|x-0.5|dx=\frac{1}{4}$.
The limit Bayesian approximation ratio gets closer to $1$ as the values of $\vec q$ become concentrated at one index.
Conversely, if all the facilities have the same capacity $\frac{1}{m}$, the Bayesian approximation ratio of the \texttt{all-median} becomes the largest possible, that is $\lim_{n\to\infty}B_{ar,1}(\RM)=m$.

\subsection{The optimal ERM for two facilities and symmetric distributions -- The Social Cost case}
\label{section:tobeextended}

In this section, we retrieve the optimal ERMs with respect to the Social Cost when $m=2$ and $\mu$ is symmetric.
Since $\mu$ is symmetric, the asymptotic expected Social Cost of any feasible $\RM$ does not depend on $\pi\in\mathcal{S}_2$.
For this reason, we fix $\pi=Id$ and search for the vector $\vec v=(v_1,v_2)$ that induces the optimal ERM.
%
%
First, we rewrite the objective value of problem \eqref{eq:opt:ERM} in terms of $y_i=F_\mu^{[-1]}(v_i)$ for $i=1,2$.
Since $F_\mu^{[-1]}$ is monotone, we have that $y_1\le y_2$.
We then set
\begin{equation}
    \mathcal{W}(y_1,y_2)=W_1\bigg(\mu,F_\mu\Big(\frac{y_1+y_2}{2}\Big)\delta_{y_1}+\Big(1-F_\mu\Big(\frac{y_1+y_2}{2}\Big)\Big)\delta_{y_2}\bigg).    
\end{equation}
Since both $F_\mu$ and $F_\mu^{[-1]}$ are bijective, any $\vec y=(y_1,y_2)$ identifies a unique $\vec v$ and vice-versa.
Using the properties of the Wasserstein distances on the line \cite{santambrogio2015optimal}, we have $\mathcal{W}(y_1,y_2)=\int_{-\infty}^{\frac{y_1+y_2}{2}}|x-y_1|\rho_\mu(x)dx+\int_{\frac{y_1+y_2}{2}}^{+\infty}|x-y_2|\rho_\mu(x)dx$.  


\begin{theorem}
\label{thm:bestpercsym}
    Let $\mu$ be a symmetric probability measure.
    Then, $\mathcal{W}$ is differentiable and its gradient is defined by the formula
    \[
    \nabla \mathcal{W}(y_1,y_2)=\Big(2F_\mu(y_1)-F_\mu\Big(\frac{y_1+y_2}{2}\Big),2F_\mu(y_2)-1-F_\mu\Big(\frac{y_1+y_2}{2}\Big)\Big).
    \]
    Furthermore, we have $\nabla \mathcal{W}(F_\mu^{[-1]}(0.25),F_\mu^{[-1]}(0.75))=(0,0)$.
\end{theorem}



\begin{proof}
    %
    Let $\vec q$ be a p.c.v., due to the symmetry of $\mu$, we assume that $y_1\le y_2$ are the positions of the facility and that $y_1$ has capacity $q_1$ and $y_2$ has capacity $q_2$.
    We then define
    \begin{align*}
        \mathcal{W}(y_1,y_2)=&\int_{-\infty}^{y_1}(y_1-x)\rho_\mu(x)dx+\int_{y_1}^{\frac{y_1+y_2}{2}}(x-y_1)\rho_\mu(x)dx\\
        &+\int_{\frac{y_1+y_2}{2}}^{y_2}(y_2-x)\rho_\mu(x)dx+\int_{y_2}^{+\infty}(x-y_2)\rho_\mu(x)dx,
    \end{align*}
    thus we have $\mathcal{W}(y_1,y_2)=A(y_1,y_2)+B(y_1,y_2)+C(y_1,y_2)+D(y_1,y_2)$, where
    \begin{align*}
        A(y_1,y_2)&=\int_{-\infty}^{y_1}(y_1-x)\rho_\mu(x)dx,\quad
        &B(y_1,y_2)=\int_{y_1}^{\frac{y_1+y_2}{2}}(x-y_1)\rho_\mu(x)dx,\\
        C(y_1,y_2)&=\int_{\frac{y_1+y_2}{2}}^{y_2}(y_2-x)\rho_\mu(x)dx,\quad
        &D(y_1,y_2)=\int_{y_2}^{+\infty}(x-y_2)\rho_\mu(x)dx.
    \end{align*}
    We now compute the derivative of $\mathcal{W}$ with respect to $y_1$ and $y_2$.
    Let us consider the derivative $\partial_{y_1}\mathcal{W}(y_1,y_2)$.
    Since $\partial_{y_1}D(y_1,y_2)=0$, we have $\partial_{y_1}\mathcal{W}(y_1,y_2)=\partial_{y_1}A(y_1,y_2)+\partial_{y_1}B(y_1,y_2)+\partial_{y_1}C(y_1,y_2)$.
    First, let us consider $\partial_{y_1}A(y_1,y_2)$.
    We have that
    \begin{align*}
        \frac{1}{h}\Big(\int_{-\infty}^{y_1+h}(y_1+h-x)\rho_\mu(x)dx&-\int_{-\infty}^{y_1}(y_1-x)\rho_\mu(x)dx\Big)\\
        &=\frac{1}{h}\Big(h\int_{-\infty}^{y_1+h}\rho_\mu(x)dx-\int_{y_1}^{y_1+h}(y_1-x)\rho_\mu(x)dx\Big)\\
        &=\int_{-\infty}^{y_1+h}\rho_\mu(x)dx+\frac{1}{h}\int_{y_1}^{y_1+h}(y_1-x)\rho_\mu(x)dx.
    \end{align*}
    By the Lebesgue's theorem, $\lim_{h\to 0}\frac{1}{h}\int_{y_1}^{y_1+h}(y_1-x)\rho_\mu(x)dx=0$, thus $\partial_{y_1}A(y_1,y_2)=F_\mu(y_1)$.
    
    %

    %
    Second, let us compute $\partial_{y_1}B(y_1,y_2)$.
    We have that
    \begin{align*}
        \frac{1}{h}\Big(&\int_{y_1+h}^{\frac{y_1+y_2}{2}+\frac{h}{2}}(x-y_1-h)\rho_\mu(x)dx-\int_{y_1}^{\frac{y_1+y_2}{2}}(x-y_1)\rho_\mu(x)dx\Big)\\
        =&\frac{1}{h}\Big(-h\int_{y_1+h}^{\frac{y_1+y_2}{2}+\frac{h}{2}}\rho_\mu(x)dx+\int_{\frac{y_1+y_2}{2}}^{\frac{y_1+y_2}{2}+\frac{h}{2}}(x-y_1)\rho_\mu(x)dx-\int_{y_1}^{y_1+h}(x-y_1)\rho_\mu(x)dx \Big).
    \end{align*}
    Owing again to Lebesgue's theorem, we have $\partial_{y_1}B(y_1,y_2)=-\Big(F_\mu\big(\frac{y_1+y_2}{2}\big)-F_\mu(y_1)\Big)+\frac{(y_2-y_1)}{4}\rho_\mu\Big(\frac{y_1+y_2}{2}\Big)$.
    By a similar computation, we have that $\partial_{y_1}C(y_1,y_2)=-\frac{y_2-y_1}{4}\rho_\mu\Big(\frac{y_1+y_2}{2}\Big)$.
    Putting all together, we infer $\partial_{y_1}\mathcal{W}(y_1,y_2)=2F_\mu(y_1)-F_\mu\Big(\frac{y_1+y_2}{2}\Big)$.
    Similarly, we have that $\partial_{y_2}\mathcal{W}(y_1,y_2)=2F_\mu(y_2)-1-F_\mu\Big(\frac{y_1+y_2}{2}\Big)$,
    %
    which concludes the first part of the proof.
    %
    %
    Finally, due to the properties of cumulative distribution functions of symmetric measures, we infer that $\nabla \mathcal{W}\big(F_\mu^{[-1]}(0.25),F_\mu^{[-1]}(0.75)\big)=(0,0)$.
\end{proof}

Now that we have an explicit formula for the gradient of $\mathcal{W}$, we need to express the set of feasible $\vec y$.
From Theorem \ref{thm:RM_feasible}, we have that $v_1\le q_1$ and $1-v_2\le q_2$, i.e. $1-q_2\le v_2$.
Since $F_\mu^{[-1]}$ is monotone, we have that $y_2=F_\mu^{[-1]}(v_2)\le F_\mu^{[-1]}(q_1)$ and that $F_\mu^{[-1]}(1-q_2)\le y_1$.
Therefore we have that the set of feasible $\vec y$ lays in a triangle, namely $T(\vec q)$, whose vertexes are $(F_\mu^{[-1]}(1-q_2),F_\mu^{[-1]}(q_1))$, $(F_\mu^{[-1]}(q_1),F_\mu^{[-1]}(q_1))$, and $(F_\mu^{[-1]}(1-q_2),F_\mu^{[-1]}(1-q_2))$.
To exemplify this procedure, we carry out the computation for the Uniform Distribution.

\paragraph{The uniform distribution.}
Theorem \ref{thm:bestpercsym} allows us to fully characterize the optimal ERM with respect to the Social Cost when $\mu$ is a uniform distribution.

\begin{theorem}
\label{thm:bestuniform}
    Let $\vec q$ be a p.c.v. such that $q_2\le q_1$, then the optimal ERM for a uniform distribution is unique and induced by $(Id,\vec v)$ where
    \begin{enumerate*}[label=(\roman*)]
        \item $\vec v=(0.25,0.75)$ if $q_2\ge 0.75$, 
        \item $\vec v=(1-q_2,q_1)$ if $3q_2-2-q_1\le 0$, and
        \item $\vec v=(1-q_2,1-\frac{q_2}{3})$ otherwise.
    \end{enumerate*}
\end{theorem}

\begin{proof}
    %
    Notice that $\mathcal{W}(\vec v)=\int_0^1 \min\{|x-v_1|,|x-v_2|\}dx$, that is $\mathcal{W}(\vec v)=\frac{v_1^2}{2}+\frac{(1-v_2)^2}{2}+\frac{(v_2-v_1)^2}{4}$, hence $\nabla \mathcal{W}(\vec v)=\frac{1}{2}\big(3v_1-v_2,3v_2-2-v_1\big)$.
    %
    %
    Moreover, a simple computation allows us to compute the Hessian of $\mathcal{W}$, that is
    \[
    H\mathcal{W}(\vec v)=\frac{1}{2}\begin{pmatrix}
    3 & -1 \\
    -1 & 3 
    \end{pmatrix}.
    \]
    Since $H\mathcal{W}$ is definite positive, we have that if $q_2\ge 0.75$, then $(0.25,0.75)$ is the minimum.
    We now consider the other cases, i.e. $q_2\le 0.75$.
    %
    %
    Notice that the set of points on which the first derivative of $\mathcal{W}$ nullifies is the line $v_2=3v_1$, while the set of points on which the second derivative of $\mathcal{W}$ nullifies is $v_2=\frac{1}{3}(v_1+2)$, thus $\nabla \mathcal{W}(\vec v)=(0,0)$ if and only if $\vec v=(0.25,0.75)$.
    Therefore, the solution must lay on the border of the feasible $\vec v$, which is the triangle, namely $T(\vec q)$, whose vertexes are $(q_1,q_1)$, $(1-q_2,1-q_2)$, and $(1-q_2,q_1)$.
    If $3q_2-2-q_1\le 0$, we have that $\partial_{y_1} \mathcal{W}(\vec v)>0$ and $\partial_{y_2} \mathcal{W}(\vec v)<0$.
    Thus the gradient cannot be perpendicular to either $\{(1-q_2,t)\;\text{s.t.}\;t\in[1-q_2,q_1]\}$ and $\{(t,q_1)\;\text{s.t.}\;t\in[1-q_2,q_1]\}$, hence the minimum is at the vertex $(1-q_2,q_1)$, which concludes the proof for this case.
    In the remaining cases, we have that the boundary of $T(\vec q)$ intersect the line on which $\partial_{y_2}\mathcal{W}=0$.
    In particular, $\{\partial_{y_2}\mathcal{W}=0\}$ intersects with the edge $\{(1-q_2,t)\;\text{s.t.}\;t\in[1-q_2,q_1]\}$ at the point $(1-q_2,1-\frac{q_2}{3})$.
    On such point, we have that the gradient is perpendicular to the boundary of $T(\vec q)$ and points inward, thus is a relative minimum.
    To conclude, we show that this is the only possible relative minimum.
    Indeed, it is easy to see that if the boundary of $T(\vec q)$ intersect the line $v_2=3v_1$, the gradient points outward due to the sign of $\nabla \mathcal{W}$.
    %
\end{proof}

\begin{example}[Continuing Example \ref{ex:impox_aoerm}]
Owing to Theorem \ref{thm:bestuniform}, we have that the best possible ERM in the instance described in Example \ref{ex:impox_aoerm} is induced by $Id$ and $\vec v=(0.6,0.8)$.
The limit Bayesian approximation ratio of the optimal mechanism is therefore $\sim 1.62$.
This result is extremely interesting if compared to the expected performances of the EEM.
Indeed, owing to the argument presented in Section \ref{ex:im}, we have that the EEM will either 
\begin{enumerate*}[label=(\roman*)]
    \item place the facility with capacity $q_2$ at $y_1=x_1$ and the other facility at $y_2=2x_{n-(\floor{q_1(n-1)}+1)}-x_1$ with probability $0.5$ or
    \item place the facility with capacity $q_2$ at $y_2=x_n$ and the other facility at $y_1=2x_{\floor{q_1(n-1)}+2}-x_n$ with probability $0.5$.
\end{enumerate*} 
Since $\mu=\mathcal{U}[0,1]$, we have that $x_1\to 0$, $x_n\to 1$, $x_{n-\floor{q_2(n-1)}+2}\to F_\mu^{[-1]}(1-q_2)=1-q_2$, and $x_{\floor{q_2(n-1)}+1}\to F_\mu^{[-1]}(q_2)=q_2$.
In particular, the limit expected Social Cost of the mechanism is then $\frac{1}{2}W_1(\mu,q_2\delta_{2q_2-1}+(1-q_2)\delta_1)+\frac{1}{2}W_1(\mu,(1-q_2)\delta_{0}+q_2\delta_{2-2q_2})$.
Since $\mu$ is symmetric with respect to $\frac{1}{2}$, we have that $W_1(\mu,q_2\delta_{2q_2-1}+(1-q_2)\delta_1)=W_1(\mu,(1-q_2)\delta_{0}+q_2\delta_{2-2q_2})=\frac{(1-q_2)^2}{2}+\frac{q_2(2-3q_2)}{2}=0.34$, hence the limit Bayesian approximation ratio of the EEM is $\sim 2.62>1.62$, that is the best ERM outperforms the EEM when we the agents are sampled from a uniform distribution.
\end{example}

\subsection{The optimal ERM for two facilities and symmetric distributions -- The Maximum Cost case}

First, we consider the case in which $p=\infty$.
In this case, given $\mu$ and $\vec q$, the solution to \eqref{eq:opt:ERM} can be built in a finite number of step.
In particular, to find a solution to \eqref{eq:opt:ERM} we need to consider the following problem
\begin{equation}
    \label{eq:optERM_maxcost}
    \min_{\pi\in \{Id,\sigma\}}\min_{\zeta\in\PP_{\pi,\vec q}(\erre)}W_\infty(\mu,\zeta),
\end{equation}
where $Id$ is the identity permutation and $\sigma$ is the permutation that swaps the two elements of $[2]$.
Every solution to \eqref{eq:optERM_maxcost} divides the support of $\mu$ into two separate and disjoint sets, namely $[a,z]$ and $(z,b]$.
Given the partition $[a,z]$ and $(z,b]$, the objective value of \eqref{eq:optERM_maxcost} is equal to $\max\{\frac{z-a}{2},\frac{b-z}{2}\}$.
Due to the capacity constraints, we must have that $\mu([a,z])\le q_{\pi(1)}$ and $\mu((z,b])\le q_{\pi(2)}$ for a suitable $\pi\in\mathcal{S}_2$.
%
%
In order to detail how to retrieve the optimal solution to \eqref{eq:optERM_maxcost}, we introduce the projection operator.
Given a compact set $K\subset \erre$, we define the projection operator $pr_{K}:\erre\to K$ as $pr_K(x)=\argmin_{y\in K}|x-y|$.
Since $K$ is compact, $pr_K$ is well-defined.
Using the projection operator, we define a routine that computes the solution to problem \eqref{eq:optERM_maxcost}, which we report in Algorithm \ref{alg:optinfty}.

\begin{algorithm}[t]
    \caption{Routine to solve problem \eqref{eq:optERM_maxcost}.}
    \label{alg:optinfty}
    \begin{algorithmic}[1]
        \State \textbf{Input:} A probability measure $\mu$ whose support is $[a,b]$; a p.c.v. $\vec q$
        \State \textbf{Output:} A solution to problem \eqref{eq:optERM_maxcost}.
        
        \State $z_0=\frac{a+b}{2}$, $K_{Id}:=[F_\mu^{-1}(1-q_2),F_\mu^{-1}(q_1)]$, and $K_{\sigma}:=[F_\mu^{-1}(1-q_1),F_\mu^{-1}(q_2)]$;

        \State $z_{Id}=pr_{K_{Id}}(z_0)$ and $z_{\sigma}=pr_{K_{\sigma}}(z_0)$;

        \If{$|z_{Id}-z_0|\le |z_{\sigma}-z_0|$}
            \State $\alpha=\mu([a,z_{Id}])$;
            \State $y_1=\frac{z_{Id}}{2}$, $y_2=\frac{z_{Id}+b}{2}$, $\pi=Id$;     
        \Else
            \State $\alpha=\mu([a,z_\sigma])$;
            \State $y_1=\frac{z_{\sigma}}{2}$, $y_2=\frac{z_{\sigma}+b}{2}$, $\pi=\sigma$;
        \EndIf
        \State \textbf{Return:} $\nu=\alpha\delta_{y_1}+(1-\alpha)\delta_{y_2}$ and $\pi$
        
    \end{algorithmic}
\end{algorithm}


\begin{theorem}
    The output of Algorithm \ref{alg:optinfty} is a solution to problem \eqref{eq:optERM_maxcost}.
\end{theorem}

\begin{proof}
    Toward a contradiction, let $\nu'$ be such that $W_\infty(\mu,\nu')<W_\infty(\mu,\nu)$, where $\nu$ is the probability measure returned by Algorithm \ref{alg:optinfty}.
    Without loss of generality, let us assume that $\pi=Id$ is the optimal permutation associated with $\nu'$.
    Since $\nu'$ is optimal, the optimal transportation plan between $\mu$ and $\nu'$ splits $[a,b]$ into two sets of the form $[a,z']$ and $(z',b]$, with $z'\in[F_\mu^{-1}(1-q_2),F_\mu^{[-1]}(q_1)]$.
    By the definition of $z$, we have that $|z-\frac{a+b}{2}|\le |z'-\frac{a+b}{2}|$, we then have that
    \[
    W_\infty(\mu,\nu')=\max\Big\{\frac{b-z'}{2},\frac{z'-a}{2}\Big\}\ge \max\Big\{\frac{b-z}{2},\frac{z-a}{2}\Big\}=W_\infty(\mu,\nu),
    \]
    which is a contradiction, thus Algorithm \ref{alg:optinfty} returns an optimal solution.
\end{proof}

Once we obtain the optimal solution to problem \eqref{eq:optERM_maxcost}, we can retrieve the optimal ERM associated with $\mu$ and $\vec q$ by performing another projection.

\begin{theorem}
\label{thm:besterm_max}
    Let $(y_1,y_2)$ be the support of $\nu$, solution to problem \eqref{eq:optERM_maxcost}.
    Let us define $K_{Id}:=[F_\mu^{[-1]}(1-q_2),F_\mu^{[-1]}(q_1)]$ and $K_\sigma=[F_\mu^{[-1]}(1-q_1),F_\mu^{[-1]}(q_2)]$.
    We then define $s^{(Id)}=(s^{(Id)}_1,s^{(Id)}_2)$ and $s^{(\sigma)}=(s_1^{(\sigma)},s_2^{(\sigma)})$, where $s_i^{(\pi)}=pr_{K_\pi}(y_i)$ for $i=1,2$ and $\pi=Id,\sigma$.
    Finally, the optimal ERM for the distribution $\mu$ and with respect to the Maximum Cost is induced by $(\vec p,\pi)$, where 
    \begin{enumerate*}[label=(\roman*)]
        \item $p_1=F_\mu(s^{(Id)}_1)$, $p_2=F_\mu(s^{(Id)}_2)$, and $\pi=Id$, if $\max\{|a-s^{(Id)}_1|,|b-s^{(Id)}_2|\}\le\max\{|a-s^{(\sigma)}_1|,|b-s^{(\sigma)}_2|\}$, or
        \item $p_1=F_\mu(s^{(\sigma)}_1)$, $p_2=F_\mu(s^{(\sigma)}_2)$, and $\pi=\sigma$ otherwise.
    \end{enumerate*}
\end{theorem}

\begin{proof}
    To prove the Theorem, we need to enforce the truthfulness constraints in problem \eqref{eq:optERM_maxcost}.
    Therefore, either the facilities are located in $K_{Id}$ or $K_{\sigma}$.
    Let us consider $K_{Id}$ first.
    It is easy to see that the best position to place the facilities are the points in $K_{Id}$ that are as close as possible to $y_1$ and $y_2$.
    Thus the facilities must be placed at the projection of $y_1$ and $y_2$ on $K_{Id}$.
    By a similar argument, we infer the same conclusion for $K_{\sigma}$.
    To conclude, it suffice to compare the two projections and select the one that has the lowest Maximum Cost.
\end{proof}

\subsection{Optimal Extended Ranking Mechanisms in the Generic Framework}

Lastly, we propose a sufficient and necessary condition that guarantees the existence of an optimal ERM whose limit Bayesian approximation ratio converges to $1$. 
Noticeably this criteria is applicable to every cost function, every $\mu\in\PP(\erre)$, and every $m\in\enne$.

\begin{theorem}
\label{thm:suff_condition_asymptoticopt}
    Given $\mu$ and $\vec q$, let $(\pi,\nu_{m})$ be a solution to Problem \eqref{eq:min_proj}.
    We denote with $\{y_j\}_{j\in[m]}$ the points in the support of $\nu_{m}$ ordered non-decreasingly.
    Moreover, we set $y_0=-\infty$ and $y_{m+1}=+\infty$.
    Let $\zeta_j=F_\mu(y_{j+1})-F_\mu(y_{j-1})$ and let $\hat \pi\in\mathcal{S}_m$ be a permutation that orders the values $\zeta_j$ in a non-increasing order.
    Then, the optimal ERM has limit Bayesian approximation ratio equal to $1$ if and only if $\zeta_{\hat \pi(j)}\le q_j$ for every $j\in [m]$.
\end{theorem}

\begin{proof}
    Let $(\sigma,\nu_{m})$ be a solution to \eqref{eq:min_proj}.
    %
    %
    Let us define $v_j=F_\mu(y_j)$ and let $\gamma=\hat \pi^{-1}\in\mathcal{S}_m$.
    By definition of $v_j$, we have that $v_{j+1}-v_{j-1}=F_\mu(y_{j+1})-F_\mu(y_{j-1})\le q_{\gamma(j)}$ for every $j\in[m]$, hence $\texttt{ERM}^{(\gamma,\vec v)}$ is feasible (Theorem \ref{thm:RM_feasible}).
    Owing to Theorem \ref{thm:limitBAR}, the asymptotic expected Social Cost of $\texttt{ERM}^{(\gamma,\vec v)}$ converges to $W_1(\mu,\nu_{\vec v})$, where $\nu_{\vec v}$ is a solution to 
    \begin{equation}
    \label{eq:probl_+contrst}
      \min_{\lambda_j\le q_{\gamma(j)},\sum_{j\in[m]\lambda_j=1}}W_1\Big(\mu,\sum_{j\in[m]}\lambda_j\delta_{y_j}\Big)  
    \end{equation}
    since $v_j=F_\mu(y_j)$.
    To conclude, we need to show that $(\gamma,\nu_{\vec v})$ is a solution to \eqref{eq:min_proj}.
    Since $\sigma\in\mathcal{S}_m$ is optimal, it suffices to show that $\min_{\zeta\in\PP_{\sigma,m}(\erre)}W_1(\mu,\zeta)\ge W_1(\mu,\nu_{\vec v})$.
    %
    First, let us consider the following minimization problem 
    \begin{equation}
        \label{eq:unconstr_pr}
        \min_{\sum_{j\in [m]}\lambda_j=1}W_1\Big(\mu,\sum_{j\in[m]}\lambda_j\delta_{y_j}\Big),
    \end{equation} 
    %
    %
    and let us set $\nu=\sum_{j\in[m]}(F_\mu(z_j)-F_\mu(z_{j-1}))\delta_{y_j}$, where $z_j=\frac{y_{j}+y_{j+1}}{2}$ for every $j\in[m-1]$, $z_0=-\infty$, and $z_m=+\infty$.
    Since the optimal transportation plan between two measures on a line is monotone \cite{villani2009optimal}, we infer $W_1(\mu,\nu)=\int_{\erre\times \erre}\min_{j\in[m]}|x-y_j|\,d\mu$,
    thus $\nu$ is a minimizer of problem \eqref{eq:unconstr_pr}.
    Finally, since $F_\mu$ is monotone increasing, it holds $F_\mu(z_j)-F_\mu(z_{j-1})\le F_\mu(y_{j+1})-F_\mu(y_{j-1})\le q_{\gamma(j)}$.
    In particular, $\nu$ is feasible for problem \eqref{eq:probl_+contrst}, thus $W_1(\mu,\nu)=W_1(\mu,\nu_{\vec v})$.
    Likewise, $\nu_m$ is feasible for problem \eqref{eq:unconstr_pr}, thus $W_1(\mu,\nu_{\vec v})=W_1(\mu,\nu)\le W_1(\mu,\nu_{m})$, which concludes the proof.
\end{proof}

\section{Numerical Experiments}
\label{sec:experiment}

In this section, we complement our theoretical study of the ERMs by running several numerical experiments.
Indeed, most of our results pertain to the limit analysis of the mechanism. For this reason, we want to test two aspects of the ERMs.
First, we want to compare the Bayesian approximation ratio of the ERMs with the Bayesian approximation ratios of other truthful mechanisms when the number of agents $n$ is small.
Since the Ranking Mechanisms are a subset of the ERMs, we only consider the EEM and, when possible, the IG.
Second, we want to evaluate the convergence speed of the Bayesian approximation ratio of $\RM$.
More specifically, we want to assess how close the Bayesian approximation ratio of an ERM and the limit detected in Theorem \ref{thm:limitBAR} are when the number of agents is small. 
We run our experiments for different distributions $\mu$ and percentage capacity vector $\vec{q}$. 
All the experiments are performed in Matlab 2023a on macOS Monterey system with Apple M1 Pro CPU and 16GB RAM.
The code is available on \url{https://anonymous.4open.science/r/Bayesian-CFLP-F585/}. 

\paragraph{Experiment setup}
We consider the Social, $l_2$, and Maximum Cost. 
For social and $\ell_2$ costs, we sample the agents' positions from three probability distributions: the uniform distribution $\mathcal{U}[0,1]$, the standard normal distribution $\mathcal{N}(0,1)$, and the exponential distribution $\text{Exp}(1)$.
For Maximum Cost, we must test the mechanisms on measures with compact support (see Theorem \ref{thm:limitBAR_p_MC}).
We therefore consider the uniform distribution and an asymmetric one $\text{Beta}(3,1)$.
Owing to Corollary \ref{crr:scale_invariance}, we do not consider other parameter choices since testing the mechanisms over the standard Gaussian distribution is the same as testing over $\mathcal{N}(\mathfrak{m},\sigma^2)$. 
To compare the ERMs to other mechanisms, we limit our tests to cases in which $m=2$, as all the known truthful mechanisms different from the Ranking Mechanisms operate only under this restriction.
We consider different percentage capacity vectors $\vec q\in (0,1)^2$. 
Specifically, we consider balanced capacities $\vec{q}=(q,q)$ and unbalanced capacities $\vec{q}=(q_1, q_2), q_1 \neq q_2$. For the case of balanced capacities, we set $q=0.7,0.8$ and $0.9$. 
For the case of unbalanced capacities, we consider the slightly unbalanced capacities i.e. $\vec{q}=(0.85,0.75)$, and extremely unbalanced capacities i.e. $\vec{q}=(0.8,0.4)$, $(0.85,0.35)$.
As benchmark mechanisms, we consider the EEM \cite{aziz2020facility} and, when possible, the IG \cite{ijcai2022p75}.

\subsection{Experiment results -- The Social and $l_2$ costs}

In this section, we report our numerical results for the Social and $l_2$ costs.
As we will see, the problems entailed to the Social and $l_2$ costs behave similarly for all the cases under study.
Since this is not the case for the Maximum Cost, we report and comment the results for this class of problems into a dedicated section.

\paragraph{Comparison with the EEM and the IG.} 
We first consider the case of balanced capacities $\vec q=(q,q)$, in which we compare the Bayesian approximation ratio of three different mechanisms: the EEM, the IG, and the ERMs.
Regardless of the distribution, we consider the optimal ERM with respect to the uniform distribution, i.e. $\RM$, with $\pi=Id$ and $\vec v=(\max\{0.25,1-q\},\min\{0.75,q\})$.
Figure~\ref{fig:EQ} (and Table~\ref{table:equicapacitated_test} in Appendix~\ref{app:social}) show the average and the $95\%$ confidence interval (CI) of Bayesian approximation ratio for $n=10,20,30,40,50$.
Each average is computed over $500$ instances.
We observe that, in most cases, the ERM achieves the lowest Bayesian approximation ratio comparing to the other two mechanisms. 
When $q=0.7$, the ERM is better than the EEM but slightly worse than the IG. 
However, the empirical Bayesian approximation ratio of the ERM and IG converges to the same value.
Next, we consider the case of unbalanced capacities, where $q_1\neq q_2$, specifically $\vec q=(0.85,0.75),(0.8,0.4),(0.85,0.35)$.
Since the IG requires the two capacities to be identical, we compare only the ERM and the EEM.
Amongst the possible ERMs, we select the one optimal with respect to the Uniform distribution, obtained via Theorem \ref{thm:bestuniform}.
Thus the parameters of $\RM$ are
\begin{enumerate*}[label=(\roman*)]
    \item $\pi=Id$ for every $\vec q$ and
    \item $\vec v=(0.25,0.75)$ for $\vec q=(0.85,0.75)$, $\vec v=(0.6,0.8)$ for $\vec q=(0.8,0.4)$, and $\vec v=(0.65,0.85)$ for $\vec q=(0.85,0.35)$.
\end{enumerate*}
In this case, we consider only symmetric probability distributions, i.e. the Gaussian and the Uniform distribution.
Figure~\ref{fig:NE} (and Table~\ref{table:nonequitest} in Appendix~\ref{app:social}) shows the average and the $95\%$ CI of Bayesian approximation ratio computed over 500 instances.
Whenever $n\ge 20$, the ERM has a much lower Bayesian approximation ratio.
Notice that when $\vec q=(0.75,0.85)$, the ERM is optimal in both cases, and its limit Bayesian approximation ratio is $1$.
Indeed, the Bayesian approximation ratio of the ERM is almost equal to $1$ for every $n\ge  10$ and gets closer as $n$ increases, however the Bayesian approximation ratio of the EEM is always $\ge 1.79$ and gets worse as $n$ increases. 

\begin{figure}[t!]
     \centering
     \includegraphics[width=0.85\columnwidth]{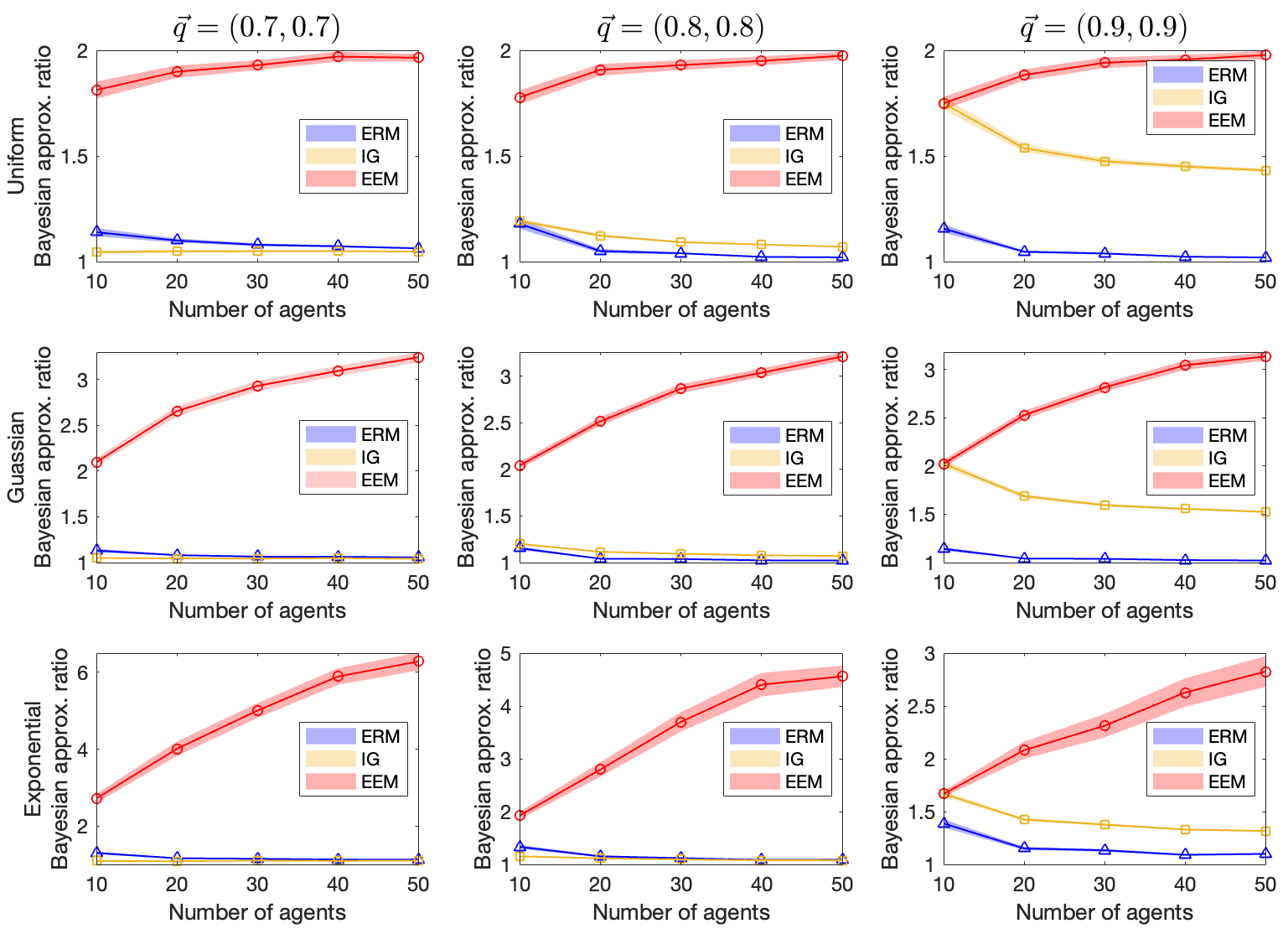}
     \caption{The Bayesian approximation ratio of ERM, IG, and EEM in the balanced case, i.e. $q_1=q_2$ for $n=10,20,\dots,50$ with respect to the Social Cost.
    Every column contains the results for different vector $\vec q$, while every row contains the results for a different probability distribution $\mu$.
    %
    %
    }
    \label{fig:EQ}
\end{figure}
\vspace{-0.1cm}

\begin{figure}[t!]
     \centering
     \includegraphics[width=0.8\columnwidth]{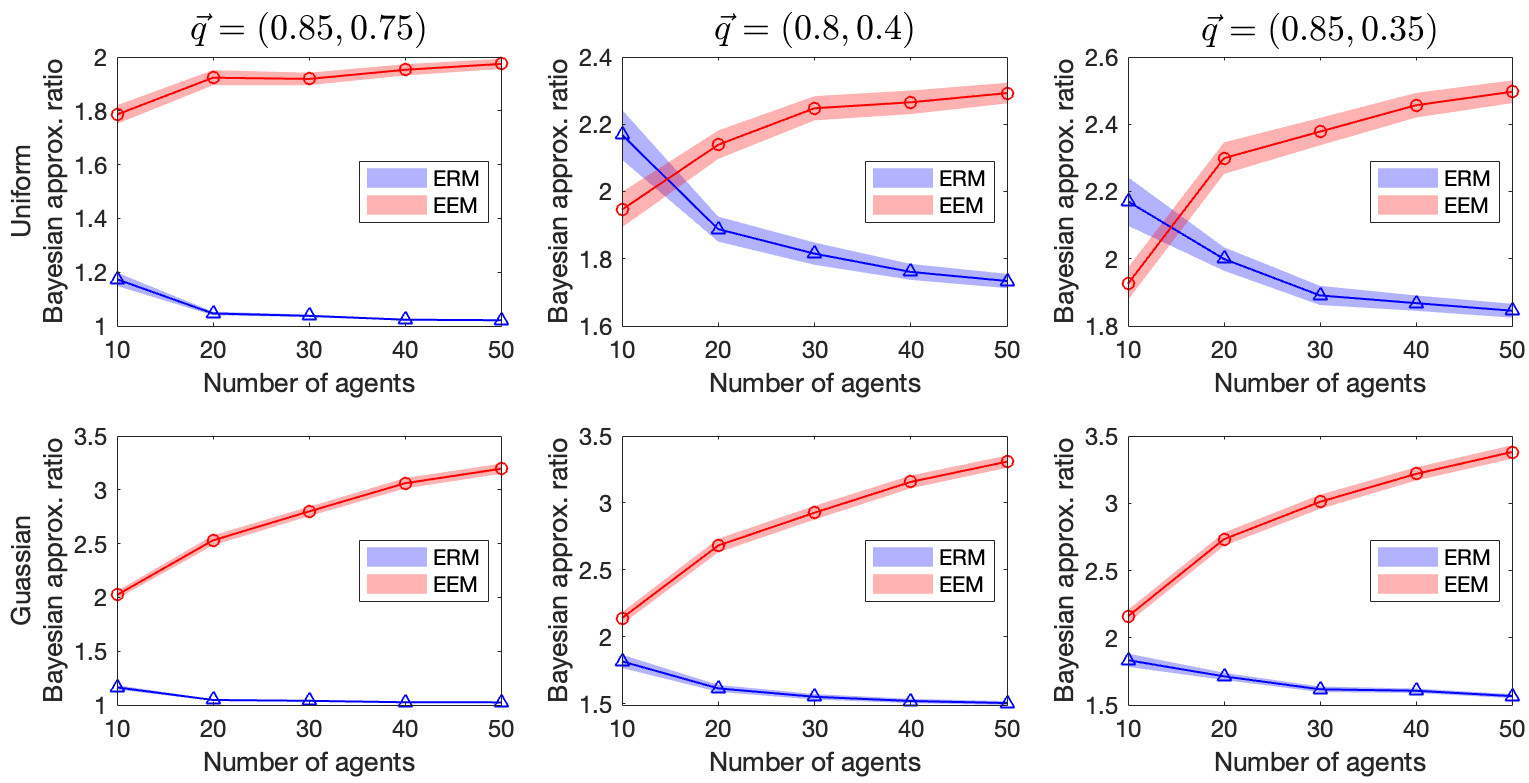}
     \caption{The Bayesian approximation ratio of ERM and EEM in the unbalanced case, i.e. $q_1\neq q_2$ for $n=10,20,\dots,50$ with respect to the Social Cost. 
Every column contains the results for a different vector $\vec q$, while every row contains the result for a different probability measure $\mu$.
%
%
}
     \label{fig:NE}
\end{figure}

The results for $l_2$ cost are very similar to those for Social Cost. 
Figure~\ref{fig:lp_EQ} and Table~\ref{tab:lp_EQ} in Appendix~\ref{app:l2} shows the average and the $95\%$ CI of Bayesian approximation ratio computed over $500$ instances under balanced capacities.
In most cases, the ERM mechanism achieves the lowest Bayesian approximation ratio comparing to the EEM and the IG. 
In Figure~\ref{fig:lp_NE} and Table~\ref{tab:lp_NE} in Appendix~\ref{app:l2} show the results under unbalanced capacities. 
We notice that, when $n\geq 30$, the ERM achives the lowest Bayesian approximation ratio. 

\paragraph{Convergence speed of the limit Bayesian approximation ratio.} 

\begin{figure}[h]
    \centering
    \includegraphics[width=0.8\columnwidth]{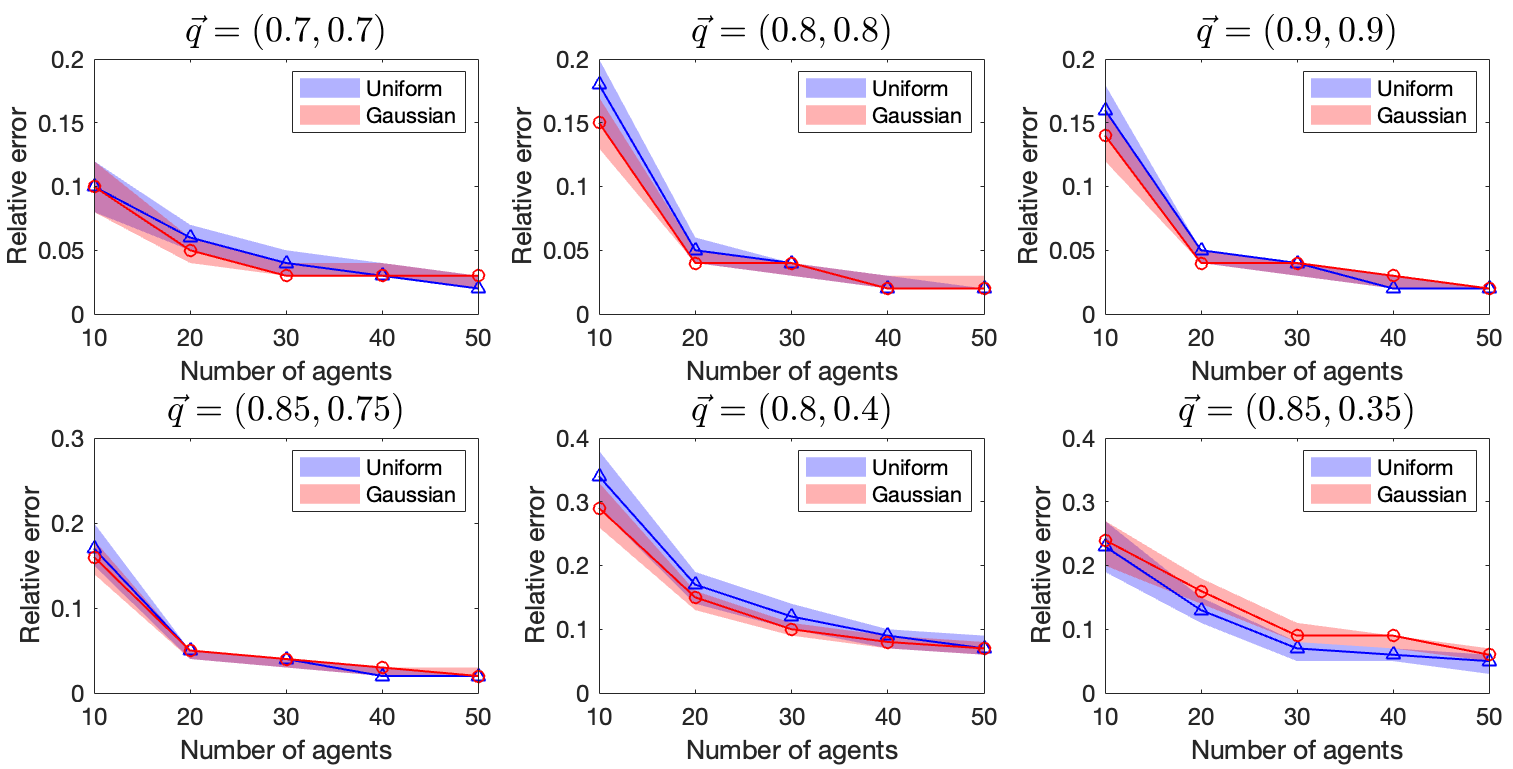}
    \caption{The relative error of ERM for $n=10,20,\dots,50$ for the Uniform and Gaussian distributions, with respect to the Social Cost. The first row shows the results for the balanced case while the second row shows the results for unbalanced case. }
    \label{fig:RE}
\end{figure}

Lastly, we test how close the Bayesian approximation ratio of the ERM is to the limit detected in Theorem \ref{thm:limitBAR}. 
That is, we calculate the relative error as $err_{rel}=\frac{\text{empirical } B_{ar}-\text{limit of } B_{ar}}{\text{limit of } B_{ar}}.$
Figure~\ref{fig:RE} (and Table~\ref{table:percenterrtest} in Appendix~\ref{app:social}) shows the relative error for the six cases.
Each average is computed over $500$ instances.
We observe that the relative error decreases as the number of agents increases, regardless of the distribution or the p.c.v. $\vec q$.
Moreover, in all cases but $\vec q=(0.8,0.4),(0.85,0.35)$, the relative errors of the ERM are less than $0.05$ as long as $n \ge 20$, which validates Theorem~\ref{thm:limitBAR} as a tool to predict the Bayesian approximation ratio for small values of $n$.
In other cases, we observe a slower convergence, as the percentage error is slightly larger and its highest value is $0.33$. 
Lastly, we point out that the results for the $l_2$ cost show no substantial differences from what we have observed for the Social cost.
For the sake of completeness, we report in Figure~\ref{fig:lp_RE} and Table~\ref{tab:lp_RE} in Appendix~\ref{app:l2} the relative error under $l_2$ cost.

\subsection{Experiment results -- The Maximum Cost}

In this section, we report the results of our experiments for the Maximum Cost.

\paragraph{Comparison with the EEM and the IG.} 

We first consider the case of balanced capacities, i.e. $\vec q=(q,q)$, in which we are able to compare the ERM to the IG and to the EEM.
We consider the best ERM obtained through Theorem \ref{thm:besterm_max}.
Figure~\ref{fig:max_EQ} (and Table~\ref{tab:max_EQ} in Appendix \ref{app:max}) shows the average and the CI of the Bayesian approximation ratio of all the three mechanisms.
For the balanced case, we observe that the results for the Maximum Cost are in line with what observed for the Social and $l_2$ costs: the best ERM has the lowest (or tied for lowest) Bayesian approximation ratio most of the cases. 
It is worth to notice that the IG and the best ERM have again the same limit.
We then consider the unbalanced case, where we can compare the best ERM only to the EEM.
Figure~\ref{fig:max_NE} (and Table~\ref{tab:max_NE} in Appendix \ref{app:max}) shows the results on unbalanced cases. 
In this case, we observe that the best ERM is outperformed by the EEM only when the distribution is symmetric and $q_2$ is much lower than $q_1$. 
Indeed, for such instances, the Bayesian approximation ratio of the EEM is lower than the one achieved by the best ERM.
Moreover the CI of the EEM is tighter than the CI of the best ERM. 
In all other cases, the best ERM outperforms the EEM, in line with what we observed in all other cases.

\begin{figure}[h!]
     \centering
     \includegraphics[width=0.8\columnwidth]{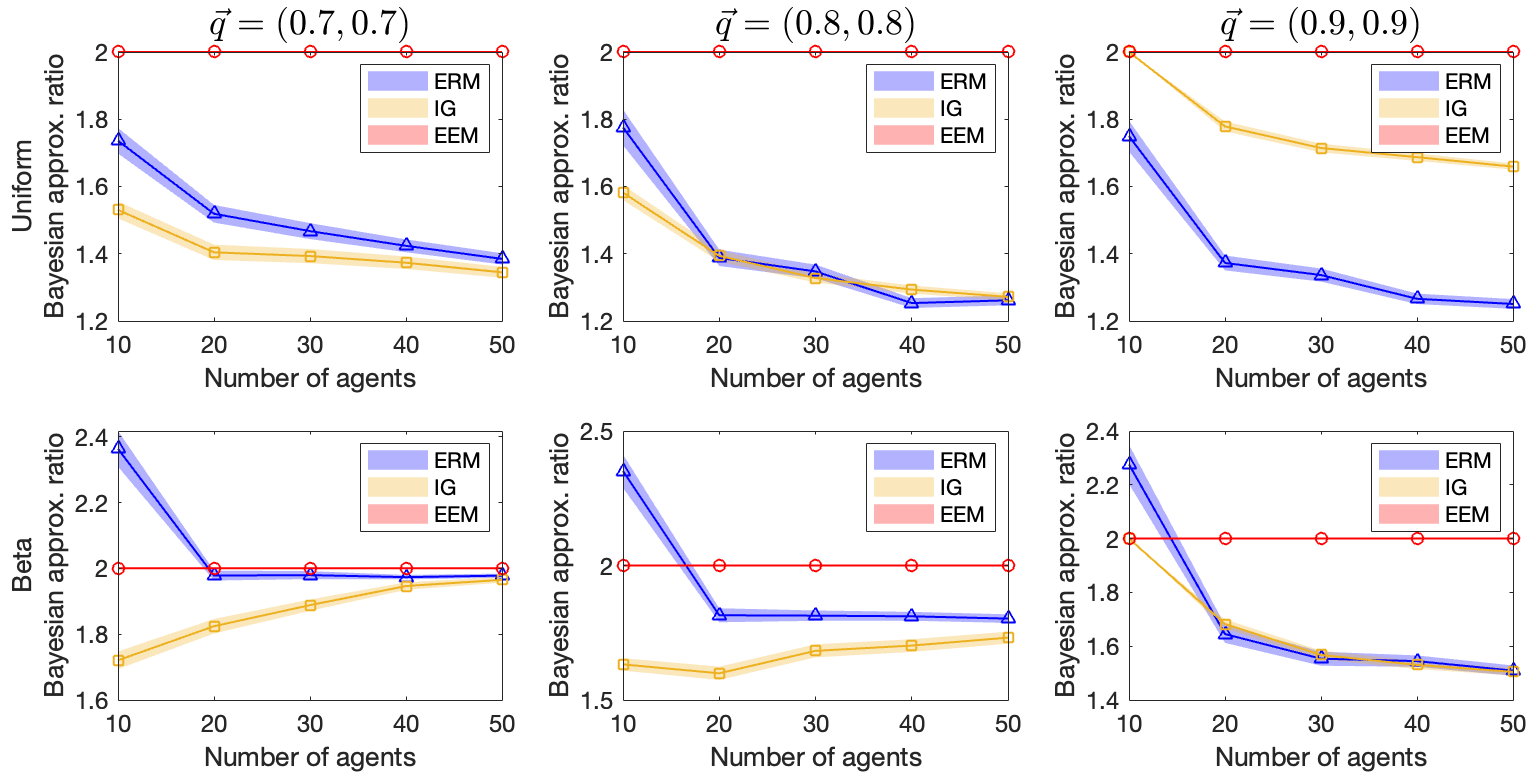}
     \caption{The Bayesian approximation ratio of ERM, IG, and EEM in the balanced case, i.e. $q_1=q_2$ for $n=10,20,\dots,50$ with respect to the Maximum Cost.
     Every column contains the results for different vector $\vec q$, while every row contains the results for a different probability distribution $\mu$.
    }
    \label{fig:max_EQ}
\end{figure}

\begin{figure}[h!]
     \centering
     \includegraphics[width=0.8\columnwidth]{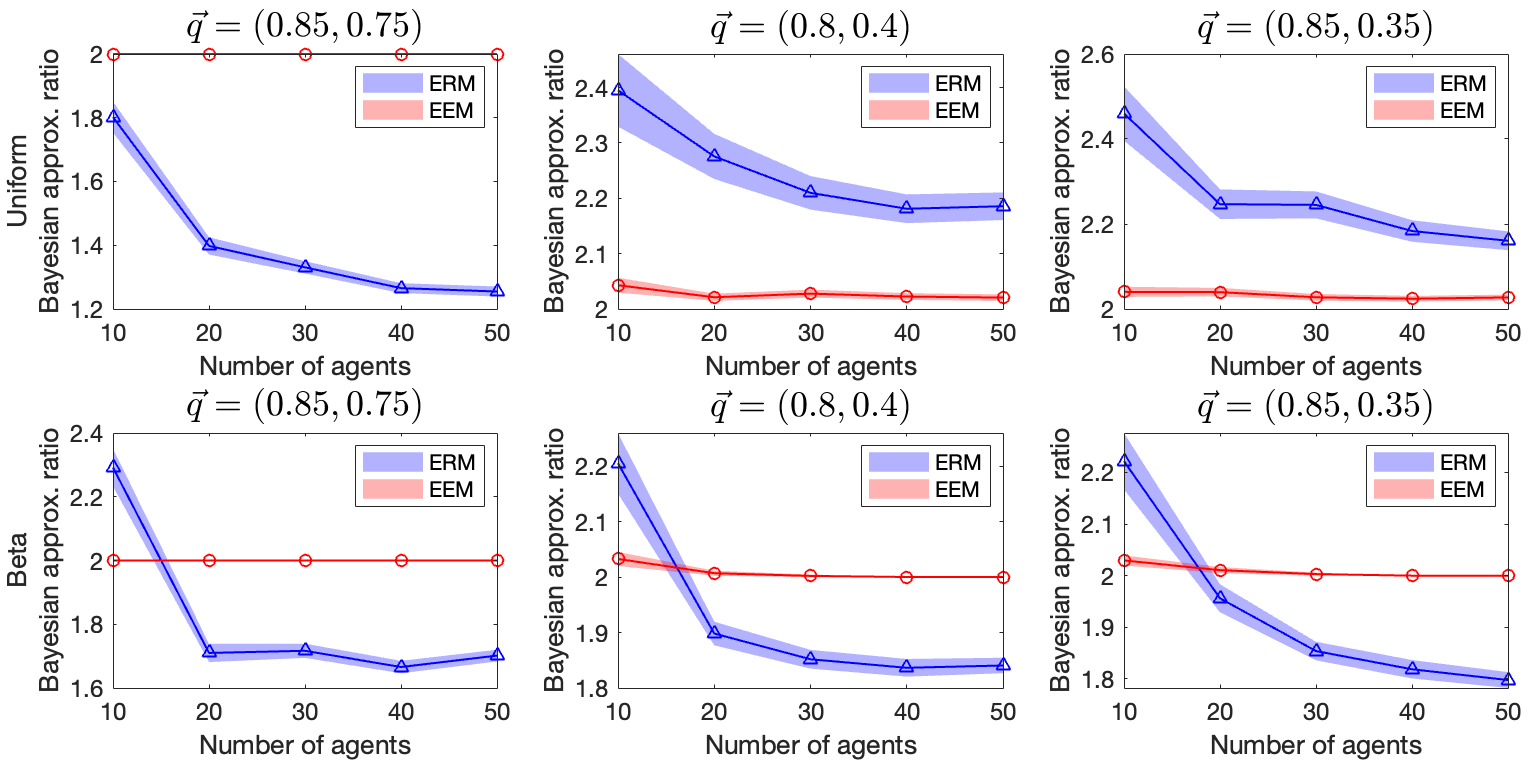}
     \caption{The Bayesian approximation ratio of ERM and EEM in the unbalanced case, i.e. $q_1\neq q_2$ for $n=10,20,\dots,50$ with respect to the Maximum Cost. 
     Every column contains the results for a different vector $\vec q$, while every row contains the result for a different probability measure $\mu$.}
     \label{fig:max_NE}
\end{figure}

\paragraph{Convergence speed of the limit Bayesian approximation ratio.} 

Lastly, we assess how close the empirical Bayesian approximation ratio is to the limit devised in Theorem \ref{thm:besterm_max} using the relative error.
We report in Figure~\ref{fig:max_RE} (and Table~\ref{tab:max_RE} in Appendix \ref{app:max}) our numerical results.
Again we observe that the relative error is small and that its absolute value becomes negligible as the number of agents increases.
Interestingly, in some cases, the relative error is negative, showing that the theoretical limit is a pessimistic (although reliable) approximation of the performances of the best ERM.

\begin{figure}[h!]
    \centering
    \includegraphics[width=0.8\columnwidth]{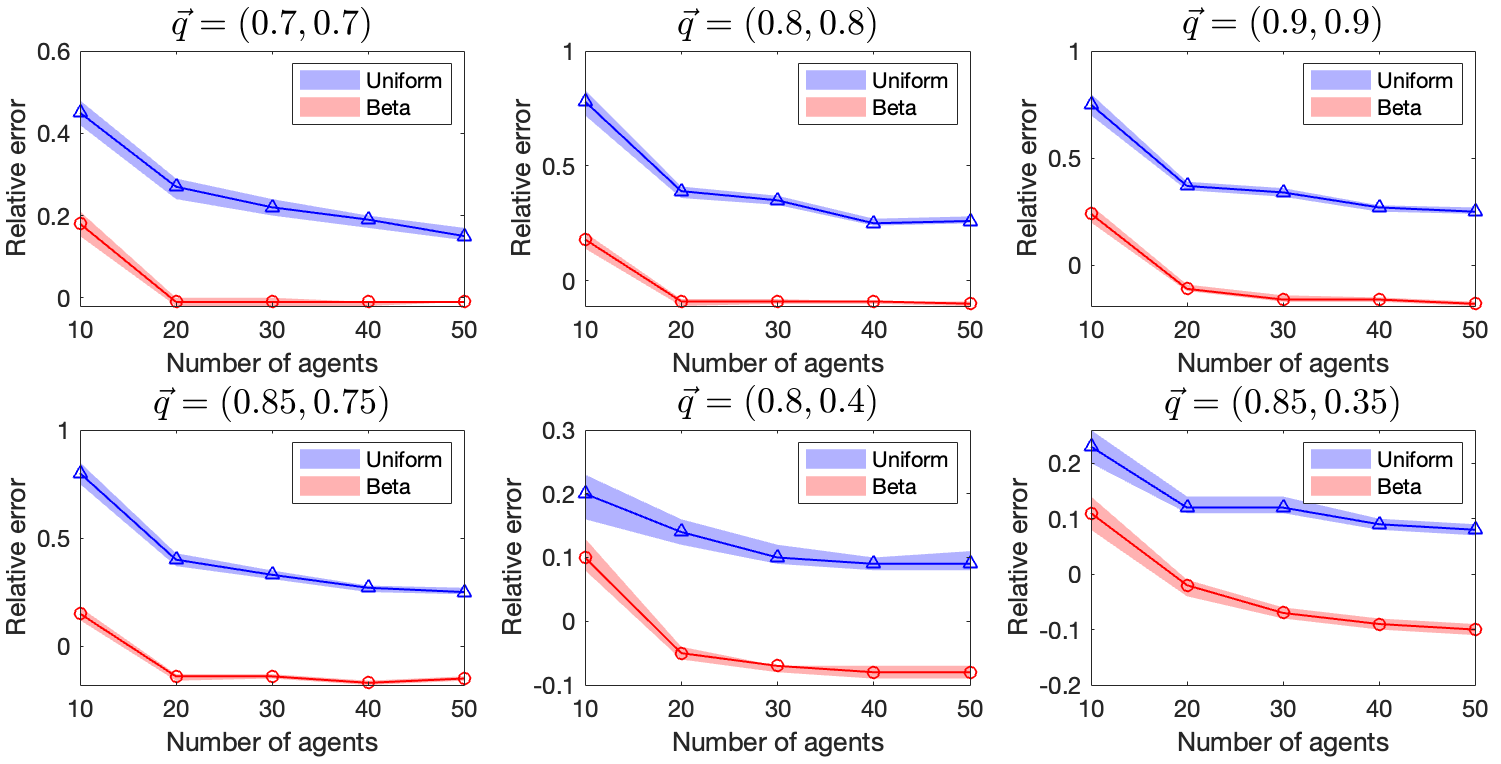}
    \caption{The relative error of ERM for $n=10,20,\dots,50$ for the Uniform and Beta distributions, with respect to the Maximum Cost. The first row shows the results for the balanced case while the second row shows the results for unbalanced case. }
    \label{fig:max_RE}
\end{figure}

\section{Conclusion and Future Works}

In this paper, we introduced the Extended Ranking Mechanisms, a generalization of Ranking Mechanisms introduced in \cite{aziz2020facility}.
After establishing the conditions under which ERMs remain truthful, we characterized the limit Bayesian approximation ratio of truthful ERM in terms of the probability distribution $\mu$, the p.c.v. $\vec q$, and the mechanism's parameters, namely $\pi$ and $\vec v$.
We have shown that, given $\mu$ and $\vec q$, there exists an optimal ERM and characterized it via a minimization problem, which we solved in two relevant frameworks.
Lastly, we conducted extensive numerical experiments to validate our findings, from which we inferred that a well-tuned ERM consistently outperforms all other known mechanisms.
An interesting extension would be to extend the ERMs to handle problems in higher dimensions using the decomposition proposed in \cite{auricchio2018computing}.
Lastly, we plan to study the connections between the Optimal Transportation problem and other Mechanism Design problems.

\acks{Gennaro Auricchio and Jie Zhang are partially supported by a Leverhulme Trust Research Project Grant (2021 – 2024). Jie Zhang is also supported by the EPSRC grant (EP/W014912/1). Mengxiao Zhang is supported by EPSRC Grant (EP/X01357X/1).}

\bibliography{sample}
\bibliographystyle{theapa}

\newpage

\appendix

\section{Experiment results -- The Social Cost}
\label{app:social}

\begin{table*}[ht]
\centering
\resizebox{\textwidth}{!}{
\begin{tabular}{c|c|c|ccc|ccc|ccc|ccc|ccc}
\toprule
\multirow{2}{*}{p.c.v. $\vec{q}$} & \multirow{2}{*}{Measure $\mu$} & \multirow{2}{*}{Mech.} & \multicolumn{3}{c|}{$n=10$} & \multicolumn{3}{c|}{20} & \multicolumn{3}{c|}{30} & \multicolumn{3}{c|}{40} & \multicolumn{3}{c}{50} \\ 
                                  &                                &                        & lb      & mean    & ub     & lb     & mean  & ub    & lb     & mean  & ub    & lb     & mean  & ub    & lb     & mean  & ub    \\ \hline
\multirow{9}{*}{$(0.7,0.7)$}      & \multirow{3}{*}{Uniform}                        & ERM                    & 1.12    & 1.14    & 1.16   & 1.09   & 1.10  & 1.11  & 1.07   & 1.08  & 1.09  & 1.07   & 1.07  & 1.08  & 1.06   & 1.06  & 1.07  \\
                                  &                                & IG                     & 1.04    & 1.04    & 1.05   & 1.04   & 1.05  & 1.05  & 1.04   & 1.05  & 1.05  & 1.05   & 1.05  & 1.05  & 1.04   & 1.05  & 1.05  \\
                                  &                                & EEM                    & 1.77    & 1.81    & 1.85   & 1.87   & 1.90  & 1.93  & 1.91   & 1.93  & 1.95  & 1.95   & 1.97  & 1.99  & 1.95   & 1.96  & 1.98  \\ \cline{2-18}
                                  & \multirow{3}{*}{Gaussian}                       & ERM                    & 1.11    & 1.13    & 1.15   & 1.07   & 1.08  & 1.09  & 1.06   & 1.06  & 1.07  & 1.06   & 1.06  & 1.07  & 1.05   & 1.06  & 1.06  \\
                                  &                                & IG                     & 1.04    & 1.05    & 1.06   & 1.04   & 1.05  & 1.05  & 1.04   & 1.04  & 1.05  & 1.04   & 1.05  & 1.05  & 1.04   & 1.04  & 1.05  \\
                                  &                                & EEM                    & 2.05    & 2.10    & 2.14   & 2.60   & 2.65  & 2.71  & 2.87   & 2.93  & 2.98  & 3.04   & 3.09  & 3.15  & 3.19   & 3.24  & 3.30  \\ \cline{2-18}
                                  & \multirow{3}{*}{Exp. }                          & ERM                    & 1.26    & 1.30    & 1.34   & 1.15   & 1.16  & 1.17  & 1.14   & 1.15  & 1.16  & 1.12   & 1.13  & 1.14  & 1.12   & 1.13  & 1.14  \\
                                  &                                & IG                     & 1.08    & 1.09    & 1.10   & 1.07   & 1.08  & 1.09  & 1.09   & 1.10  & 1.11  & 1.09   & 1.09  & 1.10  & 1.09   & 1.10  & 1.11  \\ 
                                  &                                & EEM                    & 2.60    & 2.72    & 2.84   & 3.82   & 4.01  & 4.20  & 4.81   & 5.00  & 5.20  & 5.67   & 5.89  & 6.11  & 6.05   & 6.28  & 6.52  \\ \hline \hline
\multirow{9}{*}{$(0.8,0.8)$}      & \multirow{3}{*}{Uniform}                        & ERM                    & 1.15    & 1.18    & 1.20   & 1.04   & 1.05  & 1.06  & 1.03   & 1.04  & 1.04  & 1.02   & 1.02  & 1.03  & 1.02   & 1.02  & 1.02  \\
                                  &                                & IG                     & 1.18    & 1.19    & 1.20   & 1.12   & 1.12  & 1.13  & 1.09   & 1.09  & 1.10  & 1.08   & 1.08  & 1.08  & 1.07   & 1.07  & 1.07  \\
                                  &                                & EEM                    & 1.74    & 1.78    & 1.81   & 1.88   & 1.91  & 1.94  & 1.91   & 1.93  & 1.95  & 1.93   & 1.95  & 1.97  & 1.95   & 1.97  & 1.99  \\ \cline{2-18}
                                  & \multirow{3}{*}{Gaussian}                       & ERM                    & 1.13    & 1.15    & 1.17   & 1.04   & 1.04  & 1.05  & 1.03   & 1.04  & 1.04  & 1.02   & 1.02  & 1.03  & 1.02   & 1.02  & 1.03  \\
                                  &                                & IG                     & 1.19    & 1.20    & 1.21   & 1.11   & 1.12  & 1.12  & 1.09   & 1.09  & 1.10  & 1.07   & 1.08  & 1.08  & 1.07   & 1.07  & 1.07  \\
                                  &                                & EEM                    & 2.00    & 2.04    & 2.08   & 2.47   & 2.51  & 2.56  & 2.82   & 2.86  & 2.91  & 2.99   & 3.04  & 3.08  & 3.16   & 3.21  & 3.26  \\ \cline{2-18}
                                  & \multirow{3}{*}{Exp.}                           & ERM                    & 1.29    & 1.33    & 1.37   & 1.13   & 1.15  & 1.17  & 1.11   & 1.12  & 1.13  & 1.08   & 1.09  & 1.09  & 1.08   & 1.09  & 1.10  \\
                                  &                                & IG                     & 1.14    & 1.15    & 1.16   & 1.11   & 1.12  & 1.12  & 1.09   & 1.10  & 1.10  & 1.07   & 1.08  & 1.08  & 1.07   & 1.08  & 1.08  \\
                                  &                                & EEM                    & 1.85    & 1.92    & 2.00   & 2.66   & 2.79  & 2.93  & 3.51   & 3.70  & 3.88  & 4.18   & 4.40  & 4.63  & 4.36   & 4.56  & 4.76  \\ \hline \hline
\multirow{9}{*}{$(0.9,0.9)$}      & \multirow{3}{*}{Uniform}                        & ERM                    & 1.14    & 1.16    & 1.18   & 1.04   & 1.05  & 1.05  & 1.03   & 1.04  & 1.04  & 1.02   & 1.02  & 1.03  & 1.02   & 1.02  & 1.02  \\
                                  &                                & IG                     & 1.72    & 1.75    & 1.78   & 1.52   & 1.54  & 1.55  & 1.46   & 1.48  & 1.49  & 1.44   & 1.45  & 1.46  & 1.42   & 1.43  & 1.44  \\
                                  &                                & EEM                    & 1.72    & 1.75    & 1.78   & 1.86   & 1.88  & 1.91  & 1.92   & 1.94  & 1.97  & 1.93   & 1.96  & 1.98  & 1.96   & 1.98  & 2.00  \\ \cline{2-18}
                                  & \multirow{3}{*}{Gaussian}                       & ERM                    & 1.12    & 1.14    & 1.16   & 1.04   & 1.04  & 1.05  & 1.03   & 1.04  & 1.04  & 1.02   & 1.03  & 1.03  & 1.02   & 1.02  & 1.02  \\
                                  &                                & IG                     & 1.99    & 2.03    & 2.06   & 1.67   & 1.69  & 1.71  & 1.58   & 1.60  & 1.61  & 1.54   & 1.56  & 1.57  & 1.51   & 1.53  & 1.54  \\
                                  &                                & EEM                    & 1.99    & 2.03    & 2.06   & 2.48   & 2.53  & 2.57  & 2.77   & 2.82  & 2.86  & 3.00   & 3.05  & 3.09  & 3.09   & 3.14  & 3.18  \\ \cline{2-18}
                                  & \multirow{3}{*}{Exp.}                           & ERM                    & 1.35    & 1.39    & 1.43   & 1.14   & 1.15  & 1.17  & 1.12   & 1.13  & 1.15  & 1.08   & 1.09  & 1.10  & 1.09   & 1.10  & 1.11  \\
                                  &                                & IG                     & 1.64    & 1.67    & 1.70   & 1.41   & 1.42  & 1.44  & 1.36   & 1.38  & 1.39  & 1.32   & 1.33  & 1.34  & 1.31   & 1.32  & 1.32  \\
                                  &                                & EEM                    & 1.64    & 1.67    & 1.70   & 2.00   & 2.08  & 2.16  & 2.20   & 2.31  & 2.42  & 2.49   & 2.63  & 2.76  & 2.68   & 2.83  & 2.97  \\
\bottomrule
\end{tabular} }
\caption{The Bayesian approximation ratio of ERM, IG, and EEM in the balanced case, i.e. $q_1=q_2$, with respect to the Social Cost. 
Every row contains the results for different mechanism, $\mu$, and $\vec q$. 
Every column contains the results for different $n\in \{10,\ldots,50\}$.
Every value is computed as the average of $500$ runs.}
\label{table:equicapacitated_test}
\end{table*}

\begin{table*}[ht]
\centering
\resizebox{\textwidth}{!}{
\begin{tabular}{c|c|c|ccc|ccc|ccc|ccc|ccc}
\toprule
\multirow{2}{*}{p.c.v. $\vec{q}$} & \multirow{2}{*}{Measure $\mu$} & \multirow{2}{*}{Mech.} & \multicolumn{3}{c|}{$n=10$} & \multicolumn{3}{c|}{20} & \multicolumn{3}{c|}{30} & \multicolumn{3}{c|}{40} & \multicolumn{3}{c}{50} \\
                                  &                                &                        & lb      & mean    & ub     & lb     & mean  & ub    & lb     & mean  & ub    & lb     & mean  & ub    & lb     & mean  & ub    \\
\hline                                 
\multirow{4}{*}{$(0.85,0.75)$}    & \multirow{2}{*}{Uniform}       & ERM                    & 1.15    & 1.17    & 1.20   & 1.04   & 1.05  & 1.05  & 1.03   & 1.04  & 1.04  & 1.02   & 1.02  & 1.03  & 1.02   & 1.02  & 1.02  \\
                                  &                                & EEM                    & 1.75    & 1.79    & 1.82   & 1.89   & 1.92  & 1.95  & 1.90   & 1.92  & 1.94  & 1.93   & 1.95  & 1.97  & 1.96   & 1.97  & 1.99  \\ \cline{2-18} 
                                  & \multirow{2}{*}{Gaussian}      & ERM                    & 1.14    & 1.16    & 1.18   & 1.04   & 1.05  & 1.05  & 1.03   & 1.04  & 1.04  & 1.02   & 1.03  & 1.03  & 1.02   & 1.02  & 1.03  \\
                                  &                                & EEM                    & 1.99    & 2.02    & 2.06   & 2.48   & 2.53  & 2.58  & 2.75   & 2.80  & 2.84  & 3.01   & 3.06  & 3.11  & 3.15   & 3.20  & 3.24  \\ \hline \hline
\multirow{4}{*}{$(0.8,0.4)$}      & \multirow{2}{*}{Uniform}       & ERM                    & 2.09    & 2.17    & 2.24   & 1.85   & 1.89  & 1.93  & 1.78   & 1.82  & 1.85  & 1.74   & 1.76  & 1.78  & 1.71   & 1.73  & 1.76  \\
                                  &                                & EEM                    & 1.89    & 1.95    & 2.00   & 2.10   & 2.14  & 2.18  & 2.21   & 2.25  & 2.28  & 2.23   & 2.27  & 2.30  & 2.26   & 2.29  & 2.32  \\ \cline{2-18}
                                  & \multirow{2}{*}{Gaussian}      & ERM                    & 1.77    & 1.82    & 1.87   & 1.59   & 1.62  & 1.64  & 1.53   & 1.55  & 1.57  & 1.51   & 1.52  & 1.54  & 1.49   & 1.51  & 1.52  \\
                                  &                                & EEM                    & 2.10    & 2.14    & 2.19   & 2.63   & 2.68  & 2.73  & 2.88   & 2.93  & 2.98  & 3.11   & 3.16  & 3.21  & 3.26   & 3.31  & 3.35  \\ \hline \hline
\multirow{4}{*}{$(0.85,0.35)$}    & \multirow{2}{*}{Uniform}       & ERM                    & 2.10    & 2.17    & 2.24   & 1.96   & 2.00  & 2.03  & 1.86   & 1.89  & 1.92  & 1.85   & 1.87  & 1.89  & 1.82   & 1.85  & 1.87  \\
                                  &                                & EEM                    & 1.88    & 1.93    & 1.97   & 2.25   & 2.30  & 2.35  & 2.34   & 2.38  & 2.42  & 2.42   & 2.46  & 2.49  & 2.46   & 2.50  & 2.53  \\ \cline{2-18}
                                  & \multirow{2}{*}{Gaussian}      & ERM                    & 1.78    & 1.83    & 1.88   & 1.69   & 1.71  & 1.74  & 1.60   & 1.62  & 1.64  & 1.59   & 1.61  & 1.62  & 1.55   & 1.57  & 1.58  \\
                                  &                                & EEM                    & 2.11    & 2.16    & 2.21   & 2.68   & 2.73  & 2.78  & 2.96   & 3.01  & 3.06  & 3.17   & 3.22  & 3.27  & 3.33   & 3.38  & 3.43 \\
\bottomrule                                 
\end{tabular}
}
\caption{The Bayesian approximation ratio of ERM and EEM in the unbalanced case, i.e. $q_1\neq q_2$, with respect to the Social Cost.
Every row contains the results for different mechanism, $\mu$, and $\vec q$.
Every column contains the results for different $n\in \{10,\ldots,50\}$.
Every value is computed as the average of $500$ runs.}
\label{table:nonequitest}
\end{table*}

\begin{table}[]
\resizebox{\textwidth}{!}{
\begin{tabular}{c|c|ccc|ccc|ccc|ccc|ccc}
\toprule
\multirow{2}{*}{p.c.v. $\vec{q}$} & \multirow{2}{*}{Measure $\mu$} & \multicolumn{3}{c|}{$n=10$} & \multicolumn{3}{c|}{$20$} & \multicolumn{3}{c|}{$30$} & \multicolumn{3}{c|}{$40$} & \multicolumn{3}{c}{$50$} \\
                                  &                                & lb      & mean    & ub     & lb     & mean   & ub     & lb     & mean   & ub     & lb     & mean   & ub     & lb     & mean   & ub     \\ \hline
\multirow{2}{*}{$(0.7,0.7)$}                & Uniform                        & 0.08    & 0.10    & 0.12   & 0.05   & 0.06   & 0.07   & 0.03   & 0.04   & 0.05   & 0.03   & 0.03   & 0.04   & 0.02   & 0.02   & 0.03   \\
                                  & Gaussian                       & 0.08    & 0.10    & 0.12   & 0.04   & 0.05   & 0.06   & 0.03   & 0.03   & 0.04   & 0.03   & 0.03   & 0.04   & 0.02   & 0.03   & 0.03   \\ \hline
\multirow{2}{*}{$(0.8,0.8)$}      & Uniform                        & 0.15    & 0.18    & 0.20   & 0.04   & 0.05   & 0.06   & 0.03   & 0.04   & 0.04   & 0.02   & 0.02   & 0.03   & 0.02   & 0.02   & 0.02   \\
                                  & Gaussian                       & 0.13    & 0.15    & 0.17   & 0.04   & 0.04   & 0.05   & 0.03   & 0.04   & 0.04   & 0.02   & 0.02   & 0.03   & 0.02   & 0.02   & 0.03   \\ \hline
\multirow{2}{*}{$(0.9,0.9)$}      & Uniform                        & 0.14    & 0.16    & 0.18   & 0.04   & 0.05   & 0.05   & 0.03   & 0.04   & 0.04   & 0.02   & 0.02   & 0.03   & 0.02   & 0.02   & 0.02   \\
                                  & Gaussian                       & 0.12    & 0.14    & 0.16   & 0.04   & 0.04   & 0.05   & 0.03   & 0.04   & 0.04   & 0.02   & 0.03   & 0.03   & 0.02   & 0.02   & 0.02   \\ \hline \hline
\multirow{2}{*}{$(0.85,0.75)$}    & Uniform                        & 0.15    & 0.17    & 0.20   & 0.04   & 0.05   & 0.05   & 0.03   & 0.04   & 0.04   & 0.02   & 0.02   & 0.03   & 0.02   & 0.02   & 0.02   \\
                                  & Gaussian                       & 0.14    & 0.16    & 0.18   & 0.04   & 0.05   & 0.05   & 0.03   & 0.04   & 0.04   & 0.02   & 0.03   & 0.03   & 0.02   & 0.02   & 0.03   \\ \hline
\multirow{2}{*}{$(0.8,0.4)$}      & Uniform                        & 0.29    & 0.34    & 0.38   & 0.14   & 0.17   & 0.19   & 0.10   & 0.12   & 0.14   & 0.07   & 0.09   & 0.10   & 0.06   & 0.07   & 0.09   \\
                                  & Gaussian                       & 0.26    & 0.29    & 0.33   & 0.13   & 0.15   & 0.16   & 0.09   & 0.10   & 0.11   & 0.07   & 0.08   & 0.09   & 0.06   & 0.07   & 0.08   \\ \hline
\multirow{2}{*}{$(0.85,0.35)$}    & Uniform           &0.19	 &0.23	  &0.27	   &0.11	&0.13	  &0.15    &	0.05   &0.07	& 0.08	 & 0.05	  &0.06	  & 0.07	& 0.03	 & 0.05	 & 0.06   \\
                                  & Gaussian                       & 0.20    & 0.24    & 0.27   & 0.14   & 0.16   & 0.18   & 0.08   & 0.09   & 0.11   & 0.07   & 0.09   & 0.09   & 0.05   & 0.06   & 0.07  \\
\bottomrule
\end{tabular}}
\caption{The relative error of ERM with respect to the Social Cost. For every row, we compute the relative error of the ERM for different $\vec q$, different probability distributions, while for every column we compute the relative error for different $n\in \{10,\dots,50\}$.}
\label{table:percenterrtest}
\end{table}

\section{Experiment results -- The $l_2$ Cost}
\label{app:l2}

\begin{figure}[H]
     \centering
     \includegraphics[width=0.8\columnwidth]{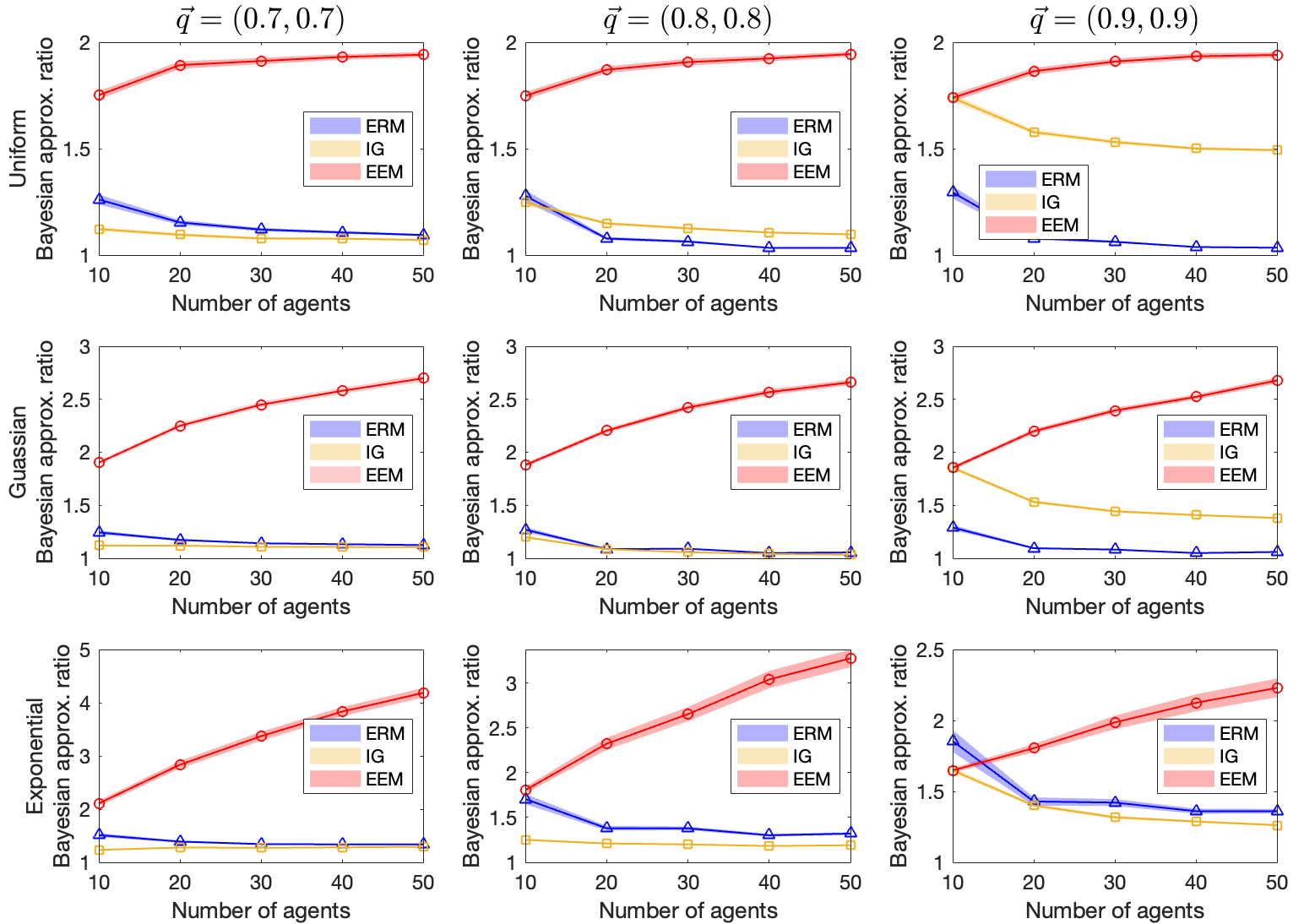}
     \caption{The Bayesian approximation ratio of ERM, IG, and EEM in the balanced case, i.e. $q_1=q_2$ for $n=10,20,\dots,50$, with respect to the $l_2$ cost. Every column contains the results for different vector $\vec q$, while every row contains the results for a different probability distribution $\mu$.
     }
     %
     \label{fig:lp_EQ}
\end{figure}

\begin{table}[H]
\centering
\resizebox{\textwidth}{!}{
\begin{tabular}{c|c|c|ccc|ccc|ccc|ccc|ccc}
\toprule
\multirow{2}{*}{p.c.v. $\vec{q}$} & \multirow{2}{*}{Measure $\mu$} & \multirow{2}{*}{Mech.} & \multicolumn{3}{c|}{$n=10$} & \multicolumn{3}{c|}{20} & \multicolumn{3}{c|}{30} & \multicolumn{3}{c|}{40} & \multicolumn{3}{c}{50} \\
                                  &                                & \multicolumn{1}{c}{}                       & \multicolumn{1}{|c}{lb} & \multicolumn{1}{c}{mean} & \multicolumn{1}{c|}{ub} & \multicolumn{1}{c}{lb} & \multicolumn{1}{c}{mean} & \multicolumn{1}{c|}{ub} & \multicolumn{1}{c}{lb} & \multicolumn{1}{c}{mean} & \multicolumn{1}{c|}{ub} & \multicolumn{1}{c}{lb} & \multicolumn{1}{c}{mean} & \multicolumn{1}{c|}{ub} & \multicolumn{1}{c}{lb} & \multicolumn{1}{c}{mean} & \multicolumn{1}{c}{ub} \\ \hline
\multirow{9}{*}{$(0.7,0.7)$}      & \multirow{3}{*}{Uniform}       & ERM                                        & 1.24                   & 1.26                     & 1.28                   & 1.14                   & 1.15                     & 1.16                   & 1.11                   & 1.12                     & 1.13                   & 1.10                   & 1.11                     & 1.11                   & 1.09                   & 1.09                     & 1.10                   \\ 
                                  &                                & IG                                         & 1.11                   & 1.12                     & 1.13                   & 1.09                   & 1.10                     & 1.10                   & 1.07                   & 1.08                     & 1.08                   & 1.07                   & 1.08                     & 1.08                   & 1.07                   & 1.07                     & 1.08                   \\ 
                                  &                                & EEM                                        & 1.73                   & 1.75                     & 1.77                   & 1.88                   & 1.89                     & 1.91                   & 1.90                   & 1.91                     & 1.93                   & 1.92                   & 1.93                     & 1.94                   & 1.93                   & 1.94                     & 1.95                   \\ \cline{2-18}
                                  & \multirow{3}{*}{Gaussian}      & ERM                                        & 1.23                   & 1.24                     & 1.26                   & 1.16                   & 1.17                     & 1.18                   & 1.13                   & 1.14                     & 1.15                   & 1.13                   & 1.13                     & 1.14                   & 1.12                   & 1.12                     & 1.13                   \\
                                  &                                & IG                                         & 1.11                   & 1.12                     & 1.13                   & 1.11                   & 1.12                     & 1.13                   & 1.10                   & 1.11                     & 1.11                   & 1.10                   & 1.11                     & 1.11                   & 1.10                   & 1.11                     & 1.11                   \\ 
                                  &                                & EEM                                        & 1.88                   & 1.90                     & 1.92                   & 2.23                   & 2.25                     & 2.28                   & 2.42                   & 2.45                     & 2.48                   & 2.55                   & 2.58                     & 2.61                   & 2.67                   & 2.70                     & 2.73                   \\  \cline{2-18}
                                  & \multirow{3}{*}{Exp.}          & ERM                                        & 1.47                   & 1.51                     & 1.55                   & 1.37                   & 1.39                     & 1.41                   & 1.33                   & 1.34                     & 1.36                   & 1.32                   & 1.34                     & 1.35                   & 1.32                   & 1.34                     & 1.35                   \\ 
                                  &                                & IG                                         & 1.22                   & 1.23                     & 1.25                   & 1.26                   & 1.28                     & 1.29                   & 1.26                   & 1.27                     & 1.28                   & 1.27                   & 1.28                     & 1.29                   & 1.28                   & 1.29                     & 1.30                   \\ 
                                  &                                & EEM                                        & 2.06                   & 2.11                     & 2.16                   & 2.76                   & 2.83                     & 2.90                   & 3.29                   & 3.37                     & 3.45                   & 3.75                   & 3.83                     & 3.92                   & 4.10                   & 4.19                     & 4.28                   \\ \hline \hline
\multirow{9}{*}{$(0.8,0.8)$}      & \multirow{3}{*}{Uniform}       & ERM                                        & 1.25                   & 1.28                     & 1.31                   & 1.07                   & 1.08                     & 1.09                   & 1.06                   & 1.06                     & 1.07                   & 1.03                   & 1.03                     & 1.04                   & 1.03                   & 1.03                     & 1.04                   \\ 
                                  &                                & IG                                         & 1.23                   & 1.25                     & 1.26                   & 1.14                   & 1.15                     & 1.16                   & 1.12                   & 1.13                     & 1.13                   & 1.10                   & 1.11                     & 1.11                   & 1.09                   & 1.10                     & 1.10                   \\ 
                                  &                                & EEM                                        & 1.73                    & 1.75                     & 1.77                   & 1.85                   & 1.87                     & 1.89                   & 1.89                   & 1.91                     & 1.92                   & 1.91                   & 1.92                     & 1.94                   & 1.93                   & 1.94                     & 1.96                   \\ \cline{2-18}
                                  & \multirow{3}{*}{Gaussian}      & ERM                                        & 1.25                   & 1.27                     & 1.29                   & 1.08                   & 1.09                     & 1.09                   & 1.08                   & 1.09                     & 1.10                   & 1.05                   & 1.05                     & 1.05                   & 1.05                   & 1.06                     & 1.06                   \\ 
                                  &                                & IG                                         & 1.19                   & 1.20                     & 1.21                   & 1.08                   & 1.09                     & 1.09                   & 1.05                   & 1.06                     & 1.06                   & 1.04                   & 1.04                     & 1.04                   & 1.03                   & 1.03                     & 1.04                   \\  
                                  &                                & EEM                                        & 1.86                   & 1.88                     & 1.90                   & 2.18                   & 2.20                     & 2.22                   & 2.40                   & 2.42                     & 2.45                   & 2.54                   & 2.57                     & 2.59                   & 2.63                   & 2.66                     & 2.69                   \\ \cline{2-18}
                                  & \multirow{3}{*}{Exp.}          & ERM                                        & 1.65                   & 1.71                     & 1.76                   & 1.35                   & 1.38                     & 1.41                   & 1.36                   & 1.38                     & 1.40                   & 1.29                   & 1.30                     & 1.32                   & 1.31                   & 1.32                     & 1.33                   \\
                                  &                                & IG                                         & 1.24                   & 1.25                     & 1.26                   & 1.20                   & 1.21                     & 1.22                   & 1.19                   & 1.20                     & 1.21                   & 1.17                   & 1.18                     & 1.19                   & 1.18                   & 1.19                     & 1.20                   \\
                                  &                                & EEM                                        & 1.77                   & 1.81                     & 1.84                   & 2.26                   & 2.32                     & 2.39                   & 2.57                   & 2.65                     & 2.73                   & 2.94                   & 3.04                     & 3.14                   & 3.18                   & 3.28                     & 3.37                   \\ \hline \hline
\multirow{9}{*}{$(0.9,0.9)$}      & \multirow{3}{*}{Uniform}       & ERM                                        & 1.26                   & 1.29                     & 1.32                   & 1.07                   & 1.08                     & 1.09                   & 1.06                   & 1.06                     & 1.07                   & 1.03                   & 1.04                     & 1.04                   & 1.03                   & 1.03                     & 1.04                   \\
                                  &                                & IG                                         & 1.72                   & 1.74                     & 1.76                   & 1.57                   & 1.58                     & 1.59                   & 1.52                   & 1.53                     & 1.54                   & 1.49                   & 1.50                     & 1.51                   & 1.49                   & 1.49                     & 1.50                   \\
                                  &                                & EEM                                        & 1.72                   & 1.74                     & 1.76                   & 1.85                   & 1.86                     & 1.88                   & 1.90                   & 1.91                     & 1.92                   & 1.92                   & 1.93                     & 1.95                   & 1.93                   & 1.94                     & 1.95                   \\ \cline{2-18}
                                  & \multirow{3}{*}{Gaussian}      & ERM                                        & 1.27                   & 1.29                     & 1.31                   & 1.09                   & 1.09                     & 1.10                   & 1.08                   & 1.08                     & 1.09                   & 1.05                   & 1.05                     & 1.05                   & 1.06                   & 1.06                     & 1.06                   \\
                                  &                                & IG                                         & 1.84                   & 1.85                     & 1.87                   & 1.52                   & 1.53                     & 1.55                   & 1.43                   & 1.44                     & 1.46                   & 1.40                   & 1.41                     & 1.42                   & 1.37                   & 1.38                     & 1.39                   \\
                                  &                                & EEM                                        & 1.84                   & 1.85                     & 1.87                   & 2.18                   & 2.20                     & 2.22                   & 2.37                   & 2.39                     & 2.42                   & 2.50                   & 2.52                     & 2.55                   & 2.65                   & 2.68                     & 2.71                   \\ \cline{2-18}
                                  & \multirow{3}{*}{Exp.}          & ERM                                        & 1.78                   & 1.86                     & 1.93                   & 1.40                   & 1.43                     & 1.46                   & 1.40                   & 1.42                     & 1.45                   & 1.34                   & 1.36                     & 1.38                   & 1.34                   & 1.36                     & 1.38                   \\
                                  &                                & IG                                         & 1.63                   & 1.65                     & 1.66                   & 1.39                   & 1.40                     & 1.41                   & 1.31                   & 1.32                     & 1.33                   & 1.28                   & 1.29                     & 1.30                   & 1.25                   & 1.26                     & 1.27                   \\
                                  &                                & EEM                                        & 1.63                   & 1.65                     & 1.66                   & 1.78                   & 1.80                     & 1.83                   & 1.94                   & 1.99                     & 2.03                   & 2.06                   & 2.12                     & 2.18                   & 2.16                   & 2.23                     & 2.30       \\ \bottomrule           
\end{tabular}}
\caption{The Bayesian approximation ratio of ERM, IG, and EEM in the balanced case, i.e. $q_1=q_2$, under $l_2$ cost. Every row contains the results for different mechanism, $\mu$, and $\vec q$. Every column contains the results for different $n\in \{10,\ldots,50\}$. Every value is computed as the average of $500$ runs.}
\label{tab:lp_EQ}
\end{table}

\begin{figure}[H]
     \centering
     \includegraphics[width=0.8\columnwidth]{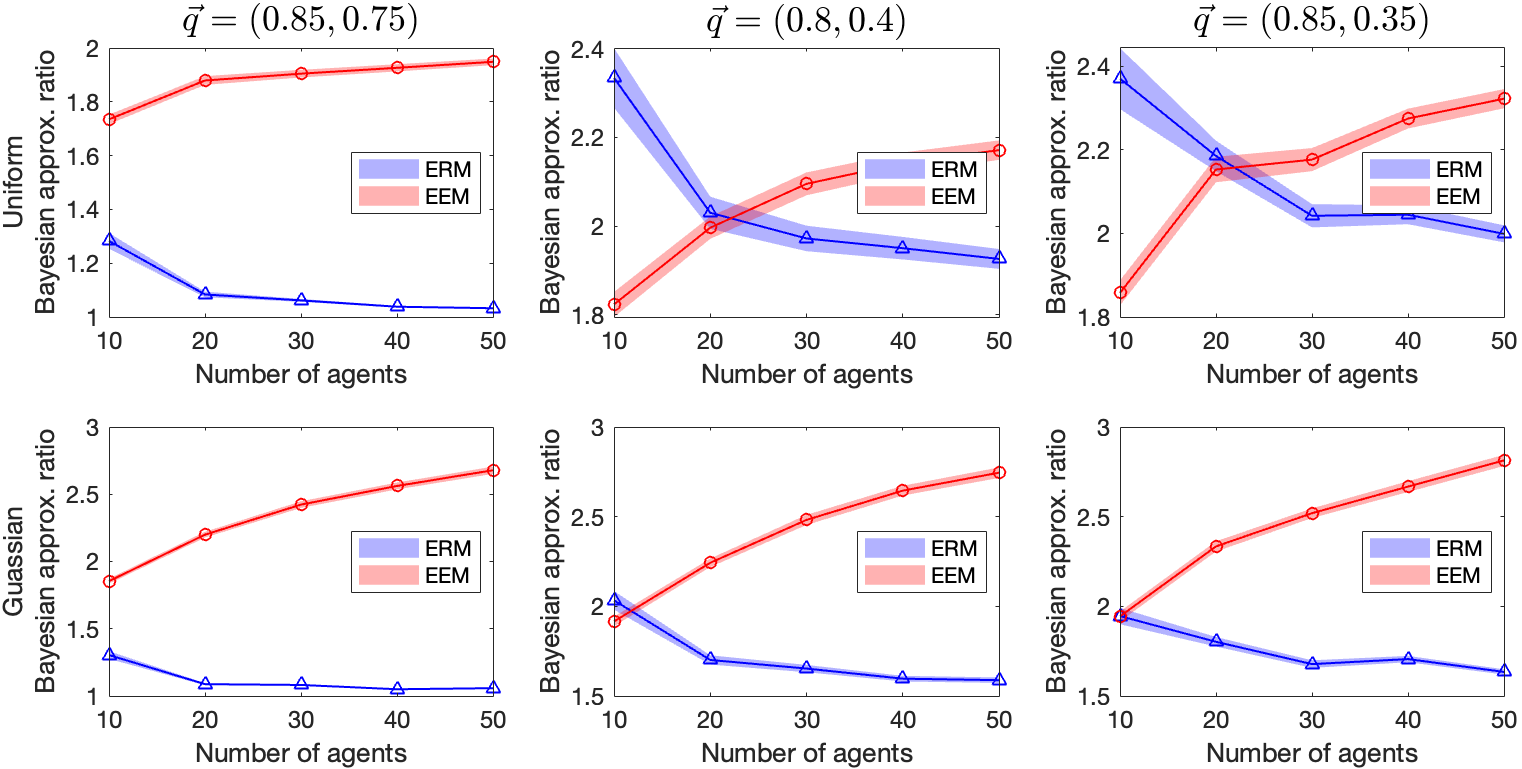}
     \caption{The Bayesian approximation ratio of ERM and EEM in the unbalanced case, i.e. $q_1\neq q_2$ for $n=10,20,\dots,50$ with respect to the $l_2$ cost. Every column contains the results for a different vector $\vec q$, while every row contains the results for a different probability measure $\mu$.
     }
     %
     \label{fig:lp_NE}
\end{figure}

\begin{table}[H]
\centering
\resizebox{\textwidth}{!}{
\begin{tabular}{c|c|c|ccc|ccc|ccc|ccc|ccc}
\toprule
\multirow{2}{*}{p.c.v. $\vec{q}$} & \multirow{2}{*}{Measure $\mu$} & \multirow{2}{*}{Mech.} & \multicolumn{3}{c|}{$n=10$} & \multicolumn{3}{c|}{20} & \multicolumn{3}{c|}{30} & \multicolumn{3}{c|}{40} & \multicolumn{3}{c}{50} \\
                                  &                                & \multicolumn{1}{c}{}                       & \multicolumn{1}{|c}{lb} & \multicolumn{1}{c}{mean} & \multicolumn{1}{c|}{ub} & \multicolumn{1}{c}{lb} & \multicolumn{1}{c}{mean} & \multicolumn{1}{c|}{ub} & \multicolumn{1}{c}{lb} & \multicolumn{1}{c}{mean} & \multicolumn{1}{c|}{ub} & \multicolumn{1}{c}{lb} & \multicolumn{1}{c}{mean} & \multicolumn{1}{c|}{ub} & \multicolumn{1}{c}{lb} & \multicolumn{1}{c}{mean} & \multicolumn{1}{c}{ub} \\ \hline
\multirow{4}{*}{$(0.85,0.75)$}    & \multirow{2}{*}{Uniform}       & ERM                                        & 1.25                   & 1.28                     & 1.31                   & 1.07                   & 1.08                     & 1.09                   & 1.06                   & 1.06                     & 1.07                   & 1.03                   & 1.04                     & 1.04                   & 1.03                   & 1.03                     & 1.04                   \\
                                  &                                & EEM                                        & 1.72                   & 1.73                     & 1.75                   & 1.86                   & 1.88                     & 1.90                   & 1.89                   & 1.90                     & 1.92                   & 1.91                   & 1.93                     & 1.94                   & 1.94                   & 1.95                     & 1.96                   \\ \cline{2-18}
                                  & \multirow{2}{*}{Gaussian}      & ERM                                        & 1.28                   & 1.30                     & 1.33                   & 1.08                   & 1.09                     & 1.10                   & 1.08                   & 1.08                     & 1.09                   & 1.05                   & 1.05                     & 1.05                   & 1.05                   & 1.06                     & 1.06                   \\
                                  &                                & EEM                                        & 1.84                   & 1.85                     & 1.87                   & 2.18                   & 2.20                     & 2.22                   & 2.40                   & 2.42                     & 2.45                   & 2.54                   & 2.56                     & 2.59                   & 2.65                   & 2.68                     & 2.71                   \\ \hline \hline
\multirow{4}{*}{$(0.8,0.4)$}      & \multirow{2}{*}{Uniform}       & ERM                                        & 2.27                   & 2.33                     & 2.40                   & 1.99                   & 2.03                     & 2.07                   & 1.94                   & 1.97                     & 2.00                   & 1.92                   & 1.95                     & 1.98                   & 1.90                   & 1.93                     & 1.95                   \\
                                  &                                & EEM                                        & 1.79                   & 1.82                     & 1.85                   & 1.97                   & 2.00                     & 2.02                   & 2.07                   & 2.10                     & 2.12                   & 2.12                   & 2.14                     & 2.17                   & 2.15                   & 2.17                     & 2.19                   \\ \cline{2-18}
                                  & \multirow{2}{*}{Gaussian}      & ERM                                        & 1.98                   & 2.03                     & 2.09                   & 1.68                   & 1.70                     & 1.73                   & 1.63                   & 1.65                     & 1.67                   & 1.58                   & 1.60                     & 1.61                   & 1.57                   & 1.59                     & 1.60                   \\
                                  &                                & EEM                                        & 1.89                   & 1.92                     & 1.94                   & 2.22                   & 2.24                     & 2.27                   & 2.45                   & 2.48                     & 2.51                   & 2.62                   & 2.65                     & 2.67                   & 2.72                   & 2.75                     & 2.77                   \\ \hline \hline
\multirow{4}{*}{$(0.85,0.35)$}    & \multirow{2}{*}{Uniform}       & ERM                                        & 2.30                   & 2.37                     & 2.44                   & 2.15                   & 2.19                     & 2.22                   & 2.01                   & 2.04                     & 2.07                   & 2.02                   & 2.04                     & 2.07                   & 1.98                   & 2.00                     & 2.02                   \\
                                  &                                & EEM                                        & 1.83                   & 1.86                     & 1.89                   & 2.12                   & 2.15                     & 2.18                   & 2.15                   & 2.18                     & 2.20                   & 2.25                   & 2.27                     & 2.30                   & 2.30                   & 2.32                     & 2.35                   \\ \cline{2-18}
                                  & \multirow{2}{*}{Gaussian}      & ERM                                        & 1.90                   & 1.95                     & 1.99                   & 1.78                   & 1.80                     & 1.83                   & 1.66                   & 1.68                     & 1.70                   & 1.69                   & 1.71                     & 1.72                   & 1.62                   & 1.63                     & 1.65                   \\
                                  &                                & EEM                                        & 1.92                   & 1.94                     & 1.97                   & 2.31                   & 2.33                     & 2.36                   & 2.49                   & 2.52                     & 2.55                   & 2.64                   & 2.67                     & 2.70                   & 2.78                   & 2.81                     & 2.84          \\ \bottomrule         
\end{tabular}}
\caption{The Bayesian approximation ratio of ERM and EEM in the unbalanced case, i.e. $q_1\neq q_2$, under $l_2$ cost. Every row contains the results for different mechanism, $\mu$, and $\vec q$. Every column contains the results for different $n\in \{10,\ldots,50\}$. Every value is computed as the average of $500$ runs.}
\label{tab:lp_NE}
\end{table}


\begin{figure}[h]
    \centering
    \includegraphics[width=0.8\columnwidth]{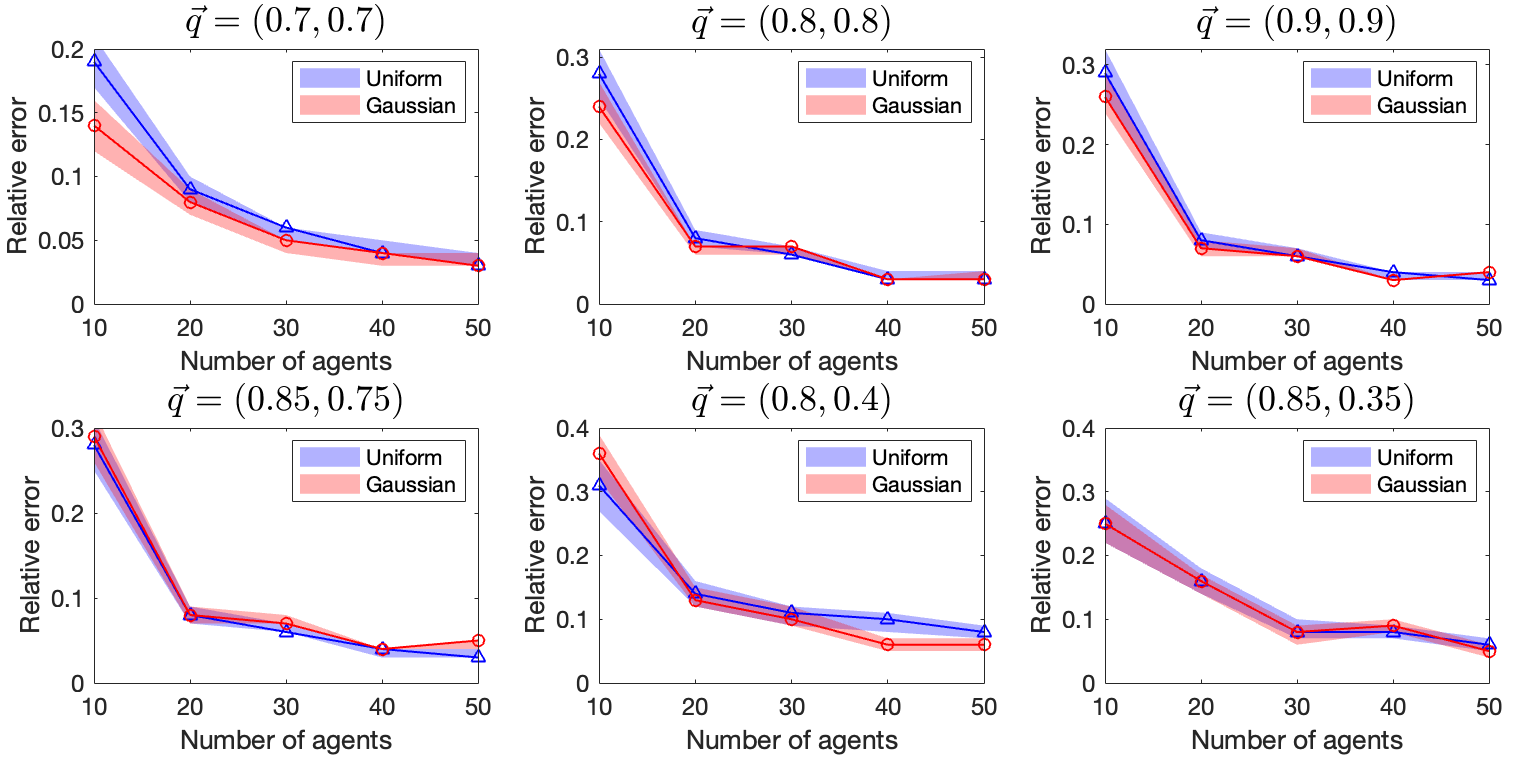}
    \caption{The relative error of ERM for $n=10,20,\dots,50$ for the Uniform and Gaussian distributions, with respect to the $l_2$ cost. The first row shows the results for the balanced case while the second row shows the results for unbalanced case. }
    \label{fig:lp_RE}
\end{figure}

\begin{table}[]
\resizebox{\textwidth}{!}{
\begin{tabular}{c|c|ccc|ccc|ccc|ccc|ccc}
\toprule
\multirow{2}{*}{p.c.v. $\vec{q}$} & \multirow{2}{*}{Measure $\mu$} & \multicolumn{3}{c|}{$n=10$} & \multicolumn{3}{c|}{$20$} & \multicolumn{3}{c|}{$30$} & \multicolumn{3}{c|}{$40$} & \multicolumn{3}{c}{$50$} \\
                                  &                                & lb      & mean    & ub     & lb     & mean   & ub     & lb     & mean   & ub     & lb     & mean   & ub     & lb     & mean   & ub     \\ \hline
\multirow{2}{*}{(0.7, 0.7)}       & Uniform                                            & 0.17   & 0.19   & 0.21   & 0.08   & 0.09   & 0.10   & 0.05   & 0.06   & 0.06   & 0.04   & 0.04   & 0.05   & 0.03   & 0.03   & 0.04   \\
                                  & Guassian                                           & 0.12   & 0.14   & 0.16   & 0.07   & 0.08   & 0.09   & 0.04   & 0.05   & 0.05   & 0.03   & 0.04   & 0.04   & 0.03   & 0.03   & 0.04   \\ \hline
\multirow{2}{*}{(0.8, 0.8)}       & Uniform                                            & 0.25   & 0.28   & 0.31   & 0.07   & 0.08   & 0.09   & 0.06   & 0.06   & 0.07   & 0.03   & 0.03   & 0.04   & 0.03   & 0.03   & 0.04   \\
                                  & Guassian                                           & 0.22   & 0.24   & 0.27   & 0.06   & 0.07   & 0.07   & 0.06   & 0.07   & 0.07   & 0.03   & 0.03   & 0.03   & 0.03   & 0.03   & 0.04   \\ \hline
\multirow{2}{*}{(0.9, 0.9)}       & Uniform                                            & 0.26   & 0.29   & 0.32   & 0.07   & 0.08   & 0.09   & 0.06   & 0.06   & 0.07   & 0.03   & 0.04   & 0.04   & 0.03   & 0.03   & 0.04   \\
                                  & Guassian                                           & 0.24   & 0.26   & 0.29   & 0.06   & 0.07   & 0.08   & 0.06   & 0.06   & 0.07   & 0.03   & 0.03   & 0.03   & 0.04   & 0.04   & 0.04   \\ \hline \hline
\multirow{2}{*}{(0.85, 0.75)}     & Uniform                                            & 0.25   & 0.28   & 0.31   & 0.07   & 0.08   & 0.09   & 0.06   & 0.06   & 0.07   & 0.03   & 0.04   & 0.04   & 0.03   & 0.03   & 0.04   \\
                                  & Guassian                                           & 0.26   & 0.29   & 0.32   & 0.07   & 0.08   & 0.09   & 0.07   & 0.07   & 0.08   & 0.04   & 0.04   & 0.04   & 0.04   & 0.05   & 0.05   \\ \hline
\multirow{2}{*}{(0.8, 0.4)}       & Uniform                                            & 0.27   & 0.31   & 0.35   & 0.12   & 0.14   & 0.16   & 0.09   & 0.11   & 0.12   & 0.08   & 0.10   & 0.11   & 0.07   & 0.08   & 0.09   \\
                                  & Guassian                                           & 0.32   & 0.36   & 0.39   & 0.12   & 0.13   & 0.15   & 0.09   & 0.10   & 0.12   & 0.05   & 0.06   & 0.07   & 0.05   & 0.06   & 0.07   \\ \hline
\multirow{2}{*}{(0.85, 0.35)}     & Uniform                                            & 0.22   & 0.25   & 0.29   & 0.14   & 0.16   & 0.18   & 0.07   & 0.08   & 0.10   & 0.07   & 0.08   & 0.09   & 0.05   & 0.06   & 0.07   \\
                                  & Guassian                                           & 0.22   & 0.25   & 0.28   & 0.14   & 0.16   & 0.17   & 0.06   & 0.08   & 0.09   & 0.08   & 0.09   & 0.10   & 0.04   & 0.05   & 0.06  \\\bottomrule
\end{tabular} }
\caption{The relative error of ERM with respect to the $l_2$ cost. For every row, we compute the relative error of the ERM for different $\vec q$, different probability distributions, while for every column we compute the relative error for different $n\in \{10,\dots,50\}$.}
\label{tab:lp_RE}
\end{table}

\clearpage 


\section{Experiment results -- Maximum Cost}
\label{app:max}

\begin{table}[h]
\centering
\resizebox{\textwidth}{!}{
\begin{tabular}{c|c|c|ccc|ccc|ccc|ccc|ccc}
\toprule
\multirow{2}{*}{p.c.v. $\vec{q}$} & \multirow{2}{*}{Measure $\mu$} & \multirow{2}{*}{Mech.} & \multicolumn{3}{c|}{$n=10$} & \multicolumn{3}{c|}{20} & \multicolumn{3}{c|}{30} & \multicolumn{3}{c|}{40} & \multicolumn{3}{c}{50} \\ 
                                  &                                &                        & lb      & mean    & ub     & lb     & mean  & ub    & lb     & mean  & ub    & lb     & mean  & ub    & lb     & mean  & ub    \\ \hline
\multirow{6}{*}{$(0.7,0.7)$}                          & \multirow{3}{*}{Uniform}       & ERM                    & 1.70    & 1.74    & 1.77   & 1.49   & 1.52   & 1.55   & 1.44   & 1.47   & 1.49   & 1.40   & 1.42   & 1.44   & 1.37   & 1.38   & 1.40   \\
                                                      &                                & IG                     & 1.51    & 1.53    & 1.56   & 1.38   & 1.40   & 1.43   & 1.37   & 1.39   & 1.41   & 1.35   & 1.37   & 1.39   & 1.33   & 1.34   & 1.36   \\
                                                      &                                & EEM                    & 2.00    & 2.00    & 2.00   & 2.00   & 2.00   & 2.00   & 2.00   & 2.00   & 2.00   & 2.00   & 2.00   & 2.00   & 2.00   & 2.00   & 2.00   \\\cline{2-18}
                                                      & \multirow{3}{*}{Beta}     & ERM                    & 2.31    & 2.36    & 2.42   & 1.96   & 1.98   & 1.99   & 1.97   & 1.98   & 1.99   & 1.97   & 1.97   & 1.98   & 1.97   & 1.98   & 1.98   \\
                                                      &                                & IG                     & 1.69    & 1.72    & 1.75   & 1.80   & 1.82   & 1.85   & 1.87   & 1.89   & 1.91   & 1.93   & 1.95   & 1.96   & 1.96   & 1.97   & 1.97   \\
                                                      &                                & EEM                    & 2.00    & 2.00    & 2.00   & 2.00   & 2.00   & 2.00   & 2.00   & 2.00   & 2.00   & 2.00   & 2.00   & 2.00   & 2.00   & 2.00   & 2.00   \\ \hline \hline
\multirow{6}{*}{$(0.8,0.8)$}                          & \multirow{3}{*}{Uniform}       & ERM                    & 1.72    & 1.78    & 1.83   & 1.36   & 1.39   & 1.41   & 1.33   & 1.35   & 1.37   & 1.24   & 1.25   & 1.27   & 1.25   & 1.26   & 1.28   \\
                                                      &                                & IG                     & 1.56    & 1.58    & 1.60   & 1.38   & 1.39   & 1.41   & 1.31   & 1.33   & 1.34   & 1.28   & 1.29   & 1.31   & 1.26   & 1.27   & 1.28   \\
                                                      &                                & EEM                    & 2.00    & 2.00    & 2.00   & 2.00   & 2.00   & 2.00   & 2.00   & 2.00   & 2.00   & 2.00   & 2.00   & 2.00   & 2.00   & 2.00   & 2.00   \\ \cline{2-18}
                                                      & \multirow{3}{*}{Beta}     & ERM                    & 2.29    & 2.35    & 2.41   & 1.79   & 1.82   & 1.84   & 1.79   & 1.81   & 1.83   & 1.79   & 1.81   & 1.83   & 1.79   & 1.80   & 1.82   \\
                                                      &                                & IG                     & 1.61    & 1.63    & 1.66   & 1.57   & 1.60   & 1.62   & 1.66   & 1.68   & 1.71   & 1.68   & 1.70   & 1.73   & 1.71   & 1.73   & 1.75   \\
                                                      &                                & EEM                    & 2.00    & 2.00    & 2.00   & 2.00   & 2.00   & 2.00   & 2.00   & 2.00   & 2.00   & 2.00   & 2.00   & 2.00   & 2.00   & 2.00   & 2.00   \\ \hline \hline
\multirow{6}{*}{$(0.9,0.9)$}                          & \multirow{3}{*}{Uniform}       & ERM                    & 1.70    & 1.75    & 1.80   & 1.35   & 1.37   & 1.39   & 1.32   & 1.34   & 1.36   & 1.25   & 1.27   & 1.28   & 1.24   & 1.25   & 1.27   \\
                                                      &                                & IG                     & 2.00    & 2.00    & 2.00   & 1.76   & 1.78   & 1.79   & 1.70   & 1.71   & 1.73   & 1.68   & 1.69   & 1.70   & 1.65   & 1.66   & 1.67   \\
                                                      &                                & EEM                    & 2.00    & 2.00    & 2.00   & 2.00   & 2.00   & 2.00   & 2.00   & 2.00   & 2.00   & 2.00   & 2.00   & 2.00   & 2.00   & 2.00   & 2.00   \\ \cline{2-18}
                                                      & \multirow{3}{*}{Beta}     & ERM                    & 2.20    & 2.28    & 2.35   & 1.61   & 1.65   & 1.68   & 1.53   & 1.55   & 1.58   & 1.52   & 1.54   & 1.57   & 1.49   & 1.51   & 1.53   \\
                                                      &                                & IG                     & 2.00    & 2.00    & 2.00   & 1.67   & 1.68   & 1.70   & 1.55   & 1.57   & 1.58   & 1.52   & 1.53   & 1.55   & 1.49   & 1.50   & 1.52   \\
                                                      &                                & EEM                    & 2.00    & 2.00    & 2.00   & 2.00   & 2.00   & 2.00   & 2.00   & 2.00   & 2.00   & 2.00   & 2.00   & 2.00   & 2.00   & 2.00   & 2.00  \\ \bottomrule
\end{tabular}}
\caption{The Bayesian approximation ratio of ERM, IG, and EEM in the balanced case, i.e. $q_1=q_2$ with respect to the Maximum Cost.
Every row contains the results for different mechanism, $\mu$, and $\vec q$. 
Every column contains the results for different $n\in \{10,\ldots,50\}$.
Every value is computed as the average of $500$ runs.}
\label{tab:max_EQ}
\end{table}

\begin{table}[h]
\centering
\resizebox{\textwidth}{!}{
\begin{tabular}{c|c|c|ccc|ccc|ccc|ccc|ccc}
\toprule
\multirow{2}{*}{p.c.v. $\vec{q}$} & \multirow{2}{*}{Measure $\mu$} & \multirow{2}{*}{Mech.} & \multicolumn{3}{c|}{$n=10$} & \multicolumn{3}{c|}{20} & \multicolumn{3}{c|}{30} & \multicolumn{3}{c|}{40} & \multicolumn{3}{c}{50} \\
                                  &                                &                        & lb      & mean    & ub     & lb     & mean  & ub    & lb     & mean  & ub    & lb     & mean  & ub    & lb     & mean  & ub    \\ \hline  
\multirow{4}{*}{$(0.85,0.75)$}                          & \multirow{2}{*}{Uniform}       & ERM                    & 1.75    & 1.80    & 1.85   & 1.37   & 1.40   & 1.43   & 1.31   & 1.33   & 1.35   & 1.25   & 1.27   & 1.28   & 1.24   & 1.25   & 1.27   \\
                                                      &                                & EEM                    & 2.00    & 2.00    & 2.00   & 2.00   & 2.00   & 2.00   & 2.00   & 2.00   & 2.00   & 2.00   & 2.00   & 2.00   & 2.00   & 2.00   & 2.00   \\ \cline{2-18}
                                                      & \multirow{2}{*}{Beta}     & ERM                    & 2.23    & 2.29    & 2.35   & 1.68   & 1.71   & 1.74   & 1.69   & 1.72   & 1.74   & 1.65   & 1.67   & 1.69   & 1.68   & 1.70   & 1.72   \\
                                                      &                                & EEM                    & 2.00    & 2.00    & 2.00   & 2.00   & 2.00   & 2.00   & 2.00   & 2.00   & 2.00   & 2.00   & 2.00   & 2.00   & 2.00   & 2.00   & 2.00   \\ \hline \hline
\multirow{4}{*}{$(0.8,0.4)$}                        & \multirow{2}{*}{Uniform}       & ERM                    & 2.33    & 2.40    & 2.46   & 2.23   & 2.28   & 2.32   & 2.18   & 2.21   & 2.24   & 2.16   & 2.18   & 2.21   & 2.16   & 2.19   & 2.21   \\ 
                                                      &                                & EEM                    & 2.03    & 2.04    & 2.06   & 2.01   & 2.02   & 2.03   & 2.02   & 2.03   & 2.03   & 2.02   & 2.02   & 2.03   & 2.01   & 2.02   & 2.03   \\ \cline{2-18}
                                                      & \multirow{2}{*}{Beta}     & ERM                    & 2.15    & 2.20    & 2.26   & 1.88   & 1.90   & 1.92   & 1.83   & 1.85   & 1.87   & 1.82   & 1.84   & 1.85   & 1.83   & 1.84   & 1.85   \\
                                                      &                                & EEM                    & 2.02    & 2.03    & 2.05   & 2.00   & 2.01   & 2.01   & 2.00   & 2.00   & 2.00   & 2.00   & 2.00   & 2.00   & 2.00   & 2.00   & 2.00   \\ \hline \hline
\multirow{4}{*}{$(0.85,0.35)$}                        & \multirow{2}{*}{Uniform}       & ERM                    & 2.40    & 2.46    & 2.52   & 2.21   & 2.25   & 2.28   & 2.21   & 2.24   & 2.28   & 2.16   & 2.18   & 2.21   & 2.14   & 2.16   & 2.18   \\
                                                      &                                & EEM                    & 2.03    & 2.04    & 2.05   & 2.03   & 2.04   & 2.05   & 2.02   & 2.03   & 2.04   & 2.02   & 2.02   & 2.03   & 2.02   & 2.03   & 2.03   \\ \cline{2-18}
                                                      & \multirow{2}{*}{Beta}     & ERM                    & 2.17    & 2.22    & 2.28   & 1.93   & 1.96   & 1.98   & 1.84   & 1.85   & 1.87   & 1.80   & 1.82   & 1.84   & 1.78   & 1.80   & 1.81   \\
                                                      &                                & EEM                    & 2.02    & 2.03    & 2.04   & 2.00   & 2.01   & 2.02   & 2.00   & 2.00   & 2.01   & 2.00   & 2.00   & 2.00   & 2.00   & 2.00   & 2.00  \\ \bottomrule
\end{tabular}}
\caption{The Bayesian approximation ratio of ERM and EEM in the unbalanced case, i.e. $q_1\neq q_2$, with respect to the Maximum Cost. 
Every row contains the results for different mechanism, $\mu$, and $\vec q$.
Every column contains the results for different $n\in \{10,\ldots,50\}$.
Every value is computed as the average of $500$ runs.}
\label{tab:max_NE}
\end{table}

\begin{table}[]
\resizebox{\textwidth}{!}{
\begin{tabular}{c|c|ccc|ccc|ccc|ccc|ccc}
\toprule
\multirow{2}{*}{p.c.v. $\vec{q}$} & \multirow{2}{*}{Measure $\mu$} & \multicolumn{3}{c|}{$n=10$} & \multicolumn{3}{c|}{$20$} & \multicolumn{3}{c|}{$30$} & \multicolumn{3}{c|}{$40$} & \multicolumn{3}{c}{$50$} \\
                                  &                                & lb      & mean    & ub     & lb     & mean   & ub     & lb     & mean   & ub     & lb     & mean   & ub     & lb     & mean   & ub     \\ \hline
\multirow{2}{*}{(0.7, 0.7)}                           & Uniform                                            & 0.42   & 0.45   & 0.48   & 0.24   & 0.27   & 0.29   & 0.20   & 0.22   & 0.24   & 0.17   & 0.19   & 0.20   & 0.14   & 0.15   & 0.17   \\
                                                      & Beta                                               & 0.15   & 0.18   & 0.21   & -0.02  & -0.01  & 0.00   & -0.02  & -0.01  & 0.00   & -0.02  & -0.01  & -0.01  & -0.01  & -0.01  & -0.01  \\ \hline
\multirow{2}{*}{(0.8, 0.8)}                           & Uniform                                            & 0.72   & 0.78   & 0.83   & 0.36   & 0.39   & 0.41   & 0.33   & 0.35   & 0.37   & 0.24   & 0.25   & 0.27   & 0.25   & 0.26   & 0.28   \\
                                                      & Beta                                               & 0.14   & 0.18   & 0.21   & -0.11  & -0.09  & -0.08  & -0.10  & -0.09  & -0.08  & -0.10  & -0.09  & -0.09  & -0.11  & -0.10  & -0.09  \\ \hline
\multirow{2}{*}{(0.9, 0.9)}                           & Uniform                                            & 0.70   & 0.75   & 0.80   & 0.35   & 0.37   & 0.39   & 0.32   & 0.34   & 0.36   & 0.25   & 0.27   & 0.28   & 0.24   & 0.25   & 0.27   \\
                                                      & Beta                                               & 0.20   & 0.24   & 0.28   & -0.12  & -0.11  & -0.09  & -0.17  & -0.16  & -0.14  & -0.17  & -0.16  & -0.15  & -0.19  & -0.18  & -0.17  \\ \hline \hline
\multirow{2}{*}{(0.85, 0.75)}                         & Uniform                                            & 0.75   & 0.80   & 0.85   & 0.37   & 0.40   & 0.43   & 0.31   & 0.33   & 0.35   & 0.25   & 0.27   & 0.28   & 0.24   & 0.25   & 0.27   \\
                                                      & Beta                                               & 0.12   & 0.15   & 0.18   & -0.16  & -0.14  & -0.13  & -0.15  & -0.14  & -0.13  & -0.18  & -0.17  & -0.16  & -0.16  & -0.15  & -0.14  \\ \hline
\multirow{2}{*}{(0.8, 0.4)}                           & Uniform                                            & 0.16   & 0.20   & 0.23   & 0.12   & 0.14   & 0.16   & 0.09   & 0.10   & 0.12   & 0.08   & 0.09   & 0.10   & 0.08   & 0.09   & 0.11   \\
                                                      & Beta                                               & 0.08   & 0.10   & 0.13   & -0.06  & -0.05  & -0.04  & -0.08  & -0.07  & -0.07  & -0.09  & -0.08  & -0.07  & -0.09  & -0.08  & -0.07  \\ \hline
\multirow{2}{*}{(0.85, 0.35)}                         & Uniform                                            & 0.20   & 0.23   & 0.26   & 0.11   & 0.12   & 0.14   & 0.11   & 0.12   & 0.14   & 0.08   & 0.09   & 0.10   & 0.07   & 0.08   & 0.09   \\
                                                      & Beta                                               & 0.08   & 0.11   & 0.14   & -0.04  & -0.02  & -0.01  & -0.08  & -0.07  & -0.06  & -0.10  & -0.09  & -0.08  & -0.11  & -0.10  & -0.09   \\ \bottomrule
\end{tabular} }
\caption{The relative error of ERM with respect to the Maximum Cost. For every row, we compute the relative error of the ERM for different $\vec q$, different probability distributions, while for every column we compute the relative error for different $n\in \{10,\dots,50\}$.}
\label{tab:max_RE}
\end{table}

\end{document}

%% file: mathinput.tex
\def \PP{\mathcal{P}}
\def \EE{\mathbb{E}}
\def \erre{\mathbb{R}}
\def \MCo{MC}
\def \SCo{SC}

\def \RM{\mathtt{ERM}^{(\pi,\vec v)}}

\newcommand{\floor}[1]{\left\lfloor #1 \right\rfloor}

\newtheorem{mechanism}{Mechanism}

\newtheorem{theorem}{Theorem}

\newtheorem{corollary}{Corollary}
\newtheorem{lemma}{Lemma}

\newtheorem{example}{Example}

\DeclareMathOperator*{\argmin}{arg\,min}

\def \PP{\mathcal{P}}

\def \erre{\mathbb{R}}

\def \enne{\mathbb{N}}